\documentclass[a4paper,onecolumn,11pt,accepted=2021-06-03]{quantumarticle}
\pdfoutput=1
\usepackage[utf8]{inputenc}
\usepackage[english]{babel}
\usepackage[T1]{fontenc}

\usepackage{amssymb,amsmath,amsthm,amsfonts,dsfont}
\usepackage[numbers,comma,sort&compress]{natbib}
\usepackage{tkz-graph}
\usepackage[colorlinks]{hyperref}
\usepackage[all]{hypcap}
\usepackage{url}
\usepackage{graphicx}
\usepackage[font=small]{caption}
\usepackage{subcaption}
\usepackage{float}
\usepackage{algorithmic}
\usepackage[ruled,vlined,linesnumbered]{algorithm2e}
\usepackage{footnote}
\usepackage{textcomp}
\usepackage{gensymb}
\usepackage[margin=1in]{geometry}

\usepackage{xcolor}
\usepackage{pgfplots}

\usepackage{mathtools}
\mathtoolsset{centercolon}
\DeclarePairedDelimiter\bra{\langle}{\rvert}
\DeclarePairedDelimiter\ket{\lvert}{\rangle}
\DeclarePairedDelimiterX\braket[2]{\langle}{\rangle}{#1 \delimsize\vert #2}

\newtheorem{theorem}{Theorem}

\newtheorem{lemma}{Lemma}
\newtheorem{proposition}{Proposition}

\newtheorem{corollary}{Corollary}

\newcommand{\eq}[1]{(\ref{eq:#1})}
\newcommand{\thm}[1]{\hyperref[thm:#1]{Theorem~\ref*{thm:#1}}}
\newcommand{\defn}[1]{\hyperref[defn:#1]{Definition~\ref*{defn:#1}}}
\newcommand{\lem}[1]{\hyperref[lem:#1]{Lemma~\ref*{lem:#1}}}
\newcommand{\prop}[1]{\hyperref[prop:#1]{Proposition~\ref*{prop:#1}}}
\newcommand{\fig}[1]{\hyperref[fig:#1]{Figure~\ref*{fig:#1}}}
\newcommand{\tab}[1]{\hyperref[tab:#1]{Table~\ref*{tab:#1}}}
\renewcommand{\sec}[1]{\hyperref[sec:#1]{Section~\ref*{sec:#1}}}
\newcommand{\append}[1]{\hyperref[append:#1]{Appendix~\ref*{append:#1}}}
\newcommand{\cor}[1]{\hyperref[cor:#1]{Corollary~\ref*{cor:#1}}}
\newcommand{\obs}[1]{\hyperref[obs:#1]{Observation~\ref*{obs:#1}}}

\DeclareMathOperator{\spn}{span}

\DeclareMathOperator*{\argmax}{argmax}

\definecolor{amethyst}{rgb}{0.6, 0.4, 0.8}

\newcommand{\comment}[1]{}

\newcommand{\abs}[1]{\left\lvert#1\right\rvert}
\newcommand{\norm}[1]{\left\lVert#1\right\rVert}
\newcommand{\vertiii}[1]{{\left\vert\kern-0.25ex\left\vert\kern-0.25ex\left\vert #1
		\right\vert\kern-0.25ex\right\vert\kern-0.25ex\right\vert}}

\pgfplotsset{
	log x ticks with fixed point/.style={
		xticklabel={
			\pgfkeys{/pgf/fpu=true}
			\pgfmathparse{exp(\tick)}%
			\pgfmathprintnumber[fixed relative, precision=3]{\pgfmathresult}
			\pgfkeys{/pgf/fpu=false}
		}
	},
	log y ticks with fixed point/.style={
		yticklabel={
			\pgfkeys{/pgf/fpu=true}
			\pgfmathparse{exp(\tick)}%
			\pgfmathprintnumber[fixed relative, precision=3]{\pgfmathresult}
			\pgfkeys{/pgf/fpu=false}
		}
	}
}

\newcommand{\OO}[1]{O\left(#1\right)}
\newcommand{\cO}[1]{\mathcal{O}\left(#1\right)}
\newcommand{\Om}[1]{\Omega\left(#1\right)}
\newcommand{\tildecO}[1]{\widetilde{\mathcal{O}}\left(#1\right)}
\usepackage{mathrsfs}

\usepackage{tabu}

\usepackage{cancel}

\DeclareFontFamily{U}{matha}{\hyphenchar\font45}
\DeclareFontShape{U}{matha}{m}{n}{
	<5> <6> <7> <8> <9> <10> gen * matha
	<10.95> matha10 <12> <14.4> <17.28> <20.74> <24.88> matha12
}{}
\DeclareSymbolFont{matha}{U}{matha}{m}{n}
\DeclareFontFamily{U}{mathx}{\hyphenchar\font45}
\DeclareFontShape{U}{mathx}{m}{n}{
	<5> <6> <7> <8> <9> <10>
	<10.95> <12> <14.4> <17.28> <20.74> <24.88>
	mathx10
}{}
\DeclareSymbolFont{mathx}{U}{mathx}{m}{n}
\DeclareMathSymbol{\obot}         {2}{matha}{"6B}
\DeclareMathSymbol{\bigobot}       {1}{mathx}{"CB}

\begin{document}
\title{Nearly tight Trotterization of interacting electrons}
\author[aff1,aff2]{Yuan Su}
\orcid{0000-0003-1144-3563}
\author[aff1,aff3]{Hsin-Yuan Huang}
\orcid{0000-0001-5317-2613}
\author[aff4]{Earl T.\ Campbell}
\affiliation[aff1]{Institute for Quantum Information and Matter, Caltech, Pasadena, CA 91125, USA}
\affiliation[aff2]{Google Research, Venice, CA 90291, USA}
\affiliation[aff3]{Department of Computing and Mathematical Sciences, Caltech, Pasadena, CA 91125, USA}
\affiliation[aff4]{AWS Center for Quantum Computing, Pasadena, CA 91125, USA}
\maketitle

\begin{abstract}
	We consider simulating quantum systems on digital quantum computers. We show that the performance of quantum simulation can be improved by simultaneously exploiting commutativity of the target Hamiltonian, sparsity of interactions, and prior knowledge of the initial state. We achieve this using Trotterization for a class of interacting electrons that encompasses various physical systems, including the plane-wave-basis electronic structure and the Fermi-Hubbard model. We estimate the simulation error by taking the transition amplitude of nested commutators of the Hamiltonian terms within the $\eta$-electron manifold. We develop multiple techniques for bounding the transition amplitude and expectation of general fermionic operators, which may be of independent interest. We show that it suffices to use $\left(\frac{n^{5/3}}{\eta^{2/3}}+n^{4/3}\eta^{2/3}\right)n^{o(1)}$ gates to simulate electronic structure in the plane-wave basis with $n$ spin orbitals and $\eta$ electrons, improving the best previous result in second quantization up to a negligible factor while outperforming the first-quantized simulation when $n=\eta^{2-o(1)}$. We also obtain an improvement for simulating the Fermi-Hubbard model. We construct concrete examples for which our bounds are almost saturated, giving a nearly tight Trotterization of interacting electrons.
\end{abstract}

\newpage
{
	\clearpage\tableofcontents
}
\newpage

\section{Introduction}
\label{sec:intro}

Simulating quantum systems to model their dynamics and energy spectra is one of the most promising applications of digital quantum computers. Indeed, the difficulty of performing such simulations on classical computers led Feynman \cite{Fey82} and others to propose the idea of quantum computation. In 1996, Lloyd proposed the first explicit quantum algorithm for simulating local Hamiltonians \cite{lloyd1996universal}. Since then, various quantum simulation algorithms have been developed \cite{berry2007efficient,FractionalQuery14,Berry15,BCK15,QSP17,Low2019hamiltonian,LW18,haah2018quantum,childs2018faster,campbell2019random,ouyang2020compilation,LKW19}, with potential applications in studying condensed matter physics \cite{childs2018toward}, quantum chemistry \cite{cao2019quantum,mcardle2020quantum}, quantum field theories \cite{JLP12,JLP14b}, superstring/M-theory \cite{GHHL20}, as well as in designing other classical \cite{KAA20} and quantum algorithms \cite{Berry14,BS17,CCDFGS03,HHL09,AL19,LT20,WK20,GSLW19,Lin2020optimalpolynomial,Rall21}.

Lloyd's original work considered the simulation of $k$-local Hamiltonians. This was subsequently extended to the study of $d$-sparse Hamiltonians \cite{AT03,berry2007efficient}, which provides a framework that abstracts the design of quantum algorithms from the underlying physical settings. However, despite their theoretical value, algorithms for sparse Hamiltonian simulation do not always provide the fastest approach for simulating concrete physical systems. Hamiltonians arising in practice often have additional features beyond sparseness, such as locality \cite{haah2018quantum,Tran18}, commutativity \cite{CS19,CSTWZ19,Somma16}, and symmetry \cite{TSCT20,HLMJH20}, that can be used to improve the performance of simulation. Besides, prior knowledge of the initial state \cite{Somma15,SS20,BWMMNC18,CBC20} and the norm distribution of Hamiltonian terms \cite{campbell2019random,HP18,CHKT20,meister2020tailoring,LHLAC20} have also been proven useful for digital quantum simulation.

We show that a number of these features, in particular the sparsity, commutativity, and initial-state information, can be combined to give an even faster simulation. We achieve this improvement for a class of interacting-electronic Hamiltonians, which includes many physically relevant systems such as the plane-wave-basis electronic-structure Hamiltonian and the Fermi-Hubbard model. Our approach uses Trotterization---a method widely applied in digital quantum simulation.

Our analysis proceeds by computing the transition amplitude of simulation error within the $\eta$-electron manifold. To this end, we develop multiple techniques for bounding the transition amplitude/expectation of a general fermionic operator, which may be of independent interest. For an $n$-spin-orbital electronic-structure problem in the plane-wave basis, our result improves the best previous result in second quantization \cite{LW18,BWMMNC18,CSTWZ19} up to a negligible factor while outperforming the first-quantized result \cite{Babbush2019} when $n=\eta^{2-o(1)}$. We also obtain an improvement for simulating the Fermi-Hubbard model. We construct concrete examples for which our bounds are almost saturated, giving a nearly tight Trotterization of interacting electrons.

\subsection{Combining interaction sparsity, commutativity, and initial-state knowledge}
\label{sec:intro_combine}
Sparsity can be used to improve digital quantum simulation in multiple ways. A common notion of $d$-sparsity concerns the target Hamiltonian itself, where each row and column of the Hamiltonian contains $d$ nonzero elements accessed by querying quantum oracles. As aforementioned, this provides an abstract framework for designing efficient simulation algorithms and is versatile for establishing lower bounds \cite{berry2007efficient}, although it sometimes ignores other important properties of the target system, such as locality, commutativity, and symmetry. Another notion of sparsity, closely related to our paper, considers the interactions between the underlying qubits or modes \cite{Berry2019qubitizationof,XSSS20,PHOW20,MBLA14}. The sparsity of interactions does not in general imply the underlying Hamiltonian is sparse, but it provides a tighter bound on the number of terms in the Hamiltonian and may thus be favorable to digital quantum simulation.

Trotterization (and its alternative variants \cite{CSTWZ19,childs2018faster,ouyang2020compilation,HP18,LKW19}) provides a simple approach to digital quantum simulation and is so far the only known approach that can exploit the commutativity of the Hamiltonian. Indeed, in the extreme case where all the terms in the Hamiltonian commute, we can simultaneously diagonalize them and apply the first-order Lie-Trotter formula $\mathscr{S}_1(t)$ without error. Previous studies have also established commutator error bounds for certain low-order formulas \cite{Suzuki85} and specific systems \cite{CS19,Somma16}. An analysis of a general formula $\mathscr{S}_p(t)$ is, however, considerably more difficult and has remained elusive until the recent proof of the commutator scaling of Trotter error \cite{CSTWZ19}.

A different direction to speeding up digital quantum simulation is to exploit information about the initial state. The error of digital quantum simulation is commonly quantified in previous work by the spectral-norm distance, which considers all possible states in the underlying Hilbert space.  But if the state is known to be within some subspace throughout the simulation, then in principle this knowledge could be used to improve the algorithm. For instance, digital quantum simulation in practice often starts with an initial state in the low-energy subspace of the Hamiltonian, so a worst-case spectral-norm analysis will inevitably overestimate the error. To address this, recent studies have considered a low-energy projection on the simulation error and provided improved approaches, using either Trotterization \cite{Somma15,SS20,BWMMNC18,CBC20} or more advanced quantum algorithms \cite{LC17}, that can be advantageous when the energy of the initial state is sufficiently small.

Ideally, the sparsity of interactions, commutativity of the Hamiltonian, and prior knowledge about the initial state can all be combined to yield an even faster digital quantum simulation. This combination, however, appears to be technically challenging to achieve. Indeed, the state-of-the-art analysis of Trotterization represents the simulation error in terms of nested commutators of Hamiltonian terms with exponential conjugations \cite[Theorem 5]{CSTWZ19}. This representation is versatile for computing the commutator scaling of Trotter error, but it yields little information about the energy of the initial state. To the best of our knowledge, the only previous attempt to address this problem was made by Somma for simulating bosonic Hamiltonians \cite{Somma15}, whose solution seems to have a divergence issue.\footnote{Recent work \cite{AFL21} developed a tighter analysis of Trotter error for time-dependent Hamiltonian simulation that exploits both commutativity of the target Hamiltonian and knowledge about the initial state.} Instead, we combine the sparsity, the commutativity and the initial-state information to give an improved simulation of a class of interacting electrons.

\subsection{Simulating interacting electrons}
\label{sec:intro_electron}
Simulating interacting electrons has emerged as one of the most important applications of digital quantum simulation \cite{bauer2020quantum}. Following pioneering work such as \cite{aspuru2005simulated,OGKL01}, recent developments of efficient quantum algorithms for electronic structure simulation have dramatically reduced the simulation cost through various techniques \cite{mcardle2020quantum,cao2019quantum,LBGHMWB20,BLHSRRT20,WBA11,SRL12,TL13,Berry2019qubitizationof,MYMLMBC13,BWMMNC18,RWSWT17}.

Here, we consider simulating the following class of interacting electrons by Trotterization:
\begin{equation}
\label{eq:second_quantized_ham}
H=T+V:=\sum_{j,k}\tau_{j,k}A_j^\dagger A_k+\sum_{l,m}\nu_{l,m}N_l N_m,
\end{equation}
where $A_j^\dagger$, $A_k$ are the fermionic \emph{creation} and \emph{annihilation operators}, $N_l$ are the \emph{occupation-number operators}, $\tau$, $\nu$ are coefficient matrices, and the summation is over $n$ spin orbitals. The specific definitions of these fermionic operators are given in \sec{prelim_quantization}. We say the interactions are \emph{$d$-sparse} if there are at most $d$ nonzero elements within each row/column of $\tau$ and $\nu$. This model represents various systems arising in physics and chemistry, including the electronic-structure Hamiltonians in the plane-wave basis \cite{BWMMNC18} and the Fermi-Hubbard model \cite{CBC20,kivlichan2020improved}.

To apply Trotterization, we need to express the Hamiltonian as a sum of elementary terms, each of which can be directly exponentiated on a quantum computer; see \sec{prelim_pf} for a review of this algorithm. For the electronic Hamiltonian \eq{second_quantized_ham}, it suffices to consider the two-term decomposition $H=T+V$, as the exponentials of $T$ and $V$ can be directly implemented using various quantum circuits. For instance, all the terms in $V$ commute with each other so $e^{-i t V}= \prod_{l,m}  e^{-i t \nu_{l,m}N_l N_m }$, where each $e^{-i t \nu_{l,m}N_l N_m }$ corresponds to a two-qubit operation under the Jordan-Wigner transformation. On the other hand, the exponential $e^{-i t T}$ can be implemented by diagonalization, i.e., $e^{-i t T}= U e^{-i \sum \lambda_\ell N_\ell} U^\dagger $, where $U$ can be efficiently implemented using Givens rotations~\cite{Pou15,kivlichan2018quantum}. In cases where $\tau$ and $\nu$ are translationally invariant $\tau_{j,k}=\tau_{j+q,k+q}$, $\nu_{l,m}=\nu_{l+q,m+q}$, we can implement $e^{-itT}$ using the fast fermionic Fourier transform~\cite{BWMMNC18} and a related circuit implementation exists for $e^{-itV}$ \cite{LW18}.

We now apply a $p$th-order Trotterization $\mathscr{S}_p(t)$ to approximate the evolution of the electronic Hamiltonian \eq{second_quantized_ham} for time $t$. We prove the following bound on the error of this approximation.
\begin{theorem}[Fermionic seminorm of Trotter error]
	\label{thm:fermionic_seminorm}
	Let $H=T+V=\sum_{j,k}\tau_{j,k}A_j^\dagger A_k+\sum_{l,m}\nu_{l,m}N_lN_m$ be an interacting-electronic Hamiltonian \eq{second_quantized_ham} with $n$ spin orbitals, which we simulate using a $p$th-order formula $\mathscr{S}_p(t)$. Then,
	\begin{equation}
	\label{eq:general_bound}
	\norm{\mathscr{S}_p(t)-e^{-itH}}_{\eta}
	=\cO{\left(\norm{\tau}+\norm{\nu}_{\max}\eta\right)^{p-1}
		\norm{\tau}\norm{\nu}_{\max}\eta^2 t^{p+1}}.
	\end{equation}
	Furthermore, if the interactions are $d$-sparse,
	\begin{equation}
	\label{eq:sparse_bound}
	\norm{\mathscr{S}_p(t)-e^{-itH}}_{\eta}
	=\cO{\left(\norm{\tau}_{\max}+\norm{\nu}_{\max}\right)^{p-1}
		\norm{\tau}_{\max}\norm{\nu}_{\max}d^{p+1}\eta t^{p+1}}.
	\end{equation}
	Here, $\norm{\cdot}$ is the \emph{spectral norm}, $\norm{\cdot}_{\max}$ is the \emph{max-norm} denoting the largest matrix element in absolute value, and
	\begin{equation}
	\label{eq:fermionic_seminorm_def}
		\norm{X}_{\eta}
		:=\max_{\ket{\psi_\eta},\ket{\phi_\eta}}\abs{\bra{\phi_\eta}X\ket{\psi_\eta}}
	\end{equation}
	is the \emph{fermionic $\eta$-seminorm} for \emph{number-preserving operator} $X$, where $\ket{\psi_\eta}$, $\ket{\phi_\eta}$ are quantum states with $\eta$ electrons.
\end{theorem}

This theorem follows from an inductive estimate of the fermionic seminorm of nested commutators of Hamiltonian terms, and will be formally proved in \sec{recurse} and \sec{path}. Note that in order to use the prior knowledge of the initial state, we have considered the fermionic seminorm $\norm{\cdot}_{\eta}$ of Trotter error with respect to the $\eta$-electron manifold. This seminorm is closely related to other metrics used to quantify the impact of initial-state information to digital quantum simulation \cite{Somma15,SS20,BWMMNC18,CBC20}; see \sec{prelim_seminorm} for a detailed discussion. The resulting bound depends on the number of electrons $\eta$, as well as the spectral norm $\norm{\tau}$, the max-norm $\norm{\tau}_{\max},\norm{\nu}_{\max}$, and the sparsity of interactions $d$, but there is no dependence on the total number of spin orbitals $n$. This improves over previous work \cite[Theorem F.5]{CSTWZ19} where an explicit $n$-scaling seems unavoidable. Meanwhile, other prior estimates of the fermionic seminorm \cite[Appendix G]{BWMMNC18} \cite[Theorem 13]{CBC20} did not exploit commutativity of the Hamiltonian and would introduce an additional factor of $\eta^{p}$ in the Trotter error bound. Our result thus improves the performance of digital quantum simulation by combining the initial-state information, the interaction sparsity, and commutativity of the Hamiltonian.

A common issue with the Trotterization algorithm is that existing analyses can be very loose for simulating certain physical systems. However, we address this with the following theorem, which shows that the asymptotic scaling of our bound is nearly tight.
\begin{theorem}[Tightness]
	\label{thm:fermionic_seminorm_tightness}
	For $s,w>0$ and positive integer $\eta\leq\frac{n}{2}$, there exists an interacting-electronic Hamiltonian $H=T+V=\sum_{j,k}\tau_{j,k}A_j^\dagger A_k+\sum_{l,m}\nu_{l,m}N_lN_m$ as in \eq{second_quantized_ham} with $n$ spin orbitals such that $\norm{\tau}=s$, $\norm{\nu}_{\max}=w$,
	\begin{equation}
	\label{eq:tightness_general}
		\big\Vert\underbrace{[T,\ldots[T}_{p},V]]\big\Vert_{\eta}=\Om{s^p w\eta},\qquad
		\big\Vert\underbrace{[V,\ldots[V}_{p},T]]\big\Vert_{\eta}=\Om{\left(w\eta\right)^{p} s/n}.
	\end{equation}
	In addition, for $u,w>0$ and positive integer $d\leq \eta\leq\frac{n}{2}$, there exists a $d$-sparse interacting-electronic Hamiltonian \eq{second_quantized_ham} with $n$ spin orbitals such that $\norm{\tau}_{\max}=u$, $\norm{\nu}_{\max}=w$,
	\begin{equation}
	\label{eq:tightness_sparse}
	\big\Vert\underbrace{[T,\ldots[T}_{p},V]]\big\Vert_{\eta}=\Om{\left(ud\right)^p wd},\qquad
	\big\Vert\underbrace{[V,\ldots[V}_{p},T]]\big\Vert_{\eta}=\Om{\left(wd\right)^p u}.
	\end{equation}
\end{theorem}

We prove the above theorem by choosing $T=\sum_{j,k=0}^{n-1}A_j^\dagger A_k$ and $V=\sum_{l,m=0}^{n/2-1}N_lN_m$ and computing their rescaled nested commutators, both in the original basis and the Fourier basis; see \sec{tight} for the proof. Note that both commutators $[T,\ldots[T,V]]$ and $[V,\ldots,[V,T]]$ contribute to the Trotter error, as well as other types of nested commutators which do not dominate the error scaling (\prop{trotter_error_comm_rep}). Modulo an application of the triangle inequality, \thm{fermionic_seminorm_tightness} then shows that our bound \eq{general_bound} overestimates the Trotter error by a factor of $n\eta$ in the worst case, whereas \eq{sparse_bound} overestimates a factor of at most $\eta$. For $p$ sufficiently large, this only contributes $n^{o(1)}$ and $\eta^{o(1)}$ to the gate complexity, respectively. In this sense, we have given a nearly tight Trotterization of interacting-electronic Hamiltonians \eq{second_quantized_ham}.\footnote{For specific molecules, our complexity estimates may be further improved by using additional features of the Hamiltonians; see \sec{discuss} for a further discussion of this point.}

\subsection{Main techniques}
\label{sec:intro_tech}
The proof of \thm{fermionic_seminorm} relies on multiple approaches we develop for bounding the fermionic seminorm, which may be of independent interest. Recall from \eq{fermionic_seminorm_def} that the fermionic seminorm $\norm{X}_\eta$ of a number-preserving operator $X$ is the maximum transition amplitude of $X$ within the $\eta$-electron manifold.

Our first approach is based on the observation that the fermionic seminorm of $X$ can be alternatively represented using the expectation of $X^\dagger X$, i.e.,
\begin{equation}
\norm{X}_{\eta}
=\max_{\ket{\psi_\eta},\ket{\phi_\eta}}\abs{\bra{\phi_\eta}X\ket{\psi_\eta}}
=\max_{\ket{\psi_\eta}}\sqrt{\bra{\psi_\eta}X^\dagger X\ket{\psi_\eta}}.
\end{equation}
We then upper bound $X^\dagger X$ in terms of the \emph{particle-number operator} $N=\sum_{j}N_j$, so that the expectation scales with the number of electrons $\eta=\bra{\psi_\eta}N\ket{\psi_\eta}$ instead of the total number of spin orbitals. Assuming $X$ is a sum of product of fermionic operators, we contract the summation indices in $X^\dagger X$ by using either diagonalization (\lem{diagonalization}) or an operator Cauchy-Schwarz inequality (\lem{cauchy}) \cite{Otte10}. To extend this argument to general fermionic operators, we prove a H\"{o}lder-type inequality (\lem{fermionic_holder}) and apply it recursively to bound $X^\dagger X$. We detail this recursive approach in \sec{recurse} and use it to prove \eq{general_bound}.

Our second approach starts by bounding the fermionic seminorm
\begin{equation}
	\norm{X}_{\eta}=\max_{\ket{\psi_\eta},\ket{\phi_\eta}}\abs{\bra{\phi_\eta}X\ket{\psi_\eta}}
\end{equation}
in terms of the maximum expectation value $\max_{\ket{\psi_\eta}}\abs{\bra{\psi_\eta}X\ket{\psi_\eta}}$. We then expand $X$ and $\ket{\psi_\eta}$ and give a combinatorial argument to count the number of ``paths'' which have nonzero contribution to the expectation (\prop{path_bound_seminorm}). We discuss this path-counting approach in more detail in \sec{path} and use it to prove \eq{sparse_bound}. It is worth mentioning that the path-counting technique can also be used to prove the following alternative bound for the $p$th-order Trotterization
\begin{equation}
\label{eq:general_path_bound}
	\norm{\mathscr{S}_p(t)-e^{-itH}}_{\eta}
	=\cO{\left(n\norm{\tau}_{\max}+\norm{\nu}_{\max}\eta\right)^{p-1}
		\norm{\tau}_{\max}\norm{\nu}_{\max}n\eta^2 t^{p+1}}.
\end{equation}
This is slightly weaker than \eq{general_bound} since $\norm{\tau}\leq n\norm{\tau}_{\max}$ always holds but not necessarily saturates, but in our application it yields the same gate complexity for the electronic-structure simulation in the second-quantized plane-wave basis. We discuss this further in \append{pathcountdense}.

Note that the expectation of fermionic operators, when taken with respect to the computational-basis states, can be exactly computed using the so-called Wick's theorem \cite{Wick50,peskin2019introduction}. However, this approach would introduce unnecessary term reordering which actually complicates our proof. In contrast, the underlying idea of path counting is conceptually simpler and may have potential applications in other contexts beyond the analysis of Trotter error. 

\subsection{Applications}
\label{sec:intro_app}
The nearly tight Trotterization of electronic Hamiltonian \eq{second_quantized_ham} gives improved simulations of many systems arising in condensed matter physics and quantum chemistry, including the plane-wave-basis electronic-structure Hamiltonian and the Fermi-Hubbard model.

The electronic-structure problem considers electrons interacting with each other and some fixed nuclei. An efficient simulation of such systems could help understand chemical reactions, and provide insight into material properties. Here, we consider representing the electronic-structure Hamiltonian in the plane-wave basis \cite{BWMMNC18}:
\begin{equation}
\label{eq:plane_wave_dual}
\begin{aligned}
H&=\frac{1}{2n}\sum_{j,k,\mu}\kappa_{\mu}^2\cos[\kappa_{\mu}\cdot r_{k-j}]A_{j}^\dagger A_{k}\\
&\quad-\frac{4\pi}{\omega}\sum_{j,\iota,\mu\neq 0}\frac{\zeta_\iota\cos[\kappa_{\mu}\cdot(\widetilde{r}_\iota-r_j)]}{\kappa_{\mu}^2}N_{j}
+\frac{2\pi}{\omega}\sum_{\substack{j\neq k\\\mu\neq 0}}\frac{\cos[\kappa_{\mu}\cdot r_{j-k}]}{\kappa_{\mu}^2}N_{j}N_{k},
\end{aligned}
\end{equation}
where $\omega$ is the volume of the computational cell, $\kappa_{\mu}=2\pi\mu/\omega^{1/3}$ are $n$ vectors of plane-wave frequencies, $\mu$ are three-dimensional vectors of integers with elements in the interval $[-n^{1/3},n^{1/3}]$, $r_j$ are the positions of electrons; $\zeta_\iota$ are nuclear charges; and $\widetilde{r}_\iota$ are the nuclear coordinates. We can represent this Hamiltonian in the form \eq{second_quantized_ham} with
\begin{equation}
\norm{\tau}=\cO{\frac{n^{2/3}}{\omega^{2/3}}},\qquad
\norm{\nu}_{\max}=\cO{\frac{n^{1/3}}{\omega^{1/3}}}.
\end{equation}
Assuming a constant system density $\eta=\cO{\omega}$, \thm{fermionic_seminorm} then implies that
\begin{equation}
\norm{\mathscr{S}_p(t)-e^{-itH}}_\eta
=\cO{\left(\frac{n^{2/3}}{\eta^{2/3}}+n^{1/3}\eta^{2/3}\right)^{p}n^{1/3}\eta^{2/3} t^{p+1}}.
\end{equation}
This approximation is accurate for a short-time evolution. To simulate for a longer time, we divide the evolution into $r$ steps and apply $\mathscr{S}_p(t/r)$ within each step, obtaining
\begin{equation}
\begin{aligned}
\norm{\mathscr{S}_p^r(t/r)-e^{-itH}}_\eta
=\cO{\left(\frac{n^{2/3}}{\eta^{2/3}}+n^{1/3}\eta^{2/3}\right)^{p}n^{1/3}\eta^{2/3} \frac{t^{p+1}}{r^p}}.
\end{aligned}
\end{equation}
Therefore,
\begin{equation}
r=\cO{\left(\frac{n^{2/3}}{\eta^{2/3}}+n^{1/3}\eta^{2/3}\right)\left(n^{2/3}\eta^{1/3}\right)^{1/p}}
\end{equation}
steps suffices to simulate for a constant time and accuracy with a $p$th-order Trotterization. Implementing each step using the approach of \cite[Sect.\ 5]{LW18} and choosing $p$ sufficiently large, we obtain the gate complexity
\begin{equation}
\left(\frac{n^{5/3}}{\eta^{2/3}}+n^{4/3}\eta^{2/3}\right)n^{o(1)}.
\end{equation}
Up to the negligible factor $n^{o(1)}$, this improves the best previous result in second quantization while outperforming the first-quantized simulation when $n=\eta^{2-o(1)}$. See \tab{result_summary} for a gate-count comparison. We discuss this in detail in \sec{app_plane_wave}.

\begin{table}
	\begin{center}
		\hspace*{-0.6cm}
		\begin{tabu}{c|cc}
			\hline
			Simulation Algorithm & $n,\eta$ & $\eta=\Theta(n)$\\
			\hline
			Interaction-picture (Ref.\ \cite{Babbush2019}, first quantization) & $\tildecO{n^{1/3}\eta^{8/3}}$ & $\tildecO{n^3}$\\
			Qubitization (Ref.\ \cite{Babbush2019}, first quantization) & $\tildecO{n^{2/3}\eta^{4/3}+n^{1/3}\eta^{8/3}}$ & $\tildecO{n^3}$\\
			Interaction-picture (Ref.\ \cite{LW18}, second quantization) & $\tildecO{\frac{n^{8/3}}{\eta^{2/3}}}$ & $\tildecO{n^2}$\\
			Trotterization (Ref.\ \cite{BWMMNC18}, second quantization) & $\left(n^{5/3}\eta^{1/3}+n^{4/3}\eta^{5/3}\right)n^{o(1)}$ & $n^{3+o(1)}$\\
			Trotterization (Ref.\ \cite{CSTWZ19}, second quantization) & $\left(\frac{n^{7/3}}{\eta^{1/3}}\right)n^{o(1)}$ & $n^{2+o(1)}$\\
			\hline
			Trotterization (\thm{fermionic_seminorm}, second quantization) & $\left(\frac{n^{5/3}}{\eta^{2/3}}+n^{4/3}\eta^{2/3}\right)n^{o(1)}$ & $n^{2+o(1)}$\\
			\hline
		\end{tabu}
	\end{center}
	\caption{Comparison of our result and previous results for simulating the plane-wave-basis electronic structure with $n$ spin orbitals and $\eta$ electrons. We use $\tildecO{\cdot}$ to suppress polylogarithmic factors in the gate complexity scaling.}
	\label{tab:result_summary}
\end{table}

We also consider applications to the Fermi-Hubbard model, which is believed to capture the physics of some high temperature superconductors. This model is classically challenging to simulate \cite{zheng2017stripe,HalfPercent}, but is a potential candidate for near-term quantum simulation \cite{Cai20,cade2019strategies,Linke18,GoogleHubbard20}. We have
\begin{equation}
\label{eq:fermi_hubbard}
H=-s \sum_{\langle j,k\rangle,\sigma}\left(A_{j,\sigma}^\dagger A_{k,\sigma}+A_{k,\sigma}^\dagger A_{j,\sigma}\right)
+v\sum_{j}N_{j,0}N_{j,1},
\end{equation}
where $\langle j,k\rangle$ denotes a summation over nearest-neighbor lattice sites and $\sigma\in\{0,1\}$ labels the spin degree of freedom. The Fermi-Hubbard model represents a lattice system with nearest-neighbor interactions and, according to \cite{CS19}, can be simulated with $\cO{n^{1+1/p}}$ gates using a $p$th-order Trotterization for a constant time and accuracy. On the other hand, recent work \cite{CBC20} shows that the Trotterization algorithm has gate complexity $\cO{n\eta^{1+1/p}}$ when restricted to the $\eta$-electron manifold. By simultaneously using the sparsity of interactions, commutativity of the Hamiltonian and information about the initial state, we show in \sec{app_hubbard} that $\cO{n\eta^{1/p}}$ gates suffices, improving both results for the Fermi-Hubbard model.\footnotemark
\footnotetext{Note however that this does not significantly improve the approach based on Lieb-Robinson bounds \cite{haah2018quantum}, since that approach has gate complexity $\tildecO{n}$ when using a high-precision quantum simulation algorithm as a subroutine.}

We conclude the paper in \sec{discuss} with a discussion of the results and some open questions.

\section{Preliminaries}
\label{sec:prelim}

In this section, we summarize preliminaries of this paper, including a discussion of the Trotterization algorithm and its error analysis in \sec{prelim_pf}, a brief summary of the second-quantization representation in \sec{prelim_quantization}, and an introduction to the fermionic seminorm and its properties in \sec{prelim_seminorm}.

\subsection{Trotterization and Trotter error}
\label{sec:prelim_pf}
The Trotterization algorithm approximates the evolution of a sum of Hamiltonian terms using exponentials of the individual terms. For the interacting-electronic Hamiltonian \eq{second_quantized_ham}, it suffices to consider a two-term decomposition $H=T+V$, as the exponentials of $T$ and $V$ can be directly implemented on a quantum computer. Then, the ideal evolution under $H$ for time $t$ is given by $e^{-itH}=e^{-it(T+V)}$, which can be approximated by a $p$th-order product formula $\mathscr{S}_p(t)$, such as the first-order Lie-Trotter formula
\begin{equation}
\label{eq:pf1}
	\mathscr{S}_1(t):=e^{-itT}e^{-itV}
\end{equation}
and ($2k$)th-order Suzuki formulas \cite{suzuki1990fractal}
\begin{equation}
\label{eq:pf2k}
\begin{aligned}
	\mathscr{S}_{2}(t)&:=e^{-i\frac{t}{2}V}e^{-itT}e^{-i\frac{t}{2}V},\\
	\mathscr{S}_{2k}(t)&:=\mathscr{S}_{2k-2}(u_{k}t)^2 \, \mathscr{S}_{2k-2}((1-4u_{k})t) \, \mathscr{S}_{2k-2}(u_{k}t)^2,
\end{aligned}
\end{equation}
where $u_{k}:=1/(4-4^{1/(2k-1)})$. This approximation is accurate when $t$ is small. To simulate for a longer time, we divide the evolution into $r$ \emph{Trotter steps} and apply $\mathscr{S}_p(t/r)$ with \emph{Trotter error} at most $\epsilon/r$. We choose $r$ sufficiently large so that the simulation error, as quantified by the spectral norm $\norm{\mathscr{S}_p^r(t/r)-e^{-itH}}$, is at most $\epsilon$.

Trotterization (and its alternative variants) provides a simple approach to digital quantum simulation and is so far the only known approach that can exploit commutativity of the Hamiltonian. Indeed, in the extreme case where all the Hamiltonian terms commute, Trotterization can implement the exact evolution without error. Previous studies have also established commutator analysis of Trotter error for systems with geometrical locality \cite{CS19} and Lie-algebraic structures \cite{Somma16}, as well as certain low-order formulas \cite{Suzuki85}, including the first-order Lie-Trotter formula
\begin{equation}
\label{eq:pf1_error}
	\mathscr{S}_1(t)-e^{-itH}=
	\int_{0}^{t}\mathrm{d}\tau_1\int_{0}^{\tau_1}\mathrm{d}\tau_2\
	e^{-i(t-\tau_1)H}e^{-i\tau_1T}e^{i\tau_2T}\left[iT,iV\right]e^{-i\tau_2T}e^{-i\tau_1V}
\end{equation}
and the second-order Suzuki formula
\begin{equation}
\label{eq:pf2_error}
\begin{aligned}
\mathscr{S}_2(t)-e^{-itH}&=
\int_{0}^{t}\mathrm{d}\tau_1\int_{0}^{\tau_1}\mathrm{d}\tau_2\int_{0}^{\tau_2}\mathrm{d}\tau_3\
e^{-i(t-\tau_1)H}e^{-i\frac{\tau_1}{2}V}\\
&\quad\cdot\left(e^{-i\tau_3T}\left[-iT,\left[-iT,-i\frac{V}{2}\right]\right]e^{i\tau_3T}
+e^{i\frac{\tau_3}{2}V}\left[i\frac{V}{2},\left[i\frac{V}{2},iT\right]\right]e^{-i\frac{\tau_3}{2}V}\right)
e^{-i\tau_1T}e^{-i\frac{\tau_1}{2}V}.
\end{aligned}
\end{equation}
An analysis of the general case is, however, considerably more difficult and has remained elusive until the recent proof of commutator scaling of Trotter error \cite{CSTWZ19}. Here, we introduce a stronger version of that result which can be proved as in \cite[Appendix C]{CSTWZ19arXiv} by combining Theorem 3, 4, and 5 of \cite{CSTWZ19} without invoking the triangle inequality.\footnote{It is essential to invoke a representation of the Trotter error where nested commutators have at most a constant number of layers \cite{CSTWZ19}. If such a representation were not used, Trotterization would have a worse asymptotic complexity \cite[Appendix B]{WBCHT14}.}

\begin{proposition}[Commutator representation of Trotter error]
	\label{prop:trotter_error_comm_rep}
	Let $H=T+V$ be a two-term Hamiltonian and $\mathscr{S}_p(t)$ be a $p$th-order formula. Define $H_0=V$ and $H_1=T$. Then,
	\begin{equation}
	\begin{aligned}
		\mathscr{S}_p(t)-e^{-itH}
		&=\int_{0}^{t}\mathrm{d}\tau_1\int_{0}^{\tau_1}\mathrm{d}\tau_2\sum_{\pmb{\gamma},j}a_{\pmb{\gamma},j}(\tau_1,\tau_2)e^{-i(t-\tau_1)H}\\
		&\qquad\qquad\qquad\qquad\quad
		\cdot \mathscr{U}_{\pmb{\gamma},j}(\tau_1,\tau_2)
		\left[H_{\gamma_{p+1}},\cdots\left[H_{\gamma_2},H_{\gamma_1}\right]\right]
		\mathscr{W}_{\pmb{\gamma},j}(\tau_1,\tau_2),
	\end{aligned}
	\end{equation}
	where $\pmb{\gamma}\in\{0,1\}^{p+1}$ are binary vectors\footnotemark\ and $j$ goes through a constant range of numbers (depending on the order $p$). Here, $\mathscr{U}_{\pmb{\gamma},j}(\tau_1,\tau_2)$ and $\mathscr{W}_{\pmb{\gamma},j}(\tau_1,\tau_2)$ are products of evolutions of $T$ and $V$ with time variables $\tau_1$ and $\tau_2$ and $a_{\pmb{\gamma}}(\tau_1,\tau_2)$ are coefficients such that
	\footnotetext{We use bold symbols to represent vectors throughout this paper.}
	\begin{equation}
		\int_{0}^{t}\mathrm{d}\tau_1\int_{0}^{\tau_1}\mathrm{d}\tau_2\
		\abs{a_{\pmb{\gamma},j}(\tau_1,\tau_2)}
		=\cO{t^{p+1}}.
	\end{equation}
\end{proposition}

As an immediate application, we find that the spectral norm of the Trotter error scales with nested commutators of the Hamiltonian terms, i.e.,
\begin{equation}
	\norm{\mathscr{S}_p(t)-e^{-itH}}=\cO{\max_{\pmb{\gamma}}\norm{\left[H_{\gamma_{p+1}},\cdots\left[H_{\gamma_2},H_{\gamma_1}\right]\right]}t^{p+1}}.
\end{equation}
Note that the use of $\max_{\pmb{\gamma}}$ in place of $\sum_{\pmb{\gamma}}$ does not change the scaling as $\pmb{\gamma}$ only ranges over a constant number of binary vectors. We then divide the evolution into $r$ steps and apply the triangle inequality to obtain
\begin{equation}
	\norm{\mathscr{S}_p^r(t/r)-e^{-itH}}
	\leq r\norm{\mathscr{S}_p(t/r)-e^{-i\frac{t}{r}H}}
	=\cO{\max_{\pmb{\gamma}}\norm{\left[H_{\gamma_{p+1}},\cdots\left[H_{\gamma_2},H_{\gamma_1}\right]\right]}\frac{t^{p+1}}{r^p}}.
\end{equation}
It thus suffices to choose
\begin{equation}
	r=\cO{\frac{\left(\max_{\pmb{\gamma}}\norm{\left[H_{\gamma_{p+1}},\cdots\left[H_{\gamma_2},H_{\gamma_1}\right]\right]}\right)^{1/p}t^{1+1/p}}{\epsilon^{1/p}}}
\end{equation}
to ensure that the error of simulation is no more than $\epsilon$.

The above analysis is versatile for computing the commutator dependence of Trotter error. Unfortunately, the resulting bound does not use prior knowledge of the initial state and will in particular be loose if the initial state lies within a low-energy subspace. On the other hand, recent work of \c{S}ahino\u{g}lu and Somma proposed a Trotterization approach for simulating low-energy initial states but the commutativity of the Hamiltonian was ignored in their analysis \cite{SS20}. Here, we address this by simultaneously using commutativity of the Hamiltonian and prior knowledge of the initial state to improve the simulation of a class of interacting electrons. We obtain further improvement when the electronic Hamiltonian has sparse interactions. In the following, we introduce preliminaries about the second-quantization representation (\sec{prelim_quantization}) and the notion of fermionic seminorm (\sec{prelim_seminorm}), on which our analysis will be based.

\subsection{Second-quantization representation}
\label{sec:prelim_quantization}
In this section, we review several facts about the second-quantization representation that are relevant to our analysis. We refer the reader to the book of Helgaker, J\o{}rgensen, and Olsen \cite{helgaker2014molecular} for a detailed discussion of this topic.

We use the abstract Fock space to represent electronic Hamiltonians. Specifically, for a system of $n$ spin orbitals, we construct a $2^n$-dimensional space $\spn\{\ket{c_0,c_1,\ldots, c_{n-1}}\}$ spanned by the basis vectors $\ket{c_0,c_1,\ldots, c_{n-1}}$, where $c_j=1$ represents that mode $j$ is occupied and $c_j=0$ otherwise. General vectors in the Fock space, denoted by $\ket{\psi}$ or $\ket{\phi}$, are then given by linear combinations of these orthonormal basis vectors. We define the \emph{$\eta$-electron subspace} $\spn\{\ket{c_0,c_1,\ldots, c_{n-1}},\sum_jc_j=\eta\}$. By considering all $0\leq\eta\leq n$, we obtain the decomposition
\begin{equation}
	\spn\left\{\ket{c_0,c_1,\ldots, c_{n-1}}\right\}
	=\bigobot_{\eta=0}^{n}\spn\Big\{\ket{c_0,c_1,\ldots, c_{n-1}},\sum_jc_j=\eta\Big\},
\end{equation}
where $\obot$ denotes the orthogonal direct sum. Using bold symbol $\pmb c$ to represent an arbitrary \emph{fermionic configuration} and $\abs{\pmb c}=\sum_jc_j$ to denote the \emph{Hamming weight}, we have
\begin{equation}
\spn\left\{\ket{\pmb c}\right\}
=\bigobot_{\eta=0}^{n}\spn\Big\{\ket{\pmb c},\abs{\pmb c}=\eta\Big\}.
\end{equation}
We say that normalized vectors in the $\eta$-electron subspace form the \emph{$\eta$-electron manifold} and denote an arbitrary such vector by $\ket{\psi_\eta}$ or $\ket{\phi_\eta}$.

The $n$ elementary fermionic \emph{creation operators} are defined through the relations
\begin{equation} \label{eq:createrule}
\begin{aligned}
	A_j^\dagger\ket{c_0,c_1,\ldots,0_j,\ldots, c_{n-1}}&=(-1)^{\sum_{k=0}^{j-1}c_k}\ket{c_0,c_1,\ldots,1_j,\ldots, c_{n-1}},\\
	A_j^\dagger\ket{c_0,c_1,\ldots,1_j,\ldots, c_{n-1}}&=0,\\
\end{aligned}
\end{equation}
whereas the fermionic \emph{annihilation operators} are defined by
\begin{equation} \label{eq:annihirule}
\begin{aligned}
	A_j\ket{c_0,c_1,\ldots,0_j,\ldots, c_{n-1}}&=0,\\
	A_j\ket{c_0,c_1,\ldots,1_j,\ldots, c_{n-1}}&=(-1)^{\sum_{k=0}^{j-1}c_k}\ket{c_0,c_1,\ldots,0_j,\ldots, c_{n-1}}.\\
\end{aligned}
\end{equation}
The use of $\dagger$ is justified by the fact that $A_j^\dagger$ is indeed the Hermitian adjoint of $A_j$ with respect to the inner product in the Fock space. We also introduce the \emph{occupation-number operators} $N_j=A_j^\dagger A_j$ and add them together to get the \emph{particle-number operator} $N=\sum_{j=0}^{n-1}N_j$.

Fermionic creation and annihilation operators satisfy the canonical anticommutation relations
\begin{equation}
	A_j^\dagger A_k^\dagger +A_k^\dagger A_j^\dagger=A_j A_k +A_k A_j=0,\qquad
	A_j^\dagger A_k + A_k A_j^\dagger = \delta_{j,k} I,
\end{equation}
where the Kronecker-delta function $\delta_{j,k}$ is one if $j=k$ and zero otherwise. Applying these, we obtain the following commutation relations of second-quantized fermionic operators.
\begin{proposition}[Commutation relations of fermionic operators]
	\label{prop:fermionic_commutation}
	The following commutation relations hold for second-quantized fermionic operators:
	\begin{enumerate}
		\item $\left[A_l^\dagger A_m,A_j^\dagger\right]=\delta_{j,m}A_l^\dagger$, $\left[A_l^\dagger A_m,A_k\right]=-\delta_{k,l}A_m$;
		\item $\left[N_l,A_j^\dagger\right]=\delta_{l,j}A_j^\dagger$, $\left[N_l,A_k\right]=-\delta_{l,k}A_k$;
		\item $\left[N,A_j^\dagger\right]=A_j^\dagger$, $\left[N,A_k\right]=-A_k$;
		\item $\left[N_l,N_m\right]=0$.
	\end{enumerate}
\end{proposition}

We say a fermionic operator is \emph{number-preserving} if every $\eta$-electron subspace is invariant under the action of this operator. Equivalently, operator $X$ is number-preserving if and only if it commutes with the particle-number operator, i.e., $\left[N,X\right]=0$. Yet another equivalent definition is based on the notion of \emph{$\eta$-electron projections}: letting $\Pi_\eta$ be orthogonal projections onto the $\eta$-electron subspaces, then $X$ is number-preserving if and only if it commutes with every $\Pi_\eta$, namely, $\left[\Pi_\eta,X\right]=0$. In the matrix representation, $X$ is block-diagonalized by the set of $\eta$-electron projections $\{\Pi_{\eta}\}$.

A special example of number-preserving operator is the particle-number operator $N$, which acts as a scalar multiplication by $\eta$ within the $\eta$-electron subspace. Other examples include \emph{excitation operators} $A_j^\dagger A_k$, occupation-number operators $N_l$, and elementary exponentials in the Trotterization algorithm $e^{-it\sum_{j,k}\tau_{j,k}A_j^\dagger A_k}$ and $e^{-it\sum_{l,m}\nu_{l,m}N_l N_m}$. The fermionic Fourier transform \cite{Ferris14} as given below is also number-preserving:\footnotemark
\footnotetext{This can alternatively be proved using the fact that the fermionic Fourier transform is generated by the fermionic swap and Hadamard gate \cite[Appendix I]{BWMMNC18}, both of which are number-preserving.}
\begin{equation}
\mathrm{FFFT}^\dagger\cdot A_j^\dagger\cdot \mathrm{FFFT}=\frac{1}{\sqrt{n}}\sum_{l=0}^{n-1} e^{-\frac{2\pi i jl}{n}}A_l^\dagger,\qquad
\mathrm{FFFT}^\dagger\cdot A_k\cdot \mathrm{FFFT}=\frac{1}{\sqrt{n}}\sum_{m=0}^{n-1} e^{\frac{2\pi i km}{n}}A_m,
\end{equation}
since
\begin{equation}
\begin{aligned}
	\mathrm{FFFT}^\dagger\cdot N\cdot \mathrm{FFFT}
	&=\sum_{j=0}^{n-1} \mathrm{FFFT}^\dagger\cdot A_j^\dagger \cdot \mathrm{FFFT}\cdot \mathrm{FFFT}^\dagger\cdot A_j \cdot \mathrm{FFFT}\\
	&=\frac{1}{n}\sum_{l,m=0}^{n-1} \left(\sum_{j=0}^{n-1}e^{\frac{2\pi ij(m-l)}{n}}\right)A_l^\dagger A_m
	=\sum_{l=0}^{n-1}A_l^\dagger A_l = N.
\end{aligned}
\end{equation}
In fact, the set of number-preserving operators contains identity and is closed under linear combination, multiplication, Hermitian conjugation, and taking limit.
\begin{proposition}[Number-preserving operators as a closed unital $\dagger$-subalgebra]
	The following operators are respectively number-preserving:
	\begin{enumerate}
		\item the identity operator $I$;
		\item $\lambda X+\mu Y$, if $X$ and $Y$ are number-preserving, and $\lambda$ and $\mu$ are complex numbers;
		\item $XY$, if $X$ and $Y$ are number-preserving;
		\item $X^\dagger$, if $X$ is number-preserving;
		\item $\lim\limits_{i\rightarrow\infty}X_i$, if $X_i$ are number-preserving and the limit exists.
	\end{enumerate}
\end{proposition}

\subsection{Fermionic seminorm}
\label{sec:prelim_seminorm}
We now introduce the notion of fermionic seminorm, which we use to quantify the error of the Trotterization algorithm that takes the prior knowledge of the initial state into consideration.

For any number-preserving operator $X$ and $0\leq\eta\leq n$, we define the \emph{fermionic $\eta$-seminorm} as the maximum transition amplitude within the $\eta$-electron manifold:
\begin{equation}
\label{eq:fermionic_seminorm_def_prelim}
	\norm{X}_{\eta}
	:=\max_{\ket{\psi_\eta},\ket{\phi_\eta}}\abs{\bra{\phi_\eta}X\ket{\psi_\eta}},
\end{equation}
where $\ket{\psi_\eta},\ket{\phi_\eta}$ are quantum states containing $\eta$ electrons.\footnotemark\ When there is no ambiguity, we drop the dependence on $\eta$ and call $\norm{X}_\eta$ the fermionic seminorm of $X$. As the name suggests and the following proposition confirms, the fermionic seminorm is indeed a seminorm defined on the closed unital $\dagger$-subalgebra of number-preserving operators.
\footnotetext{Note that it is possible to extend this to define $\norm{\cdot}_{\eta\rightarrow\xi}$ for operators that map the $\eta$-electron subspace to $\xi$-electron subspace, although this is not needed in our analysis and will not be further pursued here.}
\begin{proposition}[Seminorm properties]
	\label{prop:fermionic_seminorm_property}
	The following properties hold for the fermionic seminorm:
	\begin{enumerate}
		\item $\norm{\lambda X}_{\eta}=\abs{\lambda}\norm{X}_{\eta}$, if $X$ is number-preserving and $\lambda$ is a complex number;
		\item $\norm{X+ Y}_{\eta}\leq\norm{X}_{\eta}+\norm{Y}_{\eta}$, if $X$ and $Y$ are number-preserving;
		\item $\norm{XY}_\eta\leq\norm{X}_\eta\norm{Y}_\eta$, if $X$ and $Y$ are number-preserving;
		\item $\norm{I}_{\eta}=1$;
		\item $\norm{UXW}_{\eta}=\norm{X}_{\eta}$, if $U,X,W$ are number-preserving and $U,W$ are unitary;
		\item $\norm{X^\dagger}_{\eta}=\norm{X}_{\eta}$, if $X$ is number-preserving.
	\end{enumerate}
\end{proposition}
\begin{proof}
	We will only prove the third statement, as the remaining follow directly from the definition of the fermionic seminorm.
	We consider
	\begin{equation}
	\begin{aligned}
		\norm{XY}_\eta
		&=\max_{\ket{\psi_\eta},\ket{\phi_\eta}}\abs{\bra{\phi_\eta}XY\ket{\psi_\eta}}\\
		&=\max_{\ket{\psi_\eta},\ket{\phi_\eta}}\abs{\bra{\phi_\eta}X\Pi_{\eta}\Pi_{\eta}Y\ket{\psi_\eta}}\\
		&\leq\max_{\ket{\phi_\eta}}\norm{\Pi_{\eta}X^\dagger\ket{\phi_\eta}}\max_{\ket{\psi_\eta}}\norm{\Pi_{\eta}Y\ket{\psi_\eta}},
	\end{aligned}
	\end{equation}
	where $\Pi_\eta$ is the orthogonal projection onto the $\eta$-electron subspace and the last step follows from the Cauchy-Schwarz inequality. To proceed, we optimize over an arbitrary state $\ket{\varphi}$ to get
	\begin{equation}
	\begin{aligned}
		\norm{\Pi_{\eta}X^\dagger\ket{\phi_\eta}}
		&=\max_{\ket{\varphi}}\abs{\bra{\varphi}\Pi_{\eta}X^\dagger\ket{\phi_\eta}}\\
		&=\max_{\ket{\varphi}}\norm{\Pi_{\eta}\ket{\varphi}}
		\abs{\frac{\bra{\varphi}\Pi_{\eta}}{\norm{\Pi_{\eta}\ket{\varphi}}}X^\dagger\ket{\phi_\eta}}\\
		&\leq\norm{X^\dagger}_{\eta}=\norm{X}_{\eta}
	\end{aligned}
	\end{equation}
	assuming $\Pi_{\eta}\ket{\varphi}\neq0$, as the case $\Pi_{\eta}\ket{\varphi}=0$ never leads to maximality. But on the other hand,
	\begin{equation}
	\begin{aligned}
		\norm{X}_{\eta}
		&=\norm{X^\dagger}_{\eta}
		=\max_{\ket{\phi_\eta},\ket{\varphi_\eta}}\abs{\bra{\varphi_\eta}X^\dagger\ket{\psi_\eta}}\\
		&=\max_{\ket{\phi_\eta},\ket{\varphi_\eta}}\abs{\bra{\varphi_\eta}\Pi_{\eta}X^\dagger\ket{\psi_\eta}}\\
		&\leq\max_{\ket{\phi_\eta},\ket{\varphi}}\abs{\bra{\varphi}\Pi_{\eta}X^\dagger\ket{\psi_\eta}}
		=\max_{\ket{\phi_\eta}}\norm{\Pi_{\eta}X^\dagger\ket{\phi_\eta}},
	\end{aligned}
	\end{equation}
	implying $\max_{\ket{\phi_\eta}}\norm{\Pi_{\eta}X^\dagger\ket{\phi_\eta}}=\norm{X}_{\eta}$. Similarly, we have $\max_{\ket{\psi_\eta}}\norm{\Pi_{\eta}Y\ket{\psi_\eta}}=\norm{Y}_{\eta}$. This completes the proof of the third statement.
\end{proof}

The fermionic seminorm, as defined in \eq{fermionic_seminorm_def_prelim} by the maximum transition amplitude within the $\eta$-electron manifold, provides a reasonable metric for quantifying the error of digital quantum simulation with initial-state constraints. Indeed, a seminorm similar to our definition was used by Somma \cite{Somma15} for analyzing quantum simulation of bosonic Hamiltonians. However, we point out that this is not the only error metric that takes the prior knowledge of the initial state into account. Recent work \cite{SS20} analyzed the low-energy simulation of $k$-local frustration-free Hamiltonians by computing the spectral norm of Trotter error projected on the low-energy subspace. However, the following proposition shows that these two error metrics are the same for fermionic systems.
\begin{proposition}[Fermionic seminorm as a projected spectral norm]
	For any number-preserving operator $X$, it holds that
	\begin{equation}
		\norm{X}_\eta=\max_{\ket{\psi_\eta},\ket{\phi_\eta}}\abs{\bra{\phi_\eta}X\ket{\psi_\eta}}
		=\norm{X\Pi_{\eta}}.
	\end{equation}
\end{proposition}
\begin{proof}
	The underlying idea behind this proposition is already hinted in the proof of \prop{fermionic_seminorm_property}. We have
	\begin{equation}
	\begin{aligned}
		\max_{\ket{\psi_\eta},\ket{\phi_\eta}}\abs{\bra{\phi_\eta}X\ket{\psi_\eta}}
		&=\max_{\ket{\psi_\eta},\ket{\phi_\eta}}\abs{\bra{\phi_\eta}\Pi_{\eta}X\Pi_{\eta}\ket{\psi_\eta}}\\
		&\leq\max_{\ket{\psi},\ket{\phi}}\abs{\bra{\phi}\Pi_{\eta}X\Pi_{\eta}\ket{\psi}}
		=\norm{\Pi_{\eta}X\Pi_{\eta}}.
	\end{aligned}
	\end{equation}
	But on the other hand,
	\begin{equation}
	\begin{aligned}
		\norm{\Pi_{\eta}X\Pi_{\eta}}
		&=\max_{\ket{\psi},\ket{\phi}}\abs{\bra{\phi}\Pi_{\eta}X\Pi_{\eta}\ket{\psi}}\\
		&=\max_{\ket{\psi},\ket{\phi}}\norm{\Pi_{\eta}\ket{\phi}}\norm{\Pi_{\eta}\ket{\psi}}
		\abs{\frac{\bra{\phi}\Pi_{\eta}}{\norm{\Pi_{\eta}\ket{\phi}}}X\frac{\Pi_{\eta}\ket{\psi}}{\norm{\Pi_{\eta}\ket{\psi}}}}\\
		&\leq\max_{\ket{\psi_\eta},\ket{\phi_\eta}}\abs{\bra{\phi_\eta}X\ket{\psi_\eta}}
	\end{aligned}
	\end{equation}
	assuming $\Pi_{\eta}\ket{\phi}\neq0$ and $\Pi_{\eta}\ket{\psi}\neq0$, as the zero vector will not lead to maximality. The proposition then follows since number-preserving operator $X$ commutes with the $\eta$-electron projection $\Pi_{\eta}$.
\end{proof}

Another common approach to quantify the error of digital quantum simulation is to take the maximum expectation $\max_{\ket{\psi_\eta}}\abs{\bra{\psi_\eta}\cdot\ket{\psi_\eta}}$ within the $\eta$-electron manifold. This approach is used by previous work \cite{BWMMNC18,Pou15,CBC20} and appears to give a natural metric when digital quantum simulation is used as a subroutine in phase estimation. We show that this only differs from our definition \eq{fermionic_seminorm_def_prelim} by at most a constant factor, reaffirming the fermionic seminorm as a proper error metric for simulating fermionic systems.
\begin{proposition}[Transition amplitude and expectation]
	\label{prop:transition_expectation}
	For any number-preserving operator $X$, the following statements hold:
	\begin{enumerate}
		\item $\max_{\ket{\psi_\eta},\ket{\phi_\eta}}\abs{\bra{\phi_\eta}X\ket{\psi_\eta}}
		=\max_{\ket{\psi_\eta}}\abs{\bra{\psi_\eta}X\ket{\psi_\eta}}$, if $X$ is Hermitian;
		\item $\max_{\ket{\psi_\eta},\ket{\phi_\eta}}\abs{\bra{\phi_\eta}X\ket{\psi_\eta}}
		=\max_{\ket{\psi_\eta}}\sqrt{\bra{\psi_\eta}X^\dagger X\ket{\psi_\eta}}$ \big(equivalently, $\norm{X^\dagger X}_\eta=\norm{X}_\eta^2$\big);
		\item $\max_{\ket{\psi_\eta}}\abs{\bra{\psi_\eta}X\ket{\psi_\eta}}
		\leq\max_{\ket{\psi_\eta},\ket{\phi_\eta}}\abs{\bra{\phi_\eta}X\ket{\psi_\eta}}
		\leq2\max_{\ket{\psi_\eta}}\abs{\bra{\psi_\eta}X\ket{\psi_\eta}}$.\footnotemark
	\end{enumerate}
\end{proposition}
\footnotetext{Note that the second inequality is tight by considering $X=A_0^\dagger A_1$ on a fermionic system with two spin orbitals and one electron.}
\begin{proof}
	The first statement follows from the fact that $\Pi_{\eta}X\Pi_{\eta}$ is Hermitian and that the spectral norm of a Hermitian operator is its largest eigenvalue in absolute value. For the second statement,
	\begin{equation}
		\max_{\ket{\psi_\eta},\ket{\phi_\eta}}\abs{\bra{\phi_\eta}X\ket{\psi_\eta}}
		=\norm{X}_{\eta}
		=\norm{X\Pi_\eta}
		=\sqrt{\norm{\Pi_\eta X^\dagger X\Pi_\eta}}
		=\max_{\ket{\psi_\eta}}\sqrt{\bra{\psi_\eta}X^\dagger X\ket{\psi_\eta}}.
	\end{equation}
	The first inequality of Statement 3 is trivial. For the second inequality, we apply the polarization identity
	\begin{equation}
	\begin{aligned}
		\bra{\phi_\eta}X\ket{\psi_\eta}
		&=\frac{1}{4}\Big(\left(\bra{\phi_\eta}+\bra{\psi_\eta}\right)X\left(\ket{\phi_\eta}+\ket{\psi_\eta}\right)
		-\left(\bra{\phi_\eta}-\bra{\psi_\eta}\right)X\left(\ket{\phi_\eta}-\ket{\psi_\eta}\right)\\
		&\qquad\quad-i\left(\bra{\phi_\eta}-i\bra{\psi_\eta}\right)X\left(\ket{\phi_\eta}+i\ket{\psi_\eta}\right)
		+i\left(\bra{\phi_\eta}+i\bra{\psi_\eta}\right)X\left(\ket{\phi_\eta}-i\ket{\psi_\eta}\right)\Big)
	\end{aligned}
	\end{equation}
	to obtain
	\begin{equation}
	\begin{aligned}
		&\abs{\bra{\phi_\eta}X\ket{\psi_\eta}}\\
		&\leq\frac{\max_{\ket{\varphi_\eta}}\abs{\bra{\varphi_\eta}X\ket{\varphi_\eta}}}{4}\left(\norm{\ket{\phi_\eta}+\ket{\psi_\eta}}^2
		+\norm{\ket{\phi_\eta}-\ket{\psi_\eta}}^2
		+\norm{\ket{\phi_\eta}+i\ket{\psi_\eta}}^2
		+\norm{\ket{\phi_\eta}-i\ket{\psi_\eta}}^2\right)\\
		&=2\max_{\ket{\varphi_\eta}}\abs{\bra{\varphi_\eta}X\ket{\varphi_\eta}},
	\end{aligned}
	\end{equation}
	from which the claimed inequality follows by maximizing over states $\ket{\psi_\eta}$ and $\ket{\phi_\eta}$.
\end{proof}

We now apply \prop{trotter_error_comm_rep} to compute the fermionic seminorm of Trotter error, obtaining
\begin{equation}
\label{eq:pf2k_eta}
\norm{\mathscr{S}_p(t)-e^{-itH}}_{\eta}=\cO{\max_{\pmb{\gamma}}\norm{\left[H_{\gamma_{p+1}},\cdots\left[H_{\gamma_2},H_{\gamma_1}\right]\right]}_{\eta}t^{p+1}}.
\end{equation}
We find that the resulting error bound depends on the fermionic seminorm of nested commutators, and the performance of digital quantum simulation can thus be potentially improved by simultaneously exploiting commutativity of the Hamiltonian and prior knowledge of the initial state. However, the main difficulty here is to give a tight estimate of $\norm{\left[H_{\gamma_{p+1}},\cdots\left[H_{\gamma_2},H_{\gamma_1}\right]\right]}_{\eta}$, which seems technically challenging to address. To this end, we develop two approaches for bounding the expectation/transition amplitude of general fermionic operators in \sec{recurse} and \sec{path} and prove our main result \thm{fermionic_seminorm}, establish the tightness of our bound in \sec{tight}, and discuss applications and further implications of our result in \sec{app} and \sec{discuss}.

\section{Recursive bound on the expectation of fermionic operators}
\label{sec:recurse}

In this section, we present our first approach for bounding the expectation of fermionic operators, and thereby bounding the fermionic seminorm of Trotter error. We introduce in \sec{recurse_tech} the main techniques used in our approach, including an operator Cauchy-Schwarz inequality, a diagonalization procedure, and a H\"{o}lder-type inequality for the expectation value. We then describe our approach in detail and apply it to prove Eq.\ \eq{general_bound} of our main result \thm{fermionic_seminorm}. The proof is based on induction: we analyze the base case in \sec{recurse_single} and the inductive step in \sec{recurse_multi}, respectively.

\subsection{Main techniques}
\label{sec:recurse_tech}
Recall that the main technical challenge to estimate the simulation error of the electronic Hamiltonian \eq{second_quantized_ham} is to bound the fermionic seminorm $\norm{\left[H_{\gamma_{p+1}},\cdots\left[H_{\gamma_2},H_{\gamma_1}\right]\right]}_{\eta}$, where $\gamma_j=0,1$, $H_0=V$ and $H_1=T$. Applying the commutation relations in \prop{fermionic_commutation}, we see that we need to analyze a general fermionic operator of the form
\begin{equation}
X=\sum_{\pmb{j},\pmb{k},\pmb{l}}w_{\pmb{j},\pmb{k},\pmb{l}}\
\cdots A_{j_x}^\dagger \cdots N_{l_z}\cdots A_{k_y}\cdots
\end{equation}

Our first approach starts by reexpressing the fermionic seminorm of $X$ using the expectation of $X^\dagger X$:
\begin{equation}
\norm{X}_{\eta}=\max_{\ket{\psi_\eta}}\sqrt{\bra{\psi_\eta}X^\dagger X\ket{\psi_\eta}}.
\end{equation}
We note that $X^\dagger X$ is a positive semidefinite operator, and an upper bound of it with respect to the partial ordering of positive semidefiniteness will therefore give a bound on the expectation value. We achieve this by contracting the corresponding indices in $X$ and $X^\dagger$, using either an operator Cauchy-Schwarz inequality (\lem{cauchy}) or diagonalization (\lem{diagonalization}).

\begin{lemma}[Operator Cauchy-Schwarz inequality {\cite[Proposition 3.4]{Otte10}}]
	\label{lem:cauchy}
	For any finite lists of operators $\{B_j\}$ and $\{C_j\}$ with the same cardinality, we have
	\begin{equation}
	-\sum_{j,k}B_j^\dagger C_k^\dagger C_kB_j
	\leq \sum_{j,k}B_j^\dagger C_k^\dagger C_jB_k
	\leq \sum_{j,k}B_j^\dagger C_k^\dagger C_kB_j,
	\end{equation}
	where Hermitian operators are partially ordered according to the positive semidefiniteness.
\end{lemma}
\begin{proof}
	We have
	\begin{equation}
	\begin{aligned}
	0&\leq\sum_{j,k}\left(C_kB_j\mp C_jB_k\right)^\dagger\left(C_kB_j\mp C_jB_k\right)\\
	&=\sum_{j,k}\left(B_j^\dagger C_k^\dagger C_kB_j
	\mp B_k^\dagger C_j^\dagger C_kB_j
	\mp B_j^\dagger C_k^\dagger C_jB_k
	+B_k^\dagger C_j^\dagger C_jB_k\right)\\
	&=2\sum_{j,k}B_j^\dagger C_k^\dagger C_kB_j
	\mp 2\sum_{j,k}B_j^\dagger C_k^\dagger C_jB_k.
	\end{aligned}
	\end{equation}
	This implies
	\begin{equation}
	\pm\sum_{j,k}B_j^\dagger C_k^\dagger C_jB_k \leq \sum_{jk}B_j^\dagger C_k^\dagger C_kB_j,
	\end{equation}
	from which the claimed inequality follows.
\end{proof}

\begin{lemma}[Diagonalization]
	\label{lem:diagonalization}
	For any finite list of operators $\{B_j\}$ and Hermitian coefficient matrix $\mu$, we have
	\begin{equation}
	-\norm{\mu}\sum_jB_j^\dagger B_j
	\leq\sum_{j,k}\mu_{j,k}B_j^\dagger B_k
	\leq\norm{\mu}\sum_jB_j^\dagger B_j,
	\end{equation}
	where Hermitian operators are partially ordered according to the positive semidefiniteness.
\end{lemma}
\begin{proof}
	Since $\mu$ is Hermitian, we may diagonalize it to $\widetilde{\mu}$ by unitary transformation $w$ as
	\begin{equation}
	\mu=w^\dagger\widetilde{\mu}w,
	\end{equation}
	where $\widetilde{\mu}$ is a diagonal matrix with all eigenvalues of $\mu$ as the diagonal elements. We then define $\widetilde{B}_l:=\sum_{k}w_{l,k}B_k$ so that
	\begin{equation}
	\sum_{j,k}\mu_{j,k}B_j^\dagger B_k
	=\sum_{l}\widetilde{\mu}_l\widetilde{B}_l^\dagger\widetilde{B}_l,
	\end{equation}
	which implies
	\begin{equation}
	-\norm{\mu}\sum_{l}\widetilde{B}_l^\dagger\widetilde{B}_l
	\leq\sum_{j,k}\mu_{j,k}B_j^\dagger B_k
	\leq\norm{\mu}\sum_{l}\widetilde{B}_l^\dagger\widetilde{B}_l.
	\end{equation}
	But $\sum_{l}\widetilde{B}_l^\dagger\widetilde{B}_l$ has identity as the coefficient matrix which is invariant under a change of basis:
	\begin{equation}
	\sum_{l}\widetilde{B}_l^\dagger\widetilde{B}_l
	=\sum_jB_j^\dagger B_j.
	\end{equation}
	This completes the proof.
\end{proof}

By applying \lem{cauchy} or \lem{diagonalization}, we can get a bound of $X^\dagger X$ with respect to the partial ordering of positive semidefiniteness, with one pair of the corresponding indices in $X$ and $X^\dagger$ contracted. Indeed, these techniques were used by Otte to establish the boundedness of quadratic fermionic operators in infinite-dimensional Hilbert spaces \cite{Otte10}. However, the difficulty here is that we need to handle more complex products of fermionic operators in the Trotter error estimate. To this end, we prove a H\"{o}lder-type inequality for the expectation value, which allows us to apply \lem{cauchy} and \lem{diagonalization} recursively to get a desired bound.
\begin{lemma}[H\"{o}lder-type inequality for expectation]
	\label{lem:fermionic_holder}
	For any finite lists of fermionic operators $\{B_j\}$ and $\{C_k\}$ with the same cardinality,
	\begin{equation}
	\max_{\ket{\psi_\eta}}\bra{\psi_\eta}\sum_{j}B_j^\dagger C_j^\dagger C_jB_j\ket{\psi_\eta}
	\leq\max_{\ket{\psi_\eta}}\bra{\psi_\eta}\sum_{j}B_j^\dagger B_j\ket{\psi_\eta}
	\max_{k,\ket{\phi_\xi}}\bra{\phi_\xi}C_k^\dagger C_k\ket{\phi_\xi},
	\end{equation}
	where we assume $B_j$ map the $\eta$-electron subspace to the $\xi$-electron subspace and $C_j$ are number-preserving. In terms of the fermionic seminorm, we have
	\begin{equation}
	\Big\Vert\sum_{j}B_j^\dagger C_j^\dagger C_jB_j\Big\Vert_{\eta}
	\leq\Big\Vert\sum_{j}B_j^\dagger B_j\Big\Vert_{\eta}
	\max_k\norm{C_k^\dagger C_k}_{\xi}.
	\end{equation}
\end{lemma}
\begin{proof}
	The claimed inequality follows from
	\begin{equation}
	\begin{aligned}
	\Big\Vert\sum_{j}B_j^\dagger C_j^\dagger C_jB_j\Big\Vert_{\eta}
	&=\Big\Vert\sum_{j}B_j^\dagger \Pi_{\xi}C_j^\dagger C_j\Pi_{\xi}B_j\Big\Vert_{\eta}\\
	&\leq\bigg\Vert\sum_{j}\norm{\Pi_{\xi}C_j^\dagger C_j\Pi_{\xi}}B_j^\dagger B_j\bigg\Vert_{\eta}\\
	&\leq\Big\Vert\sum_{j}B_j^\dagger B_j\Big\Vert_{\eta}
	\max_k\norm{\Pi_{\xi}C_k^\dagger C_k\Pi_{\xi}}\\
	&=\Big\Vert\sum_{j}B_j^\dagger B_j\Big\Vert_{\eta}
	\max_k\norm{C_k^\dagger C_k}_{\xi}.
	\end{aligned}
	\end{equation}
\end{proof}

Using the above lemmas, we can now prove Eq.\ \eq{general_bound} of our main result \thm{fermionic_seminorm} by induction. We analyze the base case in \sec{recurse_single} and the inductive step in \sec{recurse_multi}.

\subsection{Single-layer commutator}
\label{sec:recurse_single}
We now prove Eq.\ \eq{general_bound} of our main result \thm{fermionic_seminorm} by induction. In the base case, we consider simulating the interacting-electronic Hamiltonian \eq{second_quantized_ham} using the first-order formula $\mathscr{S}_1(t)$. We know from \eq{pf1_error} that
\begin{equation}
	\norm{\mathscr{S}_1(t)-e^{-itH}}_{\eta}\leq\frac{t^2}{2}\norm{\left[T,V\right]}_{\eta},
\end{equation}
where $T=\sum_{j,k}\tau_{j,k}A_j^\dagger A_k$ and $V=\sum_{l,m}\nu_{l,m}N_lN_m$. Our goal is to show that
\begin{equation}
	\norm{\left[T,V\right]}_{\eta}=\cO{\norm{\tau}\norm{\nu}_{\max}\eta^2}.
\end{equation}

To this end, we apply \prop{fermionic_commutation} to expand the single-layer commutator $\left[T,V\right]$ into linear combinations of fermionic creation, annihilation, and occupation-number operators. We have
\begin{equation}
\begin{aligned}
	\left[T,V\right]
	&=\sum_{j,k,l,m}\tau_{j,k}\nu_{l,m}\left[A_j^\dagger A_k,N_lN_m\right]\\
	&=\sum_{j,k,l,m}\tau_{j,k}\nu_{l,m}A_j^\dagger\left[ A_k,N_lN_m\right]+\sum_{j,k,l,m}\tau_{j,k}\nu_{l,m}\left[A_j^\dagger,N_lN_m\right]A_k\\
	&=\sum_{j,k,m}\tau_{j,k}\nu_{k,m}A_j^\dagger A_kN_m
	+\sum_{j,k,l}\tau_{j,k}\nu_{l,k}A_j^\dagger N_lA_k\\
	&\quad -\sum_{j,k,m}\tau_{j,k}\nu_{j,m}A_j^\dagger N_mA_k
	-\sum_{j,k,l}\tau_{j,k}\nu_{l,j}N_lA_j^\dagger A_k.
\end{aligned}
\end{equation}
At this stage, it is possible to directly bound the terms in the last equality using \lem{cauchy}, \lem{diagonalization}, and \lem{fermionic_holder} from the previous subsection. However, we will further commute the occupation-number operator in between the creation and annihilation operators, obtaining
\begin{equation}
\begin{aligned}
	\left[T,V\right]
	&=\sum_{j,k,m}\tau_{j,k}\nu_{k,m}A_j^\dagger N_mA_k
	+\sum_{j,k}\tau_{j,k}\nu_{k,k}A_j^\dagger A_k
	+\sum_{j,k,l}\tau_{j,k}\nu_{l,k}A_j^\dagger N_lA_k\\
	&\quad -\sum_{j,k,m}\tau_{j,k}\nu_{j,m}A_j^\dagger N_mA_k
	-\sum_{j,k}\tau_{j,k}\nu_{j,j}A_j^\dagger A_k
	-\sum_{j,k,l}\tau_{j,k}\nu_{l,j}A_j^\dagger N_l A_k.
\end{aligned}
\end{equation}
This additional commutation leads to an error bound with the same asymptotic scaling but a slightly larger prefactor. The benefit is that the analysis can be directly extended to handle the inductive step in the next subsection.

\begin{proposition}[Structure of single-layer commutator]
	\label{prop:singlelayer_structure}
	Let $H=T+V=\sum_{j,k}\tau_{j,k}A_j^\dagger A_k+\sum_{l,m}\nu_{l,m}N_lN_m$ be an interacting-electronic Hamiltonian \eq{second_quantized_ham}. Then, the commutator $[T,V]$ has the expansion
	\begin{equation}
	\label{eq:singlelayer_expand}
	\begin{aligned}
	\left[T,V\right]
	&=\sum_{j,k,m}\tau_{j,k}\nu_{k,m}A_j^\dagger N_mA_k
	+\sum_{j,k}\tau_{j,k}\nu_{k,k}A_j^\dagger A_k
	+\sum_{j,k,l}\tau_{j,k}\nu_{l,k}A_j^\dagger N_lA_k\\
	&\quad -\sum_{j,k,m}\tau_{j,k}\nu_{j,m}A_j^\dagger N_mA_k
	-\sum_{j,k}\tau_{j,k}\nu_{j,j}A_j^\dagger A_k
	-\sum_{j,k,l}\tau_{j,k}\nu_{l,j}A_j^\dagger N_l A_k.
	\end{aligned}
	\end{equation}
\end{proposition}

It is worth noting that the above six terms from the expansion of $[T,V]$ share a similar structure. Specifically, they all consist of creation operator $A_j^\dagger$, annihilation operator $A_k$, and (possibly) occupation-number operator $N_l$, with one coefficient matrix $\tau$ and one matrix $\nu$. The main difference between these terms is that the coefficient matrix $\nu$ is acting on different indices. See \fig{single_fermionic_chain} for a graph illustration of this structure. 

\begin{figure}[t]
	\centering
	\begin{subfigure}[t]{0.3\linewidth}
		\centering
		\begin{tikzpicture}
		\tikzstyle{hidden}=[draw=orange!75,fill=orange!50]
		
		\draw (1, 3) node[circle, black, draw](j){$j$};
		\draw (3, 3) node[circle, black, draw](k){$k$};
		\draw (3, 1) node[circle, black, draw](m){$m$};
		
		\draw[-] (j) -- node[above] {$\tau$} (k);
		\draw[-] (k) -- node[right] {$\nu$} (m);
		
		\end{tikzpicture}
		\caption{$\sum_{j,k,m}\tau_{j,k}\nu_{k,m}A_j^\dagger N_mA_k$}
	\end{subfigure}
	\begin{subfigure}[t]{0.3\linewidth}
		\centering
		\begin{tikzpicture}
		\tikzstyle{hidden}=[draw=orange!75,fill=orange!50]
		
		\draw (1, 1) node[circle, black, draw](j){$j$};
		\draw (3, 1) node[circle, black, draw](k){$k$};
		
		\draw[-] (j) -- node[above] {$\tau$} (k);
		\path[every loop/.style={in=300,out=240,looseness=8}] (k) edge [loop below] node {$\nu$} (k);
		
		\end{tikzpicture}
		\caption{$\sum_{j,k}\tau_{j,k}\nu_{k,k}A_j^\dagger A_k$}
	\end{subfigure}
	\begin{subfigure}[t]{0.3\linewidth}
		\centering
		\begin{tikzpicture}
		\tikzstyle{hidden}=[draw=orange!75,fill=orange!50]
		
		\draw (1, 3) node[circle, black, draw](j){$j$};
		\draw (3, 3) node[circle, black, draw](k){$k$};
		\draw (3, 1) node[circle, black, draw](l){$l$};
		
		\draw[-] (j) -- node[above] {$\tau$} (k);
		\draw[-] (k) -- node[right] {$\nu$} (l);
		
		\end{tikzpicture}
		\caption{$\sum_{j,k,l}\tau_{j,k}\nu_{l,k}A_j^\dagger N_lA_k$}
	\end{subfigure}
	\\\vspace*{0.5cm}
	\begin{subfigure}[t]{0.3\linewidth}
		\centering
		\begin{tikzpicture}
		\tikzstyle{hidden}=[draw=orange!75,fill=orange!50]
		
		\draw (1, 3) node[circle, black, draw](j){$j$};
		\draw (3, 3) node[circle, black, draw](k){$k$};
		\draw (1, 1) node[circle, black, draw](m){$m$};
		
		\draw[-] (j) -- node[above] {$\tau$} (k);
		\draw[-] (j) -- node[right] {$\nu$} (m);
		
		\end{tikzpicture}
		\caption{$\sum_{j,k,m}\tau_{j,k}\nu_{j,m}A_j^\dagger N_mA_k$}
	\end{subfigure}
	\begin{subfigure}[t]{0.3\linewidth}
		\centering
		\begin{tikzpicture}
		\tikzstyle{hidden}=[draw=orange!75,fill=orange!50]
		
		\draw (1, 1) node[circle, black, draw](j){$j$};
		\draw (3, 1) node[circle, black, draw](k){$k$};
		
		\draw[-] (j) -- node[above] {$\tau$} (k);
		\path[every loop/.style={in=300,out=240,looseness=8}] (j) edge [loop below] node {$\nu$} (j);
		
		\end{tikzpicture}
		\caption{$\sum_{j,k}\tau_{j,k}\nu_{j,j}A_j^\dagger A_k$}
	\end{subfigure}
	\begin{subfigure}[t]{0.3\linewidth}
		\centering
		\begin{tikzpicture}
		\tikzstyle{hidden}=[draw=orange!75,fill=orange!50]
		
		\draw (1, 3) node[circle, black, draw](j){$j$};
		\draw (3, 3) node[circle, black, draw](k){$k$};
		\draw (1, 1) node[circle, black, draw](l){$l$};
		
		\draw[-] (j) -- node[above] {$\tau$} (k);
		\draw[-] (j) -- node[right] {$\nu$} (l);
		
		\end{tikzpicture}
		\caption{$\sum_{j,k,l}\tau_{j,k}\nu_{l,j}A_j^\dagger N_l A_k$}
	\end{subfigure}
	\caption{Graph illustration of the expansion terms from the single-layer commutator $[T,V]$. Here, vertices in the graph denote the indices in the summation and edges represent the coefficients. Note that the graphs can be made directional so that they are one-to-one corresponding to fermionic operators, although this is not needed in our analysis and will not be further pursued here.}
	\label{fig:single_fermionic_chain}
\end{figure}
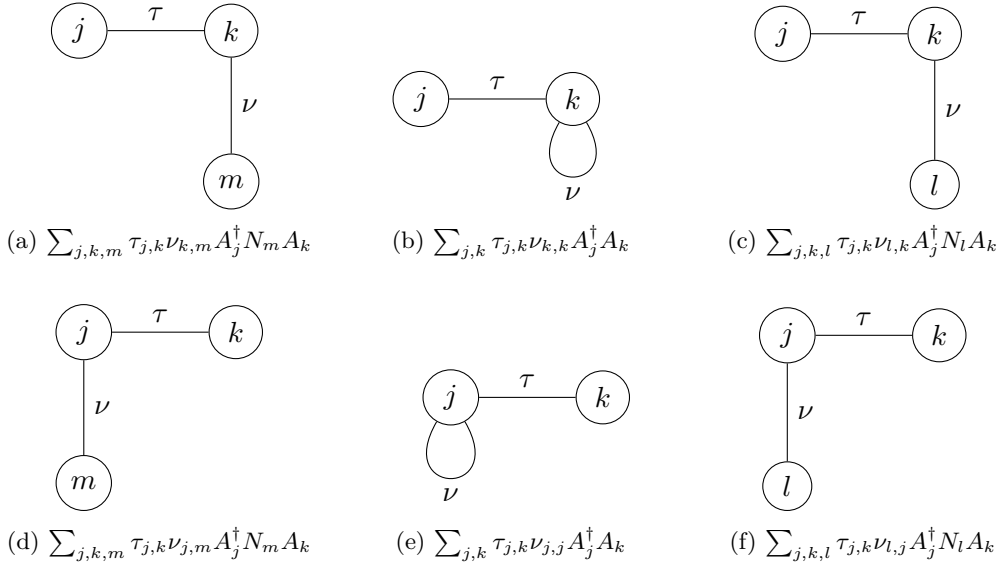

We now bound the asymptotic scaling of the fermionic seminorm for each of the six terms in the commutator expansion.
\begin{proposition}[Fermionic seminorm of single-layer commutator]
	\label{prop:singlelayer_seminorm}
	Let $H=T+V=\sum_{j,k}\tau_{j,k}A_j^\dagger A_k+\sum_{l,m}\nu_{l,m}N_lN_m$ be an interacting-electronic Hamiltonian \eq{second_quantized_ham}. Then, 
	\begin{enumerate}
		\item $\norm{\sum_{j,k,m}\tau_{j,k}\nu_{k,m}A_j^\dagger N_mA_k}_\eta\leq\norm{\tau}\norm{\nu}_{\max}\eta^2$;
		\item $\norm{\sum_{j,k}\tau_{j,k}\nu_{k,k}A_j^\dagger A_k}_\eta\leq\norm{\tau}\norm{\nu}_{\max}\eta$;
		\item $\norm{\sum_{j,k,l}\tau_{j,k}\nu_{l,k}A_j^\dagger N_lA_k}_\eta\leq\norm{\tau}\norm{\nu}_{\max}\eta^2$;
		\item $\norm{\sum_{j,k,m}\tau_{j,k}\nu_{j,m}A_j^\dagger N_mA_k}_\eta\leq\norm{\tau}\norm{\nu}_{\max}\eta^2$;
		\item $\norm{\sum_{j,k}\tau_{j,k}\nu_{j,j}A_j^\dagger A_k}_\eta\leq\norm{\tau}\norm{\nu}_{\max}\eta$;
		\item $\norm{\sum_{j,k,l}\tau_{j,k}\nu_{l,j}A_j^\dagger N_l A_k}_\eta\leq\norm{\tau}\norm{\nu}_{\max}\eta^2$.
	\end{enumerate}
\end{proposition}
\begin{proof}
	We present the proof of the first two statements here. The remaining justifications proceed in a similar way and are left to \append{singlelayer}.

	Letting $X=\sum_{j,k,m}\tau_{j,k}\nu_{k,m}A_j^\dagger N_mA_k$, we have $\norm{X}_\eta=\sqrt{\norm{X^\dagger X}_\eta}$. Now,
	\begin{equation}
	\begin{aligned}
		X^\dagger X
		&=\sum_{j_1,k_1,m_1,j_2,k_2,m_2}\bar{\tau}_{j_1,k_1}\bar{\nu}_{k_1,m_1}\tau_{j_2,k_2}\nu_{k_2,m_2}
		A_{k_1}^\dagger N_{m_1}A_{j_1} A_{j_2}^\dagger N_{m_2}A_{k_2}\\
		&=\sum_{j_1,k_1,m_1,k_2,m_2}\bar{\tau}_{j_1,k_1}\bar{\nu}_{k_1,m_1}\tau_{j_1,k_2}\nu_{k_2,m_2}
		A_{k_1}^\dagger N_{m_1} N_{m_2}A_{k_2}\\
		&\quad-\sum_{j_1,k_1,m_1,j_2,k_2,m_2}\bar{\tau}_{j_1,k_1}\bar{\nu}_{k_1,m_1}\tau_{j_2,k_2}\nu_{k_2,m_2}
		A_{k_1}^\dagger N_{m_1}A_{j_2}^\dagger A_{j_1}  N_{m_2}A_{k_2},
	\end{aligned}
	\end{equation}
	where $\bar{\tau}_{j_1,k_1}$ is the complex conjugate of $\tau_{j_1,k_1}$ and we have used the anti-commutation relation $A_{j_1} A_{j_2}^\dagger+A_{j_2}^\dagger A_{j_1}=\delta_{j_1,j_2}I$. For the second term, we let $B_{j_1}^\dagger=\sum_{k_1,m_1}\bar{\tau}_{j_1,k_1}\bar{\nu}_{k_1,m_1}A_{k_1}^\dagger N_{m_1}$ and apply the operator Cauchy-Schwarz inequality (\lem{cauchy}):
	\begin{equation}
	\label{eq:cauchy_cal}
	\begin{aligned}
		&-\sum_{j_1,k_1,m_1,j_2,k_2,m_2}\bar{\tau}_{j_1,k_1}\bar{\nu}_{k_1,m_1}\tau_{j_2,k_2}\nu_{k_2,m_2}
		A_{k_1}^\dagger N_{m_1}A_{j_2}^\dagger A_{j_1}  N_{m_2}A_{k_2}\\
		=&-\sum_{j_1,j_2}B_{j_1}^\dagger A_{j_2}^\dagger A_{j_1}B_{j_2}
		\leq\sum_{j_1,j_2}B_{j_1}^\dagger A_{j_2}^\dagger A_{j_2}B_{j_1}
		=\sum_{j_1}B_{j_1}^\dagger NB_{j_1}\\
		=&\sum_{j_1}\sum_{k_1,m_1,k_2,m_2}\bar{\tau}_{j_1,k_1}\bar{\nu}_{k_1,m_1}\tau_{j_1,k_2}\nu_{k_2,m_2}
		A_{k_1}^\dagger N_{m_1}NN_{m_2}A_{k_2}\\
		=&\sum_{j_1}\sum_{k_1,m_1,k_2,m_2}\bar{\tau}_{j_1,k_1}\bar{\nu}_{k_1,m_1}\tau_{j_1,k_2}\nu_{k_2,m_2}
		A_{k_1}^\dagger N_{m_1}N_{m_2}A_{k_2}N\\
		&-\sum_{j_1}\sum_{k_1,m_1,k_2,m_2}\bar{\tau}_{j_1,k_1}\bar{\nu}_{k_1,m_1}\tau_{j_1,k_2}\nu_{k_2,m_2}
		A_{k_1}^\dagger N_{m_1}N_{m_2}A_{k_2}.
	\end{aligned}
	\end{equation}
	This implies
	\begin{equation}
	\label{eq:single_comm_1}
	\begin{aligned}
		X^\dagger X
		&\leq\sum_{j_1}\sum_{k_1,m_1,k_2,m_2}\bar{\tau}_{j_1,k_1}\bar{\nu}_{k_1,m_1}\tau_{j_1,k_2}\nu_{k_2,m_2}
		A_{k_1}^\dagger N_{m_1}N_{m_2}A_{k_2}N\\
		&=\sum_{k_1,k_2}\left(\sum_{j_1}\bar{\tau}_{j_1,k_1}\tau_{j_1,k_2}\right)
		\left(\sum_{m_1}\bar{\nu}_{k_1,m_1}A_{k_1}^\dagger N_{m_1}\right)\left(\sum_{m_2}\nu_{k_2,m_2}N_{m_2}A_{k_2}\right)N.
	\end{aligned}
	\end{equation}

	Note that $\sum_{j_1}\bar{\tau}_{j_1,k_1}\tau_{j_1,k_2}$ gives the $(k_1,k_2)$ matrix element of $\tau^\dagger \tau$. Then, we define $C_{k_1}^\dagger=\sum_{m_1}\bar{\nu}_{k_1,m_1}A_{k_1}^\dagger N_{m_1}$ and perform diagonalization (\lem{diagonalization}):
	\begin{equation}
	\label{eq:single_comm_2}
	\begin{aligned}
		&\sum_{k_1,k_2}\left(\sum_{j_1}\bar{\tau}_{j_1,k_1}\tau_{j_1,k_2}\right)
		\left(\sum_{m_1}\bar{\nu}_{k_1,m_1}A_{k_1}^\dagger N_{m_1}\right)\left(\sum_{m_2}\nu_{k_2,m_2}N_{m_2}A_{k_2}\right)\\
		=&\sum_{k_1,k_2}\left(\tau^\dagger \tau\right)_{k_1,k_2}C_{k_1}^\dagger C_{k_2}
		\leq\norm{\tau^\dagger \tau}\sum_{k_1}C_{k_1}^\dagger C_{k_1}\\
		=&\norm{\tau^\dagger \tau}\sum_{k_1}\sum_{m_1,m_2}\bar{\nu}_{k_1,m_1}\nu_{k_1,m_2}A_{k_1}^\dagger N_{m_1}N_{m_2}A_{k_1}.
	\end{aligned}
	\end{equation}

	Now that the indices $k_1$ and $k_2$ are contracted, we can apply the H\"{o}lder-type inequality for the expectation value (\lem{fermionic_holder}). To this end, we let $D_{k_1}^\dagger=\sum_{m_1}\bar{\nu}_{k_1,m_1}N_{m_1}$ and compute
	\begin{equation}
	\label{eq:single_comm_3}
	\begin{aligned}
		\norm{\sum_{k_1}\sum_{m_1,m_2}\bar{\nu}_{k_1,m_1}\nu_{k_1,m_2}A_{k_1}^\dagger N_{m_1}N_{m_2}A_{k_1}}_\eta
		&=\norm{\sum_{k_1}A_{k_1}^\dagger D_{k_1}^\dagger D_{k_1}A_{k_1}}_\eta\\
		&\leq\norm{\sum_{k_1}A_{k_1}^\dagger A_{k_1}}_\eta \max_{k_1}\norm{D_{k_1}^\dagger D_{k_1}}_{\eta-1}.
	\end{aligned}
	\end{equation}
	The first factor can be directly bounded as
	\begin{equation}
	\label{eq:single_comm_4}
		\norm{\sum_{k_1}A_{k_1}^\dagger A_{k_1}}_\eta
		=\norm{N}_\eta=\eta.
	\end{equation}
	For the second factor, we have
	\begin{equation}
	\begin{aligned}
		D_{k_1}^\dagger D_{k_1}
		&=\sum_{m_1,m_2}\bar{\nu}_{k_1,m_1}\nu_{k_1,m_2}N_{m_1}N_{m_2}
		=\sum_{m_1,m_2}\bar{\nu}_{k_1,m_1}\nu_{k_1,m_2}N_{m_1}N_{m_2}N_{m_1}\\
		&\leq\norm{\nu}_{\max}^2\sum_{m_1,m_2}N_{m_1}N_{m_2}N_{m_1}
		=\norm{\nu}_{\max}^2N^2,
	\end{aligned}
	\end{equation}
	which implies
	\begin{equation}
	\label{eq:single_comm_5}
		\norm{D_{k_1}^\dagger D_{k_1}}_{\eta-1}
		\leq\norm{\nu}_{\max}^2\eta^2.
	\end{equation}
	Combining \eq{single_comm_1}, \eq{single_comm_2}, \eq{single_comm_3}, \eq{single_comm_4}, and \eq{single_comm_5} establishes the first statement.

	For the second statement, we let $X=\sum_{j,k}\tau_{j,k}\nu_{k,k}A_j^\dagger A_k$ and compute
	\begin{equation}
	\begin{aligned}
		X^\dagger X
		&=\sum_{j_1,k_1,j_2,k_2}\bar{\tau}_{j_1,k_1}\bar{\nu}_{k_1,k_1}\tau_{j_2,k_2}\nu_{k_2,k_2}A_{k_1}^\dagger A_{j_1}A_{j_2}^\dagger A_{k_2}\\
		&=\sum_{j_1,k_1,k_2}\bar{\tau}_{j_1,k_1}\bar{\nu}_{k_1,k_1}\tau_{j_1,k_2}\nu_{k_2,k_2}A_{k_1}^\dagger A_{k_2}
		-\sum_{j_1,k_1,j_2,k_2}\bar{\tau}_{j_1,k_1}\bar{\nu}_{k_1,k_1}\tau_{j_2,k_2}\nu_{k_2,k_2}A_{k_1}^\dagger A_{j_2}^\dagger A_{j_1}A_{k_2}.
	\end{aligned}
	\end{equation}
	Applying \lem{cauchy},
	\begin{equation}
	\begin{aligned}
		X^\dagger X
		&\leq\sum_{j_1,k_1,k_2}\bar{\tau}_{j_1,k_1}\bar{\nu}_{k_1,k_1}\tau_{j_1,k_2}\nu_{k_2,k_2}A_{k_1}^\dagger A_{k_2}
		+\sum_{j_1,k_1,j_2,k_2}\bar{\tau}_{j_1,k_1}\bar{\nu}_{k_1,k_1}\tau_{j_1,k_2}\nu_{k_2,k_2}A_{k_1}^\dagger A_{j_2}^\dagger A_{j_2}A_{k_2}\\
		&=\sum_{j_1,k_1,k_2}\bar{\tau}_{j_1,k_1}\bar{\nu}_{k_1,k_1}\tau_{j_1,k_2}\nu_{k_2,k_2}A_{k_1}^\dagger A_{k_2}
		+\sum_{j_1,k_1,k_2}\bar{\tau}_{j_1,k_1}\bar{\nu}_{k_1,k_1}\tau_{j_1,k_2}\nu_{k_2,k_2}A_{k_1}^\dagger NA_{k_2}\\
		&=\sum_{j_1,k_1,k_2}\bar{\tau}_{j_1,k_1}\bar{\nu}_{k_1,k_1}\tau_{j_1,k_2}\nu_{k_2,k_2}A_{k_1}^\dagger A_{k_2}N\\
		&=\sum_{k_1,k_2}\left(\sum_{j_1}\bar{\tau}_{j_1,k_1}\tau_{j_1,k_2}\right)\bar{\nu}_{k_1,k_1}\nu_{k_2,k_2}A_{k_1}^\dagger A_{k_2}N.
	\end{aligned}
	\end{equation}
	Performing diagonalization using \lem{diagonalization}, we have
	\begin{equation}
		X^\dagger X\leq\norm{\tau}^2\sum_{k_1}\bar{\nu}_{k_1,k_1}\nu_{k_1,k_1}A_{k_1}^\dagger A_{k_1}N.
	\end{equation}
	Note that we could directly bound the above operators as $\norm{\tau}^2\norm{\nu}_{\max}^2N^2$ and thereby complete the proof. But we choose to instead apply \lem{fermionic_holder} so that the analysis can then be directly extended to analyze multilayer nested commutators. We have
	\begin{equation}
	\begin{aligned}
		\norm{X^\dagger X}_{\eta}
		&\leq\norm{\norm{\tau}^2\sum_{k_1}\bar{\nu}_{k_1,k_1}\nu_{k_1,k_1}A_{k_1}^\dagger A_{k_1}N}_{\eta}
		=\norm{\tau}^2\eta\norm{\sum_{k_1}\bar{\nu}_{k_1,k_1}\nu_{k_1,k_1}A_{k_1}^\dagger A_{k_1}}_{\eta}\\
		&\leq\norm{\tau}^2\eta\norm{\sum_{k_1}A_{k_1}^\dagger A_{k_1}}_{\eta}\max_{k_1}\norm{\bar{\nu}_{k_1,k_1}\nu_{k_1,k_1}I}_{\eta-1}
		\leq\norm{\tau}^2\norm{\nu}_{\max}^2\eta^2.
	\end{aligned}
	\end{equation}
	The proof of the second statement is now completed. See \append{singlelayer} for the proof of the remaining statements.
\end{proof}

\subsection{Multilayer nested commutators}
\label{sec:recurse_multi}
We now analyze the error of simulating the interacting-electronic Hamiltonian \eq{second_quantized_ham} using a general $p$th-order formula $\mathscr{S}_p(t)$. We know from \eq{pf2k_eta} that
\begin{equation}
\norm{\mathscr{S}_p(t)-e^{-itH}}_{\eta}=\cO{\max_{\pmb{\gamma}}\norm{\left[H_{\gamma_{p+1}},\cdots\left[H_{\gamma_2},H_{\gamma_1}\right]\right]}_{\eta}t^{p+1}},
\end{equation}
where $H_0=V=\sum_{l,m}\nu_{l,m}N_lN_m$ and $H_1=T=\sum_{j,k}\tau_{j,k}A_j^\dagger A_k$. Our goal is to show that
\begin{equation}
	\norm{\left[H_{\gamma_{p+1}},\cdots\left[H_{\gamma_2},H_{\gamma_1}\right]\right]}_{\eta}
	=\cO{\left(\norm{\tau}+\norm{\nu}_{\max}\eta\right)^{p-1}
		\norm{\tau}\norm{\nu}_{\max}\eta^2}
\end{equation}
for each multilayer nested commutator $\left[H_{\gamma_{p+1}},\cdots\left[H_{\gamma_2},H_{\gamma_1}\right]\right]$.

To this end, we assume that $\left[H_{\gamma_{p+1}},\cdots\left[H_{\gamma_2},H_{\gamma_1}\right]\right]$ is expressed as a fermionic operator of the form
\begin{equation}
	\sum_{\pmb{j},\pmb{k},\pmb{l}}w_{\pmb{j},\pmb{k},\pmb{l}}\
	\cdots A_{j_x}^\dagger \cdots N_{l_z}\cdots A_{k_y}\cdots
\end{equation}
and analyze its commutator with either $T$ or $V$. For the commutator with $T$, we have from \prop{fermionic_commutation}
\begin{equation}
\label{eq:commutator_t}
	\left[A_j^\dagger A_k,A_{j_x}^\dagger\right]=\delta_{k,j_x}A_j^\dagger,\quad
	\left[A_j^\dagger A_k,A_{k_y}\right]=-\delta_{k_y,j}A_k,\quad
	\left[A_j^\dagger A_k,N_{l_z}\right]
	=\delta_{k,l_z}A_j^\dagger A_k
	-\delta_{j,l_z}A_j^\dagger A_k.
\end{equation}
To develop some intuitions about these commutations, we introduce the notion of \emph{fermionic chain}, which refers to a product of fermionic operators that has a creation operator on the left and an annihilation operator on the right. Then, the above commutations either extend an existing fermionic chain (in the case where commutator is taken with $A_{j_x}^\dagger$ or $A_{k_y}$), or create a new chain (in the case where commutator is taken with $N_{l_z}$). On the other hand, we also apply \prop{fermionic_commutation} to compute the commutator with $V$:
\begin{equation}
\label{eq:commutator_v}
\begin{aligned}
	\left[N_lN_m,A_{j_x}^\dagger\right]
	&=\delta_{m,j_x}N_lA_{j_x}^\dagger+\delta_{l,j_x}A_{j_x}^\dagger N_m
	=\delta_{m,j_x}A_{j_x}^\dagger N_l+\delta_{l,j_x}A_{j_x}^\dagger N_m+\delta_{m,j_x}\delta_{l,j_x}A_{j_x}^\dagger,\\
	\left[N_lN_m,A_{k_x}\right]
	&=-\delta_{m,k_x}N_lA_{k_x}-\delta_{l,k_x}A_{k_x} N_m
	=-\delta_{m,k_x}N_lA_{k_x}-\delta_{l,k_x} N_mA_{k_x}-\delta_{m,k_x}\delta_{l,k_x} A_{k_x}.
\end{aligned}
\end{equation}
Unlike the commutator with $T$, these commutations do not extend an existing chain or create a new chain. Rather, their effect is to append occupation-number operators to an existing chain.

We now apply \eq{commutator_t} and \eq{commutator_v} iteratively to expand a general multilayer nested commutator $\left[H_{\gamma_{p+1}},\cdots\left[H_{\gamma_2},H_{\gamma_1}\right]\right]$. We summarize the structure of the resulting operator in the following proposition.
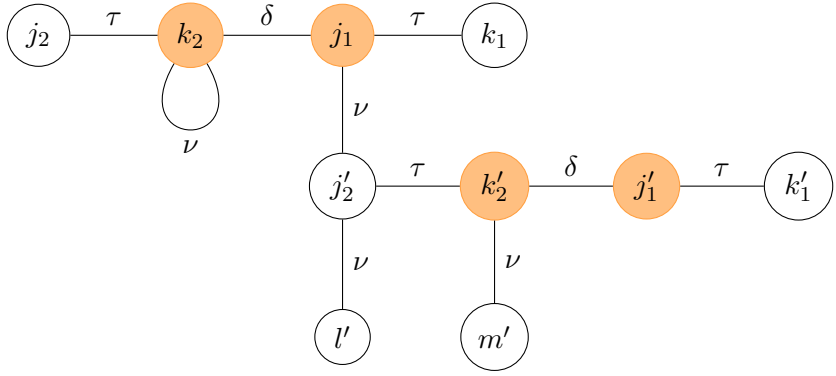
\begin{figure}[t]
	\centering
	\begin{tikzpicture}
	\tikzstyle{hidden}=[draw=orange!75,fill=orange!50]

	\draw (1, 5) node[circle, black, draw](j2){$j_2$};
	\draw (3, 5) node[circle, hidden, draw](k2){$k_2$};
	\draw (5, 5) node[circle, hidden, draw](j1){$j_1$};
	\draw (7, 5) node[circle, black, draw](k1){$k_1$};

	\draw (5, 3) node[circle, black, draw](j2tilde){${j}_2'$};
	\draw (7, 3) node[circle, hidden, draw](k2tilde){${k}_2'$};
	\draw (9, 3) node[circle, hidden, draw](j1tilde){${j}_1'$};
	\draw (11, 3) node[circle, black, draw](k1tilde){${k}_1'$};

	\draw (5, 1) node[circle, black, draw](ltilde){${l}'$};
	\draw (7, 1) node[circle, black, draw](mtilde){${m}'$};

	\draw[-] (j2) -- node[above] {$\tau$} (k2);
	\draw[-] (k2) -- node[above] {$\delta$} (j1);
	\draw[-] (j1) -- node[above] {$\tau$} (k1);
	\path[every loop/.style={in=300,out=240,looseness=8}] (k2) edge [loop below] node {$\nu$} (k2);

	\draw[-] (j1) -- node[right] {$\nu$} (j2tilde);

	\draw[-] (j2tilde) -- node[above] {$\tau$} (k2tilde);
	\draw[-] (k2tilde) -- node[above] {$\delta$} (j1tilde);
	\draw[-] (j1tilde) -- node[above] {$\tau$} (k1tilde);

	\draw[-] (j2tilde) -- node[right] {$\nu$} (ltilde);
	\draw[-] (k2tilde) -- node[right] {$\nu$} (mtilde);

	\end{tikzpicture}
	\caption{Graph illustration of the fermionic chain $X=\sum_{\pmb{j},\pmb{k}}\tau_{j_2,k_2}\delta_{k_2,j_1}\tau_{j_1,k_1} A_{j_2}^\dagger C_{2,1}B_{1,1} A_{k_1}$ with $C_{2,1}=\nu_{k_2,k_2}$. Here, $B_{1,1}=\sum_{{\pmb{j}}',{\pmb{k}}'}\beta_{1,2}\tau_{{j}_2',{k}_2'}\delta_{{k}_2',{j}_1'}\tau_{{j}_1',{k}_1'} A_{{j}_2'}^\dagger {B}_{2,1}'{C}_{2,1}' A_{{k}_1'}$ is a fermionic subchain with $\beta_{1,2}=\nu_{{j}_2',j_1}$, ${B}_{2,1}'=\sum_{l'}\nu_{{l}',{j}_2'}N_{{l}'}$ and ${C}_{2,1}'=\sum_{{m}'}\nu_{{k}_2',{m}'}N_{{m}'}$. Vertices in the graph denote the indices in the summation, whereas edges represent the coefficients. We color a vertex if there is no fermionic operator corresponding to this index (due to taking nested commutators).}
	\label{fig:fermionic_chain}
\end{figure}
\begin{proposition}[Structure of multilayer nested commutators]
	\label{prop:multilayer}
	Let $H=T+V=\sum_{j,k}\tau_{j,k}A_j^\dagger A_k+\sum_{l,m}\nu_{l,m}N_lN_m$ be an interacting-electronic Hamiltonian as in \eq{second_quantized_ham}. Then, each nested commutator $\left[H_{\gamma_{p+1}},\cdots\left[H_{\gamma_2},H_{\gamma_1}\right]\right]$ where $H_0=V$ and $H_1=T$ is a linear combination of \emph{fermionic chains}:
	\begin{equation}
	\label{eq:fermionic_chain}
		X=\sum_{\pmb j,\pmb k}\prod_{x=1}^{q}\tau_{j_x,k_x}\prod_{x=1}^{q-1}\delta_{k_{x+1},j_x}
		\cdot A_{j_q}^\dagger\prod_{x=1}^{q}\left(\prod_{y}B_{x,y}\prod_{z}C_{x,z}\right)A_{k_1}
	\end{equation}
	for some integer $q \leq p$. Here, we have $B_{x,y}=\sum_{l}\nu_{l,j_x}N_l$, $\sum_{m}\nu_{j_x,m}N_m$, $\nu_{j_x,j_x}I$ and $C_{x,z}=\sum_{l}\nu_{l,k_x}N_l$, $\sum_{m}\nu_{k_x,m}N_m$, $\nu_{k_x,k_x}I$, or they define \emph{fermionic subchains}:
	\begin{equation}
	\label{eq:fermionic_subchain}
	\begin{aligned}
	B_{x,y}&=
	\sum_{{\pmb j}',{\pmb k}'}\beta_{x,{x}_0'}
	\prod_{{x}'=1}^{{q}'}\tau_{{j}_{{x}'}',{k}_{{x}'}'}\prod_{{x}'=1}^{{q}'-1}\delta_{{k}_{{x}'+1}',{j}_{{x}'}'}
	\cdot A_{{j}_{{q}'}'}^\dagger\prod_{{x}'=1}^{{q}'}\left(\prod_{{y}'}{B}_{{x}',{y}'}'\prod_{{z}'}{C}_{{x}',{z}'}'\right)A_{{k}_1'},\\
	C_{x,z}&=
	\sum_{{\pmb j''},{\pmb k''}}\chi_{x,{x}_0''}
	\prod_{{x}''=1}^{{q}''}\tau_{{j}_{{x}''}'',{k}_{{x}''}''}\prod_{{x}''=1}^{{q}''-1}\delta_{{k}''_{{x}''+1},{j}_{{x}''}''}
	\cdot A_{{j}_{{q}''}''}^\dagger\prod_{{x}''=1}^{{q}''}\left(\prod_{{y}''}{B}_{{x}'',{y}''}''\prod_{{z}''}{C}_{{x}'',{z}''}''\right)A_{{k}_1''}.
	\end{aligned}
	\end{equation}
	The definition of fermionic subchain is similar to that of the fermionic chain, except we have $\beta_{x,{x}_0'}=\nu_{{j}_{{x}_0'}',j_x}$, $\nu_{j_x,{j}_{{x}_0'}'}$, $\nu_{{k}_{{x}_0'}',j_x}$, $\nu_{j_x,{k}_{{x}_0'}'}$ for some ${x}_0'\leq{q}'$ and $\chi_{x,{x}_0''}=\nu_{{j}_{{x}_0''}'',k_x}$, $\nu_{k_x,{j}_{{x}_0''}''}$, $\nu_{{k}_{{x}_0''}'',k_x}$, $\nu_{k_x,{k}_{{x}_0''}''}$ for some ${x}_0''\leq{q}''$.\footnotemark\ See \fig{fermionic_chain} for a graph illustration of this structure.
	\footnotetext{For readability, we have omitted the dependence of $\pmb{j},\pmb{k},{\pmb j}',{\pmb k}',{\pmb j}'',{\pmb k}''$ in the definition of fermionic chain and subchain. When written in full, operator $B_{x,y}=B_{x,y}(\pmb{j},\pmb{k})$ will depend on $\pmb{j},\pmb{k}$ and similar modifications apply to other operators.}
	
	Furthermore, there are at most $6^{p}p!$ fermionic chains in each $\left[H_{\gamma_{p+1}},\cdots\left[H_{\gamma_2},H_{\gamma_1}\right]\right]$.\footnotemark\ Within each chain, coefficient $\tau$ appears $\abs{\pmb \gamma}$ times and $\nu$ appears $p+1-\abs{\pmb \gamma}$ times, where $\abs{\pmb \gamma}:=\sum_{s=1}^{p+1}\gamma_s$. All $B_{x,y}$, $C_{x,z}$ and hence the entire chain are number-preserving. 
	\footnotetext{The exact number of fermionic chains does not matter as long as it only depends on $p$.}
\end{proposition}
\begin{proof}
	We will analyze the structure of multilayer nested commutators by induction. In the base case where $p=1$, we have $\left[H_{\gamma_2},H_{\gamma_1}\right]=[T,V]$. This commutator has the expansion \eq{singlelayer_expand} with six terms, each of which is indeed a fermionic chain and number-preserving, containing coefficient $\tau$ and $\nu$ each once. This completes the proof of the base case.

	Assuming the claim holds for the nested commutator $\left[H_{\gamma_{p+1}},\cdots\left[H_{\gamma_2},H_{\gamma_1}\right]\right]$, we now consider the structure of $\left[H_{\gamma_{p+2}},\cdots\left[H_{\gamma_2},H_{\gamma_1}\right]\right]$. By induction, this nested commutator is a linear combination of
	\begin{equation}
	\begin{aligned}
		&\sum_{\pmb j,\pmb k}\prod_{x=1}^{q}\tau_{j_x,k_x}\prod_{x=1}^{q-1}\delta_{k_{x+1},j_x}
		\cdot \left[T,A_{j_q}^\dagger\right]\prod_{x=1}^{q}\left(\prod_{y}B_{x,y}\prod_{z}C_{x,z}\right)A_{k_1},\\
		&\sum_{\pmb j,\pmb k}\prod_{x=1}^{q}\tau_{j_x,k_x}\prod_{x=1}^{q-1}\delta_{k_{x+1},j_x}
		\cdot A_{j_q}^\dagger\prod_{x=1}^{q}\left(\prod_{y}\left[T,B_{x,y}\right]\prod_{z}C_{x,z}\right)A_{k_1},\\
		&\sum_{\pmb j,\pmb k}\prod_{x=1}^{q}\tau_{j_x,k_x}\prod_{x=1}^{q-1}\delta_{k_{x+1},j_x}
		\cdot A_{j_q}^\dagger\prod_{x=1}^{q}\left(\prod_{y}B_{x,y}\prod_{z}\left[T,C_{x,z}\right]\right)A_{k_1},\\
		&\sum_{\pmb j,\pmb k}\prod_{x=1}^{q}\tau_{j_x,k_x}\prod_{x=1}^{q-1}\delta_{k_{x+1},j_x}
		\cdot A_{j_q}^\dagger\prod_{x=1}^{q}\left(\prod_{y}B_{x,y}\prod_{z}C_{x,z}\right)\left[T,A_{k_1}\right],
	\end{aligned}
	\end{equation}
	and
	\begin{equation}
	\begin{aligned}
	&\sum_{\pmb j,\pmb k}\prod_{x=1}^{q}\tau_{j_x,k_x}\prod_{x=1}^{q-1}\delta_{k_{x+1},j_x}
	\cdot \left[V,A_{j_q}^\dagger\right]\prod_{x=1}^{q}\left(\prod_{y}B_{x,y}\prod_{z}C_{x,z}\right)A_{k_1},\\
	&\sum_{\pmb j,\pmb k}\prod_{x=1}^{q}\tau_{j_x,k_x}\prod_{x=1}^{q-1}\delta_{k_{x+1},j_x}
	\cdot A_{j_q}^\dagger\prod_{x=1}^{q}\left(\prod_{y}\left[V,B_{x,y}\right]\prod_{z}C_{x,z}\right)A_{k_1},\\
	&\sum_{\pmb j,\pmb k}\prod_{x=1}^{q}\tau_{j_x,k_x}\prod_{x=1}^{q-1}\delta_{k_{x+1},j_x}
	\cdot A_{j_q}^\dagger\prod_{x=1}^{q}\left(\prod_{y}B_{x,y}\prod_{z}\left[V,C_{x,z}\right]\right)A_{k_1},\\
	&\sum_{\pmb j,\pmb k}\prod_{x=1}^{q}\tau_{j_x,k_x}\prod_{x=1}^{q-1}\delta_{k_{x+1},j_x}
	\cdot A_{j_q}^\dagger\prod_{x=1}^{q}\left(\prod_{y}B_{x,y}\prod_{z}C_{x,z}\right)\left[V,A_{k_1}\right].
	\end{aligned}
	\end{equation}
	In each case, we see from \eq{commutator_t} and \eq{commutator_v} that the result is again a fermionic chain. Specifically, commutators $\left[T,A_{j_q}^\dagger\right]$ and $\left[T,A_{k_1}\right]$ increase the ``length'' of the current fermionic chain from $q$ to $q+1$; commutators $\left[T,B_{x,y}\right]$ and $\left[T,C_{x,z}\right]$ either create a fermionic subchain or give the zero operator, or they can be computed recursively when $B_{x,y}$ and $C_{x,z}$ are fermionic subchains; commutators $\left[V,A_{j_q}^\dagger\right]$ and $\left[V,A_{k_1}\right]$ do not increase the length $q$ of the current fermionic chain, but they increase the number of $B_{x,y}$ and $C_{x,z}$ by one; commutators $\left[V,B_{x,y}\right]$ and $\left[V,C_{x,z}\right]$ either give the zero operator, or they can be computed recursively if $B_{x,y}$ and $C_{x,z}$ are fermionic subchains.
	
	Each application of commutation rules \eq{commutator_t} and \eq{commutator_v} increases the number of terms by a factor of at most $3$. The nested commutator $\left[H_{\gamma_{p+1}},\cdots\left[H_{\gamma_2},H_{\gamma_1}\right]\right]$ contains products of at most $2(p+1)$ elementary fermionic operators, giving at most $6^{p+1}(p+1)!$ terms in total. Meanwhile, the number of $\tau$ or $\nu$ increases by one depending on whether $H_{\gamma_{p+2}}=T$ or $H_{\gamma_{p+2}}=V$. The claim about the number preservation can be verified directly. This completes the inductive step.
\end{proof}

\begin{proposition}[Fermionic seminorm of fermionic chain and subchain]
	\label{prop:multilayer_seminorm}
	Let $H=T+V=\sum_{j,k}\tau_{j,k}A_j^\dagger A_k+\sum_{l,m}\nu_{l,m}N_lN_m$ be an interacting-electronic Hamiltonian \eq{second_quantized_ham}. Then, we can bound the fermionic seminorm of the fermionic chain $X$ in \eq{fermionic_chain} as
	\begin{equation}
	\label{eq:fermionic_chain_bound}
	\norm{X}_\eta\leq\norm{\tau}^q\eta
	\prod_{x=1}^{q}\left(\prod_{y}\max_{j_x}\norm{B_{x,y}}_{\eta-1}\prod_{z}\max_{k_x}\norm{C_{x,z}}_{\eta-1}\right),
	\end{equation}
	whereas fermionic subchains $B_{x,y}$, $C_{x,z}$ in \eq{fermionic_subchain} can be similarly bounded as
	\begin{equation}
	\label{eq:fermionic_subchain_bound}
	\begin{aligned}
		\norm{B_{x,y}}_\eta&\leq\norm{\tau}^{{q}'}\norm{\nu}_{\max}\eta
		\prod_{{x}'=1}^{{q}'}\left(\prod_{{y}'}\max_{{j}_{{x}'}'}\norm{B_{{x}',{y}'}}_{\eta-1}\prod_{{z}'}\max_{{k}_{{x}'}'}\norm{C_{{x}',{z}'}}_{\eta-1}\right),\\
		\norm{C_{x,y}}_\eta&\leq\norm{\tau}^{{q}''}\norm{\nu}_{\max}\eta
		\prod_{{x}''=1}^{{q}''}\left(\prod_{{y}''}\max_{{j}_{{x}''}''}\norm{B_{{x}'',{y}''}}_{\eta-1}\prod_{{z}''}\max_{{k}_{{x}''}''}\norm{C_{{x}'',{z}''}}_{\eta-1}\right).
	\end{aligned}
	\end{equation}
\end{proposition}
\begin{proof}
	We will prove this bound using \lem{cauchy}, \lem{diagonalization}, and \lem{fermionic_holder} in a similar way as in \prop{singlelayer_seminorm}. Specifically, we write $X=\sum_{j_q}A_{j_q}^\dagger D_{j_q}$, where
	\begin{equation}
	D_{j_q}=\sum_{\substack{j_1,\ldots,\\j_{q-1},\pmb k}}\prod_{x=1}^{q}\tau_{j_x,k_x}\prod_{x=1}^{q-1}\delta_{k_{x+1},j_x}
	\cdot \prod_{x=1}^{q}\left(\prod_{y}B_{x,y}\prod_{z}C_{x,z}\right)A_{k_1}.
	\end{equation}
	Then,
	\begin{equation}
	X^\dagger X
	=\sum_{j_{q_1},j_{q_2}}D_{j_{q_1}}^\dagger A_{j_{q_1}}A_{j_{q_2}}^\dagger D_{j_{q_2}}
	=\sum_{j_{q_1}}D_{j_{q_1}}^\dagger D_{j_{q_1}}
	-\sum_{j_{q_1},j_{q_2}}D_{j_{q_1}}^\dagger A_{j_{q_2}}^\dagger A_{j_{q_1}}D_{j_{q_2}}.
	\end{equation}
	Applying the operator Cauchy-Schwarz inequality (\lem{cauchy}), we obtain
	\begin{equation}
	X^\dagger X
	\leq\sum_{j_{q_1}}D_{j_{q_1}}^\dagger D_{j_{q_1}}
	+\sum_{j_{q_1},j_{q_2}}D_{j_{q_1}}^\dagger A_{j_{q_2}}^\dagger A_{j_{q_2}}D_{j_{q_1}}
	=\sum_{j_{q_1}}D_{j_{q_1}}^\dagger D_{j_{q_1}}N.
	\end{equation}
	Next, we write $D_{j_q}=\prod_{y}B_{q,y}E_{j_q}$, where
	\begin{equation}
	E_{j_q}=\sum_{\substack{j_1,\ldots,\\j_{q-1},\pmb k}}\prod_{x=1}^{q}\tau_{j_x,k_x}\prod_{x=1}^{q-1}\delta_{k_{x+1},j_x}
	\cdot \prod_{z}C_{q,z}\prod_{x=1}^{q-1}\left(\prod_{y}B_{x,y}\prod_{z}C_{x,z}\right)A_{k_1}.
	\end{equation}
	Invoking the H\"{o}lder-type inequality for the expectation value (\lem{fermionic_holder}), we get
	\begin{equation}
	\norm{X}_\eta
	=\sqrt{\norm{X^\dagger X}_\eta}
	\leq\eta^{1/2}\sqrt{\Big\Vert\sum_{j_{q_1}}D_{j_{q_1}}^\dagger D_{j_{q_1}}\Big\Vert_\eta}
	\leq\eta^{1/2}\sqrt{\Big\Vert\sum_{j_{q_1}}E_{j_{q_1}}^\dagger E_{j_{q_1}}\Big\Vert_\eta}
	\prod_{y}\max_{j_q}\norm{B_{q,y}}_{\eta-1}.
	\end{equation}
	
	We now write $E_{j_q}=\sum_{k_q}\tau_{j_q,k_q}F_{k_q}$, where
	\begin{equation}
	F_{k_q}=\sum_{\substack{j_1,\ldots,j_{q-1},\\k_1,\ldots,k_{q-1}}}\prod_{x=1}^{q-1}\tau_{j_x,k_x}\prod_{x=1}^{q-1}\delta_{k_{x+1},j_x}
	\cdot \prod_{z}C_{q,z}\prod_{x=1}^{q-1}\left(\prod_{y}B_{x,y}\prod_{z}C_{x,z}\right)A_{k_1}.
	\end{equation}
	Then,
	\begin{equation}
	\sum_{j_{q_1}}E_{j_{q_1}}^\dagger E_{j_{q_1}}
	=\sum_{k_{q_1},k_{q_2}}\left(\sum_{j_{q_1}}\bar{\tau}_{j_{q_1},k_{q_1}}\tau_{j_{q_1},k_{q_2}}\right)
	F_{k_{q_1}}^\dagger F_{k_{q_2}}.
	\end{equation}
	We perform diagonalization using \lem{diagonalization}, obtaining
	\begin{equation}
	\sum_{j_{q_1}}E_{j_{q_1}}^\dagger E_{j_{q_1}}
	\leq\norm{\tau^\dagger \tau}\sum_{k_{q_1}}F_{k_{q_1}}^\dagger F_{k_{q_1}}.
	\end{equation}
	Next, we write $F_{k_q}=\prod_{z}C_{q,z}G_{k_q}$, where
	\begin{equation}
	G_{k_q}=\sum_{\substack{j_1,\ldots,j_{q-1},\\k_1,\ldots,k_{q-1}}}\prod_{x=1}^{q-1}\tau_{j_x,k_x}\prod_{x=1}^{q-1}\delta_{k_{x+1},j_x}
	\cdot \prod_{x=1}^{q-1}\left(\prod_{y}B_{x,y}\prod_{z}C_{x,z}\right)A_{k_1}.
	\end{equation}
	Invoking again the H\"{o}lder-type inequality for the expectation value (\lem{fermionic_holder}), we get
	\begin{equation}
	\norm{X}_\eta
	\leq\norm{\tau}\eta^{1/2}\sqrt{\Big\Vert\sum_{k_{q_1}}G_{k_{q_1}}^\dagger G_{k_{q_1}}\Big\Vert_\eta}
	\prod_{y}\max_{j_q}\norm{B_{q,y}}_{\eta-1}\prod_{z}\max_{k_q}\norm{C_{q,z}}_{\eta-1}.
	\end{equation}
	
	Note that we can write $G_{k_q}=\sum_{j_{q-1}}\delta_{k_q,j_{q-1}}H_{j_{q-1}}$ with
	\begin{equation}
	H_{j_{q-1}}=\sum_{\substack{j_1,\ldots,j_{q-2},\\k_1,\ldots,k_{q-1}}}\prod_{x=1}^{q-1}\tau_{j_x,k_x}\prod_{x=1}^{q-2}\delta_{k_{x+1},j_x}
	\cdot \prod_{x=1}^{q-1}\left(\prod_{y}B_{x,y}\prod_{z}C_{x,z}\right)A_{k_1},
	\end{equation}
	which implies
	\begin{equation}
	\sum_{k_{q}}G_{k_{q}}^\dagger G_{k_{q}}
	=\sum_{j_{q-1}}H_{j_{q-1}}^\dagger H_{j_{q-1}}.
	\end{equation}
	We can now iterate this procedure $q$ times to get
	\begin{equation}
	\begin{aligned}
	\norm{X}_\eta
	&\leq\norm{\tau}^q\eta^{1/2}\sqrt{\Big\Vert\sum_{k_1}A_{k_1}^\dagger A_{k_1}\Big\Vert_{\eta}}
	\prod_{x=1}^{q}\left(\prod_{y}\max_{j_x}\norm{B_{x,y}}_{\eta-1}\prod_{z}\max_{k_x}\norm{C_{x,z}}_{\eta-1}\right)\\
	&=\norm{\tau}^q\eta
	\prod_{x=1}^{q}\left(\prod_{y}\max_{j_x}\norm{B_{x,y}}_{\eta-1}\prod_{z}\max_{k_x}\norm{C_{x,z}}_{\eta-1}\right).
	\end{aligned}
	\end{equation}
	This completes the proof of \eq{fermionic_chain_bound}.
	
	Essentially the same argument can be applied to bound the fermionic seminorm of fermionic subchains. The only difference is that we have additional coefficients $\beta_{x,{x}_0'}$ in $B_{x,y}$ and respectively $\chi_{x,{x}_0''}$ in $C_{x,z}$. But their indices will be contracted in the ${x}_0'$th and ${x}_0''$th iteraction of the above analysis and the coefficients can then be bounded by $\norm{\nu}_{\max}$, which completes the proof of \eq{fermionic_subchain_bound}.
\end{proof}

We now apply \prop{multilayer} to expand each nested commutator $\left[H_{\gamma_{p+1}},\cdots\left[H_{\gamma_2},H_{\gamma_1}\right]\right]$ into fermionic chains and use \prop{multilayer_seminorm} to bound their fermionic seminorm. The $\tau$ factors are already bounded by their spectral-norm $\norm{\tau}$ in \prop{multilayer_seminorm}. To proceed, we need to further bound each $\norm{B_{x,y}}_{\eta-1}$ and $\norm{C_{x,z}}_{\eta-1}$ separately. We have\footnotemark
\footnotetext{Note that one could also diagonalize the coefficient matrix $\nu^\dagger\nu$ and get a bound in terms of $\norm{\nu}$. However, such a bound will be loose for the electronic-structure simulation and will not be further considered here.}
\begin{equation}
\begin{aligned}
	\Big\Vert\sum_{l}\nu_{l,j_x}N_l\Big\Vert_{\eta-1}
	&=\sqrt{\Big\Vert\sum_{l_1,l_2}\bar{\nu}_{l_1,j_{x_1}}\nu_{l_2,j_{x_2}}N_{l_1}N_{l_2}\Big\Vert_{\eta-1}}
	\leq\norm{\nu}_{\max}\eta,\\
	\Big\Vert\sum_{m}\nu_{j_x,m}N_m\Big\Vert_{\eta-1}
	&=\sqrt{\Big\Vert\sum_{m_1,m_2}\bar{\nu}_{j_{x_1},m_1}\nu_{j_{x_2},m_2}N_{m_1}N_{m_2}\Big\Vert_{\eta-1}}
	\leq\norm{\nu}_{\max}\eta,\\
	\norm{\nu_{j_x,j_x}I}_{\eta-1}
	&\leq\norm{\nu}_{\max}
\end{aligned}
\end{equation}
for $B_{x,y}$ and similar estimates hold for $C_{x,z}$. In the case where $B_{x,y}$ or $C_{x,z}$ creates a fermionic subchain, we can estimate recursively using \prop{multilayer_seminorm}. In particular, we will introduce a factor of $\norm{\nu}_{\max}\eta$ each time a subchain is created. 

We know from \prop{multilayer} that the number of $\tau$ factors in each chain agrees with the number of $H_1=T$ in the nested commutator, whereas the number of $\nu$ factors coincides with the number of $H_0=V$. Since the number of fermionic chains is at most $6^pp!$, we obtain the bound
\begin{equation}
	\norm{\left[H_{\gamma_{p+1}},\cdots\left[H_{\gamma_2},H_{\gamma_1}\right]\right]}_\eta
	=\cO{\norm{\tau}^{\abs{\pmb \gamma}} (\eta \norm{\nu}_{\max})^{p+1-\abs{\pmb \gamma}}\eta}.
\end{equation}
Here, we have $1\leq\abs{\pmb \gamma}\leq p$ as $[T,T]=[V,V]=0$. This completes the proof of Eq.\ \eq{general_bound} of our main result \thm{fermionic_seminorm}.

\section{Path-counting bound on the expectation of fermionic operators}
\label{sec:path}

We now present the second strategy for bounding the expectation of fermionic operators, and apply it to estimate the fermionic seminorm of Trotter error.
Recall from \eq{pf2k_eta} that
\begin{equation}
\norm{\mathscr{S}_p(t)-e^{-itH}}_{\eta}=\cO{\max_{\pmb{\gamma}}\norm{\left[H_{\gamma_{p+1}},\cdots\left[H_{\gamma_2},H_{\gamma_1}\right]\right]}_{\eta}t^{p+1}},
\end{equation}
where $H_0=V=\sum_{l,m}\nu_{l,m}N_lN_m$, $H_1=T=\sum_{j,k}\tau_{j,k}A_j^\dagger A_k$ and $\gamma_{j} \in \{0,1\}$.
Hence to analyze the Trotter error, it suffices to bound the fermionic seminorm $\norm{\left[H_{\gamma_{p+1}},\cdots\left[H_{\gamma_2},H_{\gamma_1}\right]\right]}_{\eta}$.
We develop a general bound on this quantity in \sec{path_tech} based on a path-counting technique. We then use it to analyze the simulation of $d$-sparse interacting electrons in \sec{path-count-sparse}, proving Eq.\ \eq{sparse_bound} of our main result \thm{fermionic_seminorm}.

It is worth noting that our approach can also be adapted to establish \eq{general_path_bound}, a bound slightly weaker than our main result \eq{general_bound} but sufficient for our applications. See \append{pathcountdense} for details.

\subsection{Path-counting bound}
\label{sec:path_tech}
We start by bounding the transition amplitude between any two states in terms of the expectation value.
Since $\left[H_{\gamma_{p+1}},\cdots\left[H_{\gamma_2},H_{\gamma_1}\right]\right]$ is antihermitian, we have
\begin{equation}
\label{eq:trans2expect}
\norm{\left[H_{\gamma_{p+1}},\cdots\left[H_{\gamma_2},H_{\gamma_1}\right]\right]}_{\eta} = \max_{\ket{\psi_\eta}} \abs{\bra{\psi_\eta} \left[H_{\gamma_{p+1}},\cdots\left[H_{\gamma_2},H_{\gamma_1}\right]\right] \ket{\psi_\eta}}.
\end{equation}
We now aim to bound the expectation
\begin{equation}
\abs{\langle X\rangle} = \abs{\bra{\psi_\eta} \left[H_{\gamma_{p+1}},\cdots\left[H_{\gamma_2},H_{\gamma_1}\right]\right] \ket{\psi_\eta}}
\end{equation}
for any $\ket{\psi_\eta}$.
To this end, we first expand $\ket{\psi_\eta}$ as:
\begin{equation}
\ket{\psi_\eta} = \sum_{\pmb c \in \{0, 1\}^n, |\pmb c| = \eta} \alpha_{\pmb c} \ket{\pmb c},
\end{equation}
where $\pmb c$ is a configuration with $\eta$ electrons, and the number of ones in $\pmb c$ is given by the Hamming weight $\abs{\pmb c}=\sum_{j=0}^{n-1}c_j$.
Using the notation
\begin{equation}
\mu^0 = \nu,\quad \mu^1 = \tau,\quad H^0_{jk} = N_j N_k,\quad H^1_{jk} = A^\dagger_j A_k,
\end{equation}
we expand everything to get
\begin{align}
\abs{\langle X\rangle} &= \bigg|\sum_{j_{p+1}, k_{p+1}} \ldots \sum_{j_{1}, k_{1}} \sum_{\pmb{c}_1} \sum_{\pmb{c}_2} \bar{\alpha}_{\pmb{c}_1} \alpha_{\pmb{c}_2} \mu^{\gamma_{p+1}}_{j_{p+1} k_{p+1}} \ldots \mu^{\gamma_{1}}_{j_{1} k_{1}} \bra{\pmb{c}_1} \left[H^{\gamma_{p+1}}_{j_{p+1} k_{p+1}}, \ldots \left[H^{\gamma_{2}}_{j_{2} k_{2}}, H^{\gamma_{1}}_{j_{1} k_{1}}\right]\right] \ket{\pmb{c}_2}\bigg| \\
&\leq \norm{\tau}_{\max}^{\abs{\pmb \gamma}} \norm{\nu}_{\max}^{p+1-\abs{\pmb \gamma}} \sum_{\pmb{c}_1} \sum_{\pmb{c}_2} \left|\alpha_{\pmb{c}_1}\right| \left|\alpha_{\pmb{c}_2}\right| \times \nonumber\\
&\qquad \qquad\qquad \qquad\qquad \quad \sum_{\langle j_{p+1}, k_{p+1} \rangle} \ldots \sum_{\langle j_{1}, k_{1} \rangle}  \left|\bra{\pmb{c}_1} \left[H^{\gamma_{p+1}}_{j_{p+1} k_{p+1}}, \ldots \left[H^{\gamma_{2}}_{j_{2} k_{2}}, H^{\gamma_{1}}_{j_{1} k_{1}}\right]\right] \ket{\pmb{c}_2}\right|,
\end{align}
where $\pmb{c}_1, \pmb{c}_2$ are configurations with $\eta$ electrons, and $\langle j, k \rangle$ only sum over indices such that the corresponding $\mu_{j, k}^\gamma \neq 0$ (either $\tau$ or $\nu$ depending on $\gamma$).

Using the commutation relations in Equation~\eqref{eq:commutator_t}~and~\eqref{eq:commutator_v}, we know that the nested commutator $\left[H^{\gamma_{p+1}}_{j_{p+1} k_{p+1}}, \ldots \left[H^{\gamma_{2}}_{j_{2} k_{2}}, H^{\gamma_{1}}_{j_{1} k_{1}}\right]\right]$
can be written as a sum of
\begin{equation}
(-1)^a \ldots A^\dagger_j \ldots N_l \ldots A_k \ldots,
\end{equation}
for some $a \in \{0, 1\}$ and a sequence of elementary fermionic operators.
We call each term $P$ a \emph{fermionic path} and write $P \rhd \left(H^{\gamma_{p+1}}_{j_{p+1} k_{p+1}}, \ldots, H^{\gamma_{1}}_{j_{1} k_{1}} \right)$ to mean $P$ is one of the terms in the expansion of the nested commutator.
If the nested commutator evaluates to zero, then we consider the set
\begin{equation}
\left\{P \mbox{ such that } P \rhd \left(H^{\gamma_{p+1}}_{j_{p+1} k_{p+1}}, \ldots, H^{\gamma_{1}}_{j_{1} k_{1}} \right)\right\}
\end{equation}
to be an empty set.
One possible expansion of the nested commutator is presented in \prop{multilayer}.
This allows us to make a further expansion to yield
\begin{equation}
\label{eq:quadratic_form}
\abs{\langle X\rangle} \leq c_{\tau \nu} \sum_{\pmb{c}_1} \sum_{\pmb{c}_2} \left|\alpha_{\pmb{c}_1}\right| \left|\alpha_{\pmb{c}_2}\right| \sum_{\langle j_{p+1}, k_{p+1} \rangle} \ldots \sum_{\langle j_{1}, k_{1} \rangle} \sum_{P \rhd \left(H^{\gamma_{p+1}}_{j_{p+1} k_{p+1}}, \ldots, H^{\gamma_{1}}_{j_{1} k_{1}} \right)} \left| \bra{\pmb{c}_1} P \ket{\pmb{c}_2} \right|,
\end{equation}
where $c_{\tau \nu} = \norm{\tau}_{\max}^{\abs{\pmb \gamma}} \norm{\nu}_{\max}^{p+1-\abs{\pmb \gamma}}$.
We use the following proposition to characterize $\left| \bra{\pmb{c}_1} P \ket{\pmb{c}_2} \right|$.

\begin{lemma} \label{lem:fermionic_path}
For any computational basis state $\ket{\pmb{c}}$ where $\pmb{c}$ is a fermionic configuration, and fermionic path
\begin{equation}
P = (-1)^a \ldots A^\dagger_j \ldots N_l \ldots A_k \ldots,
\end{equation}
we have either $P \ket{\pmb{c}}$ is a computational basis state with some phase $\pm 1$ or $P \ket{\pmb{c}} = 0$.
\end{lemma}
\begin{proof}
The proof follows from a simple induction.
For the base case, we have $P = (-1)^a$ without any fermionic operator, so $P \ket{\pmb{c}}$ is a computational basis state with some phase $\pm 1$.
Now we consider the three cases: $P = N_l P'$, $P = A_k P'$, or $P = A^\dagger_j P'$.
By induction, we have $P' \ket{\pmb{c}}$ is a computational basis state $\ket{\pmb{c}'}$ with some phase $\pm 1$ or $P' \ket{\pmb{c}} = 0$. The latter is trivial. For the former case, we go through the following three cases.
\begin{itemize}
\item If $N_l$ is applied on $\ket{\pmb{c}'}$, we check if site-$l$ has an electron in configuration $\pmb{c}'$. If site-$l$ has an electron, then $N_l \ket{\pmb{c}'} = \ket{\pmb{c}'}$; otherwise, $N_l \ket{\pmb{c}'} = 0$.
\item If $A_k$ is applied on $\ket{\pmb{c}'}$, we check if site-$k$ has an electron in configuration $\pmb{c}'$. If site-$k$ has an electron, then $A_k \ket{\pmb{c}'}$ will remove the site-$k$ electron and add some phase according to the rule in Equation~\eqref{eq:annihirule}; otherwise, $A_k \ket{\pmb{c}'} = 0$.
\item If $A^\dagger_j$ is applied on $\ket{\pmb{c}'}$, we check if site-$j$ has an electron in configuration $\pmb{c}'$. If site-$j$ does not have an electron, then $A^\dagger_j \ket{\pmb{c}'}$ will create an electron at site-$j$ and add some phase according to the rule in Equation~\eqref{eq:createrule}; otherwise, $A^\dagger_j \ket{\pmb{c}'} = 0$.
\end{itemize}
Therefore, $P\ket{\pmb{c}}$ is either a computational basis state with some phase $\pm 1$ or $P \ket{\pmb{c}} = 0$.
\end{proof}

\begin{corollary}\label{cor:sumc1-c1Pc2}
We have that $\left| \bra{\pmb{c}_1} P \ket{\pmb{c}_2} \right|$ is either $0$ or $1$. Furthermore, $\sum_{\pmb{c}_1 \in \mathcal{S}} \left| \bra{\pmb{c}_1} P \ket{\pmb{c}_2} \right| \leq \norm{P \ket{\pmb{c}_2}}$ for any set $\mathcal{S}$ of configurations.
\end{corollary}

Next, we define a graph $\mathcal{G} = (\mathcal{V}, \mathcal{E})$ where the vertices $\mathcal{V}$ are the second-quantized configurations with $\eta$ electrons, and the weighted adjacency matrix for the edges $\mathcal{E}$ is defined as
\begin{equation}
w_{\pmb{c}_1, \pmb{c}_2} = \sum_{\langle j_{p+1}, k_{p+1} \rangle} \ldots \sum_{\langle j_{1}, k_{1} \rangle} \sum_{P \rhd \left(H^{\gamma_{p+1}}_{j_{p+1} k_{p+1}}, \ldots, H^{\gamma_{1}}_{j_{1} k_{1}} \right)} \left| \bra{\pmb{c}_1} P \ket{\pmb{c}_2} \right|.
\end{equation}
The weight $w_{\pmb{c}_1, \pmb{c}_2}$ counts the number of fermionic paths that can take $\ket{\pmb{c}_2}$ to $\ket{\pmb{c}_1}$.
Note that this graph may contain self-loops (equivalent to $w_{\pmb{c}_1, \pmb{c}_1} > 0$) as there are fermionic paths that leave $\ket{\pmb{c}_1}$ unchanged or simply add a phase of $-1$. We now define the \emph{degree} of $\pmb{c}_2$ as
\begin{align}
\label{eq:sym_degree}
\mathrm{deg}(\pmb{c}_2) &= \sum_{\pmb{c}_1} \frac{w_{\pmb{c}_1, \pmb{c}_2} + w_{\pmb{c}_2, \pmb{c}_1}}{2}\\
& = \sum_{\pmb{c}_1} \sum_{\langle j_{p+1}, k_{p+1} \rangle} \ldots \sum_{\langle j_{1}, k_{1} \rangle} \sum_{P \rhd \left(H^{\gamma_{p+1}}_{j_{p+1} k_{p+1}}, \ldots, H^{\gamma_{1}}_{j_{1} k_{1}} \right)} \frac{\left| \bra{\pmb{c}_1} P \ket{\pmb{c}_2} \right| + \left| \bra{\pmb{c}_2} P \ket{\pmb{c}_1} \right| }{2} \\
& = \sum_{\pmb{c}_1} \sum_{\langle j_{p+1}, k_{p+1} \rangle} \ldots \sum_{\langle j_{1}, k_{1} \rangle} \sum_{P \rhd \left(H^{\gamma_{p+1}}_{j_{p+1} k_{p+1}}, \ldots, H^{\gamma_{1}}_{j_{1} k_{1}} \right)} \frac{\left| \bra{\pmb{c}_1} P \ket{\pmb{c}_2} \right| + \left| \bra{\pmb{c}_1} P^\dagger \ket{\pmb{c}_2} \right| }{2} \\
&\leq \sum_{\langle j_{p+1}, k_{p+1} \rangle} \ldots \sum_{\langle j_{1}, k_{1} \rangle} \sum_{P \rhd \left(H^{\gamma_{p+1}}_{j_{p+1} k_{p+1}}, \ldots, H^{\gamma_{1}}_{j_{1} k_{1}} \right)} \frac{\norm{P \ket{\pmb{c}_2}} + \norm{P^\dagger \ket{\pmb{c}_2}}}{2}, \label{eq:uppdegC2}
\end{align}
which is equivalent to counting the number of fermionic paths that evaluate nonzero on the initial state $\ket{\pmb{c}_2}$.
The last inequality follows from \cor{sumc1-c1Pc2}.
We now introduce the following lemma which relates the maximum degree and the quadratic form \eq{quadratic_form} we wish to bound.

\begin{lemma} \label{lem:ev-and-maxcol}
For any real symmetric matrix $w \in \mathbb{R}^{k\times k}$ with nonnegative entries and normalized real vectors $v \in \mathbb{R}^k$ with $\norm{v} = 1$, we have\footnotemark
\begin{equation}
\sum_{i,j} w_{i,j} v_i v_j \leq \max_i \sum_j w_{i,j},
\end{equation}
\end{lemma}
\footnotetext{Note that this can alternatively be proved by analyzing the Ger\v{s}gorin discs; see for example \cite[Corollary 6.1.5]{horn2012matrix}. Had we not taken the expectation in \eq{trans2expect} and symmetrized the weighted adjacency matrix in \eq{sym_degree}, we would have needed a stronger bound here \cite[5.6.P21]{horn2012matrix}.}
\begin{proof}
Let $u_1$ be an eigenvector corresponding to the largest eigenvalue $\lambda_1$ of $w$ with $\norm{u_1}=1$.
By the Rayleigh quotient theorem \cite[Theorem 4.2.2]{horn2012matrix},
\begin{equation}
	v^T w v \leq u_1^T w u_1=\lambda_1
\end{equation}
for any $v \in \mathbb{R}^k$ with $\norm{v} = 1$, where $v^T$ denotes the vector transpose of $v$. Consider $i^* = \argmax_{j} (u_1)_j$.
We assume $(u_1)_{i^*} > 0$ without loss of generality, for otherwise we multiply $u_1$ by $-1$.
Then, we have
\begin{equation}
u_1^T w u_1 = \lambda_1 = \frac{(w u_1)_{i^*}}{(u_1)_{i^*}} = \frac{\sum_j w_{i^*, j} (u_1)_j}{(u_1)_{i^*}} \leq \sum_j w_{i^*, j} \leq \max_i \sum_j w_{i,j}.
\end{equation}
This concludes the proof.
\end{proof}

Using Lemma~\ref{lem:ev-and-maxcol}, we obtain an upper bound of $\abs{\langle X\rangle}$
\begin{align}
\abs{\langle X\rangle} &\leq c_{\tau \nu} \sum_{\pmb{c}_1} \sum_{\pmb{c}_2} \left|\alpha_{\pmb{c}_1}\right| \left|\alpha_{\pmb{c}_2}\right| \sum_{\langle j_{p+1}, k_{p+1} \rangle} \ldots \sum_{\langle j_{1}, k_{1} \rangle} \sum_{P \rhd \left(H^{\gamma_{p+1}}_{j_{p+1} k_{p+1}}, \ldots, H^{\gamma_{1}}_{j_{1} k_{1}} \right)} \left| \bra{\pmb{c}_1} P \ket{\pmb{c}_2} \right|\\
&= c_{\tau \nu} \sum_{\pmb{c}_1} \sum_{\pmb{c}_2} w_{\pmb{c}_1, \pmb{c}_2} \left|\alpha_{\pmb{c}_1}\right| \left|\alpha_{\pmb{c}_2}\right|
= c_{\tau \nu} \sum_{\pmb{c}_1} \sum_{\pmb{c}_2} \frac{w_{\pmb{c}_1, \pmb{c}_2}+w_{\pmb{c}_2, \pmb{c}_1}}{2} \left|\alpha_{\pmb{c}_1}\right| \left|\alpha_{\pmb{c}_2}\right| \\
&\leq c_{\tau \nu} \max_{\pmb{c}_2} \sum_{\pmb{c}_1} \frac{w_{\pmb{c}_1, \pmb{c}_2}+w_{\pmb{c}_2, \pmb{c}_1}}{2} 
= c_{\tau \nu} \max_{\pmb{c}} \mathrm{deg}\left(\pmb{c}\right)
\end{align}
in terms of the maximum degree of the graph $\mathcal{G}$. Finally, we arrive at the following proposition by combining the above bound with Equation~\eqref{eq:uppdegC2}.

\begin{proposition}[Path-counting bound of fermionic seminorm]
	\label{prop:path_bound_seminorm}
	Let $H=T+V=\sum_{j,k}\tau_{j,k}A_j^\dagger A_k+\sum_{l,m}\nu_{l,m}N_lN_m$ be an interacting-electronic Hamiltonian as in \eq{second_quantized_ham}. Then, each nested commutator $\left[H_{\gamma_{p+1}},\cdots\left[H_{\gamma_2},H_{\gamma_1}\right]\right]$, where $H_0=V$ and $H_1=T$, can be bounded as
	\begin{equation} \label{eq:pathcountbd}
	\begin{aligned}
	\norm{\left[H_{\gamma_{p+1}},\cdots\left[H_{\gamma_2},H_{\gamma_1}\right]\right]}_{\eta}
	\leq \norm{\tau}_{\max}^{\abs{\pmb \gamma}} \norm{\nu}_{\max}^{p+1-\abs{\pmb \gamma}} \max_{\pmb{c}_\eta} \mathrm{deg}\left(\pmb{c}_\eta\right),
	\end{aligned}
	\end{equation}
	where $\abs{\pmb \gamma}=\sum_{q=0}^{p+1}\gamma_q$ and $\pmb{c}_\eta$ is a fermionic configuration with $\eta$ electrons.
	Here, the degree of configuration $\pmb c_\eta$ is defined as
	\begin{equation}
	\mathrm{deg}\left(\pmb{c}_\eta\right) = \sum_{\langle j_{p+1}, k_{p+1} \rangle} \ldots \sum_{\langle j_{1}, k_{1} \rangle} \sum_{P \rhd \left(H^{\gamma_{p+1}}_{j_{p+1} k_{p+1}}, \ldots, H^{\gamma_{1}}_{j_{1} k_{1}} \right)} \frac{1}{2}\left(\norm{P \ket{\pmb{c}_\eta}} + \Vert P^\dagger \ket{\pmb{c}_\eta} \Vert\right),
	\end{equation}
	where for $q = 1, \ldots, p+1$, $\langle j_q, k_q \rangle$ sum over indices $j_q, k_q$ such that $\mu_{j_q,k_q}^{\nu_q}\neq 0$, fermionic path $P \rhd \left(H^{\gamma_{p+1}}_{j_{p+1} k_{p+1}}, \ldots, H^{\gamma_{1}}_{j_{1} k_{1}} \right)$ goes over all the terms $P=(-1)^a \ldots A^\dagger_j \ldots N_l \ldots A_k \ldots$ in the expansion of $\left[H^{\gamma_{p+1}}_{j_{p+1} k_{p+1}}, \ldots \left[H^{\gamma_{2}}_{j_{2} k_{2}}, H^{\gamma_{1}}_{j_{1} k_{1}}\right]\right]$, and $\mu^0 = \nu,\ \mu^1 = \tau,\ H^0_{jk} = N_j N_k,\ H^1_{jk} = A^\dagger_j A_k$.
\end{proposition}

\subsection{Counting fermionic paths for \texorpdfstring{$d$}{d}-sparse interactions}
\label{sec:path-count-sparse}
As an illustrative example, let us consider an upper bound of $\max_{\pmb{c}} \mathrm{deg}(\pmb{c})$ for electronic Hamiltonians with $d$-sparse interactions.
We will use the commutation relations in Equation~\eqref{eq:commutator_t}~and~\eqref{eq:commutator_v}, restated below
\begin{equation} \label{eq:commutator_t_2}
	\left[A_j^\dagger A_k,A_{j_x}^\dagger\right]=\delta_{k,j_x}A_j^\dagger,\quad
	\left[A_j^\dagger A_k,A_{k_y}\right]=-\delta_{k_y,j}A_k,\quad
	\left[A_j^\dagger A_k,N_{l_z}\right]
	=\delta_{k,l_z}A_j^\dagger A_k
	-\delta_{j,l_z}A_j^\dagger A_k.
\end{equation}
\begin{equation} \label{eq:commutator_v_2}
	\left[N_lN_m,A_{j_x}^\dagger\right]
	=\delta_{m,j_x}N_lA_{j_x}^\dagger+\delta_{l,j_x}A_{j_x}^\dagger N_m, \quad
	\left[N_lN_m,A_{k_x}\right]
	=-\delta_{m,k_x}N_lA_{k_x}-\delta_{l,k_x}A_{k_x} N_m.
\end{equation}
We start with an intuitive argument.
For every $q = 2, \ldots, p+1$, we have
\begin{equation}
\left[H^{\gamma_{q}}_{j_{q} k_{q}}, \ldots \left[H^{\gamma_{2}}_{j_{2} k_{2}}, H^{\gamma_{1}}_{j_{1} k_{1}}\right]\right] = \sum_{P \rhd \left(H^{\gamma_{q-1}}_{j_{q-1} k_{q-1}}, \ldots, H^{\gamma_{1}}_{j_{1} k_{1}}\right)} \left[H^{\gamma_{q}}_{j_{q} k_{q}}, P\right],
\end{equation}
where $P$ only contains fermionic operator acting on sites $j_1, k_1, \ldots, j_{q-1}, k_{q-1}$. From the commutation relations \eq{commutator_t_2} and \eq{commutator_v_2}, we see that at least one of $j_q, k_q$ must match one of the indices $j_1, k_1, \ldots, j_{q-1}, k_{q-1}$.
Furthermore, for every $j_q$, there are at most $d$ $k_q$'s that have non-zero coefficient in $\tau_{j_q, k_q}$ (for $\gamma_q = 1$) or $\nu_{j_q, k_q}$ (for $\gamma_q = 0$).
Hence, we have the following bound
\begin{equation}
\sum_{\langle j_{p+1}, k_{p+1} \rangle} \ldots \sum_{\langle j_{1}, k_{1} \rangle} \sum_{P \rhd \left(H^{\gamma_{p+1}}_{j_{p+1} k_{p+1}}, \ldots, H^{\gamma_{1}}_{j_{1} k_{1}} \right)} 1 = \mathcal{O}(n d^{p+1}).
\end{equation}
The $n$ factor follows from the fact that only one index can be freely choosen between $0, \ldots, n-1$. And for any pair of indices $\langle j_q, k_q \rangle$, one of them has to match the previous indices, while the other one can only choose from the $d$ indices under the sparsity constraint.
Hence we have the $d^p$ factor in the asymptotic bound.

However, this analysis can be further improved using certain properties of $P$.
Specifically, we will show that the rightmost fermionic operator in $P$ can be either an annihilation operator $A$ or an occupation-number operator $N$.
This means that, for $\norm{P \ket{\pmb{c}_\eta}}$ to be nonzero, the rightmost fermionic operator of $P$ must act on the $\eta$ occupied sites in the configuration $\pmb{c}_\eta$.
Therefore, we have the bound
\begin{equation}
\sum_{\langle j_{p+1}, k_{p+1} \rangle} \ldots \sum_{\langle j_{1}, k_{1} \rangle} \sum_{P \rhd \left(H^{\gamma_{p+1}}_{j_{p+1} k_{p+1}}, \ldots, H^{\gamma_{1}}_{j_{1} k_{1}} \right)} \norm{P \ket{\pmb{c}_\eta}} = \mathcal{O}(\eta d^{p+1}).
\end{equation}
Similarly, we have
\begin{equation}
\sum_{\langle j_{p+1}, k_{p+1} \rangle} \ldots \sum_{\langle j_{1}, k_{1} \rangle} \sum_{P^\dagger \rhd \left(H^{\gamma_{p+1}}_{j_{p+1} k_{p+1}}, \ldots, H^{\gamma_{1}}_{j_{1} k_{1}} \right)} \Vert P^\dagger \ket{\pmb{c}_\eta} \Vert = \mathcal{O}(\eta d^{p+1}).
\end{equation}
Combining with the Trotter error bound \eq{pf2k_eta} and the path-counting bound \eq{pathcountbd}, we obtain
\begin{equation}
\norm{\mathscr{S}_p(t)-e^{-itH}}_{\eta}
=\cO{\norm{\tau}_{\max}^{\abs{\pmb \gamma}} \norm{\nu}_{\max}^{p+1-\abs{\pmb \gamma}} d^{p+1}\eta t^{p+1}}.
\end{equation}
Finally, since $1 \leq \abs{\pmb \gamma}=\sum_{q=1}^{p+1} \gamma_q \leq p$, we have
\begin{equation}
\norm{\tau}_{\max}^{\abs{\pmb \gamma}} \norm{\nu}_{\max}^{p+1-\abs{\pmb \gamma}} \leq (\norm{\tau}_{\max} + \norm{\nu}_{\max})^p \norm{\tau}_{\max} \norm{\nu}_{\max}.
\end{equation}
This sketches the proof of the scaling in Eq.\ \eq{sparse_bound} of \thm{fermionic_seminorm}.
A rigorous proof using induction is given in \prop{cntpathsparse}.

\begin{proposition}[Sparse path-counting bound]\label{prop:cntpathsparse} 
Under the same assumption as in \prop{path_bound_seminorm},
if each column and row of coefficient matrices $\tau,\nu$ has at most $d$ nonzero elements,
\begin{align}
\sum_{\langle j_{p+1}, k_{p+1} \rangle} \ldots \sum_{\langle j_{1}, k_{1} \rangle} \sum_{P \rhd \left(H^{\gamma_{p+1}}_{j_{p+1} k_{p+1}}, \ldots, H^{\gamma_{1}}_{j_{1} k_{1}} \right)} \norm{P \ket{\pmb{c}_\eta}} &= \mathcal{O}(\eta d^{p+1}),\\
\sum_{\langle j_{p+1}, k_{p+1} \rangle} \ldots \sum_{\langle j_{1}, k_{1} \rangle} \sum_{P^\dagger \rhd \left(H^{\gamma_{p+1}}_{j_{p+1} k_{p+1}}, \ldots, H^{\gamma_{1}}_{j_{1} k_{1}} \right)} \norm{P^\dagger \ket{\pmb{c}_\eta}} &= \mathcal{O}(\eta d^{p+1}).
\end{align}
\end{proposition}
\begin{proof}
We will prove the following claims by induction on $q = 2, \ldots, p+1$.
\begin{itemize}
\item All fermionic paths $P$ start with either $N$ or $A$, but not $A^\dagger$ (we refer to the rightmost operator as the starting point).
\item All fermionic paths $P$ have at most $q+1$ elementary fermionic operators.
\item The number of fermionic paths $P$ that start with a fermionic operator acting on a specific site $i$ is at most $(2d)^qq!/2$.
\end{itemize}
The base case $q=2$ can be easily verified by noting that we only need to consider $[T, V]$ or $[V, T]$.
Using the commutation relations given in Eq.~\eq{commutator_t_2}~and~\eq{commutator_v_2},
we can see in both cases that the fermionic paths all start with either $N$ or $A$.
For every site $i$, there are at most $4 d^2$ fermionic paths starting with site $i$.
Furthermore, every fermionic path consists of $3$ elementary fermionic operators.
These results establish all the bullet points for the base case of $q=2$.

For every $q>2$, we use the induction hypothesis for $q-1$ to prove the desired result.
If $\gamma_q = 1$, then we will take another commutator with $T = \sum_{j_q, k_q} \tau_{j_q, k_q} A^\dagger_{j_q} A_{k_q}$.
We can see that all fermionic paths $P \rhd \left(H^{\gamma_{q}}_{j_{q} k_{q}}, \ldots, H^{\gamma_{1}}_{j_{1} k_{1}} \right)$ come from the expansion of
\begin{equation}
[A^\dagger_{j_q} A_{k_q}, P'],\ \forall \langle j_q, k_q \rangle,\ \forall P' \rhd \left(H^{\gamma_{q-1}}_{j_{q-1} k_{q-1}}, \ldots, H^{\gamma_{1}}_{j_{1} k_{1}} \right).
\end{equation}
Using the commutation rule $[X, Y_1 \ldots Y_\kappa] = \sum_{k=1}^\kappa Y_1\ldots Y_{k-1} [X, Y_k] Y_{k+1} \ldots Y_{\kappa}$,
we can show that all the claims hold for $q$ as follows.
First, if all the fermionic paths $P' \rhd \left(H^{\gamma_{q-1}}_{j_{q-1} k_{q-1}}, \ldots, H^{\gamma_{1}}_{j_{1} k_{1}} \right)$ start with either $N$ or $A$,
then all the paths $P \rhd \left(H^{\gamma_{q}}_{j_{q} k_{q}}, \ldots, H^{\gamma_{1}}_{j_{1} k_{1}} \right)$ will start with either $N$ or $A$.
This follows from the commutation relations: $[A^\dagger_j A_k, A_{k_q}] = -\delta_{k_y, j} A_k$ and $\left[A_j^\dagger A_k,N_{l_z}\right]
=\delta_{k,l_z}A_j^\dagger A_k
-\delta_{j,l_z}A_j^\dagger A_k$.
Furthermore, since $P'$ has at most $(q-1)+1=q$ elementary fermionic operators, the expansion of $[A^\dagger_{j_q} A_{k_q}, P']$ will have at most $q+1$ elementary fermionic operators.

We now prove an upper bound for the number of fermionic paths starting from a specific site $i$.
If we take the commutator of $A^\dagger_{j_q} A_{k_q}$ with a fermionic operator that is not the starting operator in $P$, then the starting operator is not affected.
Because of the sparsity constraint and the $\delta$ function created by the commutation relation in Eq.\ \eq{commutator_t_2}, we have created at most $2d(q-1) \times$ more fermionic paths starting with site $i$.
Now if we take the commutator of $A^\dagger_{j_q} A_{k_q}$ with the starting operator $A_{k_y}$ (for some index $k_y$) in the fermionic path $P'$, then the starting operator becomes $A_{k_q}$ and we have an additional $\delta_{k_y, j_q}$.
In this case, $k_q$ can start from any site, but there will be at most $d$ choices of $j_q$, hence $d$ choices of $k_y$.
This means we have created at most $2d \times$ more fermionic paths starting with each site.
The case where $N_{l_z}$ is the starting operator can be analyzed in a similar way.
Together, we have created at most $2dq \times$ more fermionic paths starting with each site.
This leads to an upper bound of
\begin{equation}
2dq(2d)^{q-1}(q-1)!/2 = (2d)^qq!/2
\end{equation}
fermionic paths for each fixed starting site.
The analysis for $\gamma_q = 0$ proceeds in a similar way using Eq.\ \eq{commutator_v_2}.
This completes the inductive step for $q$.

Performing the induction over $q$ from $2$ to $p+1$ shows that
the number of fermionic paths starting with site $i$ is at most
\begin{equation}
(2d)^{p+1}(p+1)!/2 = \mathcal{O}(d^{p+1}).
\end{equation}
Because $P$ starts with either $A$ or $N$,
$\norm{P \ket{\pmb{c}_\eta}}$ would be nonzero only if the starting fermionic operator acts on one of the $\eta$ occupied sites in the configuration $\pmb{c}_\eta$.
Hence there are at most $\eta \mathcal{O}(d^{p+1})$ fermionic paths with non-zero $\norm{P \ket{\pmb{c}_\eta}}$.
Finally, recall from \lem{fermionic_path} that $\norm{P \ket{\pmb{c}_\eta}}$ is either $0$ or $1$. Therefore,
\begin{equation}
\sum_{\langle j_{p+1}, k_{p+1} \rangle} \ldots \sum_{\langle j_{1}, k_{1} \rangle} \sum_{P \rhd \left(H^{\gamma_{p+1}}_{j_{p+1} k_{p+1}}, \ldots, H^{\gamma_{1}}_{j_{1} k_{1}} \right)} \norm{P \ket{\pmb{c}_\eta}} = \mathcal{O}(\eta d^{p+1}).
\end{equation}

A similar argument can be used to prove the other claimed bound
\begin{equation}
\sum_{\langle j_{p+1}, k_{p+1} \rangle} \ldots \sum_{\langle j_{1}, k_{1} \rangle} \sum_{P^\dagger \rhd \left(H^{\gamma_{p+1}}_{j_{p+1} k_{p+1}}, \ldots, H^{\gamma_{1}}_{j_{1} k_{1}} \right)} \norm{P^\dagger \ket{\pmb{c}_\eta}} = \mathcal{O}(\eta d^{p+1}).
\end{equation}
The argument uses the property that all fermionic paths $P$ end with either $N$ or $A^\dagger$ but not $A$, which again follows from the commutation relations \eq{commutator_t_2} and \eq{commutator_v_2}. The proof is now completed.
\end{proof}

It is worth mentioning that the path-counting approach can also be used to analyze the simulation of non-sparse electronic Hamiltonians. The resulting bound, as given by \eq{general_path_bound}, is slightly weaker than Eq.\ \eq{general_bound} of \thm{fermionic_seminorm}, but suffices for our applications to be discussed in \sec{app_plane_wave}. See \append{pathcountdense} for details and proofs.

\section{Tightness}
\label{sec:tight}

We have already established multiple bounds in \thm{fermionic_seminorm} on the fermionic seminorm of the Trotter error. However, a common issue with the Trotterization algorithm is that its error estimate can be very loose for simulating specific systems. Here, we prove \thm{fermionic_seminorm_tightness} that demonstrates the tightness of our analysis for the interacting-electronic Hamiltonian \eq{second_quantized_ham}.

Specifically, we construct concrete examples of interacting-electronic Hamiltonian $H=T+V$ and lower-bound the fermionic seminorm of nested commutators: $\norm{[T,\ldots[T,V]]}_{\eta}$ in \sec{tightness_t} and $\norm{[V,\ldots[V,T]]}_{\eta}$ in \sec{tightness_v}. We show that the results almost match the upper bounds in \thm{fermionic_seminorm}. Since Trotter error depends on these nested commutators, this shows that our result is nearly tight modulo an application of the triangle inequality.

\subsection{Lower-bounding \texorpdfstring{$\norm{[T,\ldots[T,V]]}_{\eta}$}{nested commutators with T}}
\label{sec:tightness_t}
We construct the electronic Hamiltonian $H=T+V$, where
\begin{equation}
\label{eq:def_tv}
T=\sum_{j,k=0}^{n-1}A_j^\dagger A_k,\qquad
V=\sum_{x,y=0}^{\frac{n}{2}-1}N_{x}N_{y}.
\end{equation}
Note that we may without loss of generality assume that $n$ is even, for otherwise we restrict to the first $n-1$ spin orbitals.
Comparing with the definition of the interacting-electronic model \eq{second_quantized_ham}, we see that the coefficient matrix $\tau$ is an all-ones matrix with spectral norm $\norm{\tau}=n$, whereas $\nu$ contains an all-ones submatrix on the top left corner with max-norm $\norm{\nu}_{\max}=1$. Our goal is to lower-bound the fermionic seminorm $\norm{[T,\ldots[T,V]]}_{\eta}$.

Due to the complicated commutation relations between $T$ and $V$, a direct computation of $[T,\ldots[T,V]]$ seems technically challenging. Instead, we perform a change of basis by applying the fermionic Fourier transform
\begin{equation}
\label{eq:ffft}
\mathrm{FFFT}^\dagger\cdot A_j^\dagger\cdot \mathrm{FFFT}=\frac{1}{\sqrt{n}}\sum_{l=0}^{n-1} e^{-\frac{2\pi i jl}{n}}A_l^\dagger,\qquad
\mathrm{FFFT}^\dagger\cdot A_k\cdot \mathrm{FFFT}=\frac{1}{\sqrt{n}}\sum_{m=0}^{n-1} e^{\frac{2\pi i km}{n}}A_m.
\end{equation}
This gives
\begin{equation}
\label{eq:def_tv_tilde}
\begin{aligned}
	\widetilde{T}&=\mathrm{FFFT}^\dagger\cdot T\cdot\mathrm{FFFT}=nN_0,\\
	\widetilde{V}&=\mathrm{FFFT}^\dagger\cdot V\cdot\mathrm{FFFT}
	=\frac{1}{n^2}\sum_{j,k,l,m}
	\left(\sum_{x=0}^{\frac{n}{2}-1}e^{\frac{2\pi i x(k-j)}{n}}\right)
	\left(\sum_{y=0}^{\frac{n}{2}-1}e^{\frac{2\pi i y(m-l)}{n}}\right)
	A_j^\dagger A_k A_l^\dagger A_m.
\end{aligned}
\end{equation}
We also define the $\eta$-electron states for $\eta\leq\frac{n}{2}$:
\begin{equation}
\label{eq:def_state_tilde}
\begin{aligned}
	\ket{\widetilde{\psi}_\eta}=\frac{\ket{010\cdots 0\overbrace{1\cdots 1}^{\eta-1}}+\ket{100\cdots 0\overbrace{1\cdots 1}^{\eta-1}}}{\sqrt{2}},\quad
	\ket{\widetilde{\phi}_\eta}=\frac{\ket{010\cdots 0\overbrace{1\cdots 1}^{\eta-1}}+i\ket{100\cdots 0\overbrace{1\cdots 1}^{\eta-1}}}{\sqrt{2}}.
\end{aligned}
\end{equation}
The following proposition shows that the above choice of operators and states almost saturates the fermionic seminorm of nested commutators.
\begin{proposition}
	\label{prop:tightness_t}
	Define $\widetilde{T}$, $\widetilde{V}$ as in \eq{def_tv_tilde} and $\ket{\widetilde{\psi}_\eta}$, $\ket{\widetilde{\phi}_\eta}$ as in \eq{def_state_tilde}. Then,
	\begin{equation}
		\begin{rcases*}
		\big|\bra{\widetilde{\psi}_\eta}\overbrace{\big[\widetilde{T},\ldots\big[\widetilde{T}}^{p},\widetilde{V}\big]\big]\ket{\widetilde{\psi}_\eta}\big|, & p\text{ odd}\\
		\big|\bra{\widetilde{\phi}_\eta}\underbrace{\big[\widetilde{T},\ldots\big[\widetilde{T}}_{p},\widetilde{V}\big]\big]\ket{\widetilde{\phi}_\eta}\big|, & p\text{ even}
		\end{rcases*}
		=\frac{n^p\eta}{\pi}+\cO{n^p}.
	\end{equation}
\end{proposition}

A proof of this proposition is given in \append{commutator_t}. By rescaling the Hamiltonian constructed in \eq{def_tv_tilde}, we can demonstrate the tightness of our bound as follows. For any $s,w>0$, we define the rescaled Hamiltonian
\begin{equation}
\label{eq:def_tv_rescale}
T=\frac{s}{n}\sum_{j,k=0}^{n-1}A_j^\dagger A_k,\qquad
V=w\sum_{x,y=0}^{\frac{n}{2}-1}N_{x}N_{y}.
\end{equation}
Comparing with the definition of the interacting-electronic model \eq{second_quantized_ham}, we see that $\norm{\tau}=s$ and $\norm{\nu}_{\max}=w$. The above proposition then shows that
\begin{equation}
	\Big\Vert\underbrace{\big[{T},\ldots\big[{T}}_{p},{V}\big]\big]\Big\Vert_\eta
	=\Big\Vert\underbrace{\big[\widetilde{T},\ldots\big[\widetilde{T}}_{p},\widetilde{V}\big]\big]\Big\Vert_\eta
	=\Om{s^{p}w\eta},
\end{equation}
where we have used the unitary invariance of the fermionic seminorm in the first equality. This establishes the first claimed bound in \eq{tightness_general} of \thm{fermionic_seminorm_tightness}.

Note that a similar example can be constructed to demonstrate the tightness of our bound for simulating sparse electronic Hamiltonians. Specifically, suppose we have $u,w>0$ and positive integer $2\leq d\leq \eta\leq\frac{n}{2}$.\footnotemark\ We may assume without loss of generality that $d$ is even, for otherwise we use $d-1$. We then define
\footnotetext{The special case $d=1$ can be handled separately by choosing $T=uA_0^\dagger A_1+uA_1^\dagger A_0$ and $V=wN_0$.}
\begin{equation}
\label{eq:def_tv_sparse}
	T=u\sum_{j,k=0}^{d-1}A_j^\dagger A_k,\qquad
	V=w\sum_{x,y=0}^{\frac{d}{2}-1}N_{x}N_{y}.
\end{equation}
Comparing with the definition of the interacting-electronic model \eq{second_quantized_ham}, we see that $\norm{\tau}_{\max}=u$ and $\norm{\nu}_{\max}=w$. We also perform a fermionic Fourier transform to define $\widetilde{T}$ and $\widetilde{V}$, but only to the first $d$ spin orbitals
\begin{equation}
\begin{aligned}
	\mathrm{FFFT}_d^\dagger\cdot A_j^\dagger\cdot \mathrm{FFFT}_d&=
	\begin{cases}
	\frac{1}{\sqrt{d}}\sum_{l=0}^{d-1} e^{-\frac{2\pi i jl}{d}}A_l^\dagger, &0\leq j\leq d-1,\\
	A_j^\dagger, &j\geq d.
	\end{cases}\\
	\mathrm{FFFT}_d^\dagger\cdot A_k\cdot \mathrm{FFFT}_d&=
	\begin{cases}
	\frac{1}{\sqrt{d}}\sum_{m=0}^{d-1} e^{\frac{2\pi i km}{d}}A_m, &0\leq k\leq d-1,\\
	A_k, &k\geq d.
	\end{cases}
\end{aligned}
\end{equation}
Then, a similar calculation shows that
\begin{equation}
\Big\Vert\underbrace{\big[{T},\ldots\big[{T}}_{p},{V}\big]\big]\Big\Vert_\eta
=\Big\Vert\underbrace{\big[\widetilde{T},\ldots\big[\widetilde{T}}_{p},\widetilde{V}\big]\big]\Big\Vert_\eta
=\Om{(ud)^{p}wd}.
\end{equation}
This establishes the first claimed bound in \eq{tightness_sparse} of \thm{fermionic_seminorm_tightness}.

\subsection{Lower-bounding \texorpdfstring{$\norm{[V,\ldots[V,T]]}_{\eta}$}{nested commutators with V}}
\label{sec:tightness_v}
Recall from the previous section that we have constructed the electronic Hamiltonian \eq{def_tv} to prove the tightness of our bound. Comparing to the definition of the interacting-electronic model \eq{second_quantized_ham}, we see that the coefficient matrix $\tau$ has spectral norm $\norm{\tau}=n$, whereas coefficient matrix $\nu$ has max-norm $\norm{\nu}_{\max}=1$. Our goal in this subsection is to lower-bound the fermionic seminorm $\norm{[V,\ldots[V,T]]}_{\eta}$. To this end, we define the $\eta$-electron states for $\eta\leq\frac{n}{2}$:
\begin{equation}
\label{eq:def_state}
\begin{aligned}
	\ket{\psi_\eta}&=\frac{\ket{\overbrace{0\underbrace{1\cdots 1}_{\eta-1}0\cdots 0}^{\frac{n}{2}}10\cdots 0}
		+i\ket{\overbrace{1\underbrace{1\cdots 1}_{\eta-1}0\cdots 0}^{\frac{n}{2}}00\cdots 0}}{\sqrt{2}},\\
	\ket{\phi_\eta}&=\frac{\ket{\overbrace{0\underbrace{1\cdots 1}_{\eta-1}0\cdots 0}^{\frac{n}{2}}10\cdots 0}
		+\ket{\overbrace{1\underbrace{1\cdots 1}_{\eta-1}0\cdots 0}^{\frac{n}{2}}00\cdots 0}}{\sqrt{2}}.
\end{aligned}
\end{equation}
Similar to the previous subsection, we may assume that $n$ is even. We have the following proposition showing that the fermionic seminorm of nested commutators is nearly attained.
\begin{proposition}
	\label{prop:tightness_v}
	Define ${T}$, ${V}$ as in \eq{def_tv} and $\ket{{\psi}_\eta}$, $\ket{{\phi}_\eta}$ as in \eq{def_state}. Then,
	\begin{equation}
		\begin{rcases*}
		\big|\bra{{\psi}_\eta}\overbrace{\big[{V},\ldots\big[{V}}^{p},{T}\big]\big]\ket{{\psi}_\eta}\big|, & p\text{ odd}\\
		\big|\bra{{\phi}_\eta}\underbrace{\big[{V},\ldots\big[{V}}_{p},{T}\big]\big]\ket{{\phi}_\eta}\big|, & p\text{ even}
		\end{rcases*}
		=2^p\eta^p+\cO{\eta^{p-1}}.
	\end{equation}
\end{proposition}

A proof of this proposition is given in \append{commutator_v}. By rescaling the Hamiltonian constructed in \eq{def_tv_tilde}, we can demonstrate the tightness of our bound as follows. For any $s,w>0$, we define the rescaled Hamiltonian as in \eq{def_tv_rescale}. Comparing with the definition of the interacting-electronic model \eq{second_quantized_ham}, we see that $\norm{\tau}=s$ and $\norm{\nu}_{\max}=w$. The above proposition then shows that
\begin{equation}
\Big\Vert\underbrace{\big[{V},\ldots\big[{V}}_{p},{T}\big]\big]\Big\Vert_\eta
=\Om{(w\eta)^{p}s/n}.
\end{equation}
This establishes the second claimed bound in \eq{tightness_general} of \thm{fermionic_seminorm_tightness}.

Note that a similar example can be constructed to demonstrate the tightness of our bound for simulating sparse electronic Hamiltonians. Specifically, for $u,w>0$ and integer $2\leq d\leq \eta\leq\frac{n}{2}$,\footnotemark\ we define the electronic Hamiltonian as in \eq{def_tv_sparse}. Comparing with the definition of the interacting-electronic model \eq{second_quantized_ham}, we see that $\norm{\tau}_{\max}=u$ and $\norm{\nu}_{\max}=w$. A similar calculation then shows that
\footnotetext{Similar to above, the special case $d=1$ can be handled using $T=uA_0^\dagger A_1+uA_1^\dagger A_0$ and $V=wN_0$.}
\begin{equation}
\Big\Vert\underbrace{\big[{V},\ldots\big[{V}}_{p},{T}\big]\big]\Big\Vert_\eta
=\Om{(wd)^{p}u}.
\end{equation}
This proves the second claimed bound in \eq{tightness_sparse} of \thm{fermionic_seminorm_tightness}.

\section{Applications}
\label{sec:app}

The class of interacting-electronic Hamiltonians \eq{second_quantized_ham} encompasses various quantum systems arising in physics and chemistry, for which the performance of digital quantum simulation can be improved using our result. As for illustration, we consider improving quantum simulation of the plane-wave-basis electronic structure in \sec{app_plane_wave} and the Fermi-Hubbard model in \sec{app_hubbard}.

\subsection{Plane-wave-basis electronic structure}
\label{sec:app_plane_wave}
Simulating the electronic-structure Hamiltonians is one of the most promising applications of digital quantum computers. Recall that in the second-quantized plane-wave basis, such a Hamiltonian takes the form
\begin{equation}
\begin{aligned}
H&=\frac{1}{2n}\sum_{j,k,\mu}\kappa_{\mu}^2\cos[\kappa_{\mu}\cdot r_{k-j}]A_{j}^\dagger A_{k}\\
&\quad-\frac{4\pi}{\omega}\sum_{j,\iota,\mu\neq 0}\frac{\zeta_\iota\cos[\kappa_{\mu}\cdot(\widetilde{r}_\iota-r_j)]}{\kappa_{\mu}^2}N_{j}
+\frac{2\pi}{\omega}\sum_{\substack{j\neq k\\\mu\neq 0}}\frac{\cos[\kappa_{\mu}\cdot r_{j-k}]}{\kappa_{\mu}^2}N_{j}N_{k},
\end{aligned}
\end{equation}
where $\omega$ is the volume of the computational cell, $\kappa_{\mu}=2\pi\mu/\omega^{1/3}$ are $n$ vectors of plane-wave frequencies, $\mu$ are three-dimensional vectors of integers with elements in the interval $[-n^{1/3},n^{1/3}]$, $r_j$ are the positions of electrons, $\zeta_\iota$ are nuclear charges, and $\widetilde{r}_\iota$ are the nuclear coordinates. We further rewrite the second term as
\begin{equation}
\begin{aligned}
-\frac{4\pi}{\omega}\sum_{j,\iota,\mu\neq 0}\frac{\zeta_\iota\cos[\kappa_{\mu}\cdot(\widetilde{r}_\iota-r_j)]}{\kappa_{\mu}^2}N_{j}
=-\frac{4\pi}{\omega\eta}\sum_{j,k,\iota,\mu\neq 0}\frac{\zeta_\iota\cos[\kappa_{\mu}\cdot(\widetilde{r}_\iota-r_j)]}{\kappa_{\mu}^2}N_{j}N_{k},
\end{aligned}
\end{equation}
which is valid since we estimate the simulation error within the $\eta$-electron manifold. Comparing with the definition of interacting-electronic model \eq{second_quantized_ham}, we see that
\begin{equation}
\begin{aligned}
	\tau_{j,k}&=\frac{1}{2n}\sum_{\mu}\kappa_{\mu}^2\cos[\kappa_{\mu}\cdot r_{k-j}],\\
	\nu_{l,m}&=-\frac{4\pi}{\omega\eta}\sum_{\iota,\mu\neq 0}\frac{\zeta_\iota\cos[\kappa_{\mu}\cdot(\widetilde{r}_\iota-r_l)]}{\kappa_{\mu}^2}
	+\frac{2\pi}{\omega}\sum_{\mu\neq 0}\frac{\cos[\kappa_{\mu}\cdot r_{l-m}]}{\kappa_{\mu}^2}\left(1-\delta_{l,m}\right).
\end{aligned}
\end{equation}

To proceed, we need to bound the spectral norm $\norm{\tau}$ and the max-norm $\norm{\nu}_{\max}$ of the coefficient matrices. We have
\begin{equation}
	\norm{\tau}=\cO{\frac{n^{2/3}}{\omega^{2/3}}},\qquad
	\norm{\tau}_{\max}=\cO{\frac{1}{n^{1/3}\omega^{2/3}}},\qquad
	\norm{\nu}_{\max}=\cO{\frac{n^{1/3}}{\omega^{1/3}}},
\end{equation}
where the first equality follows from \cite[Eq.\ (F10)]{BWMMNC18}, the second equality follows from \cite[Eq.\ (F11)--(F13)]{BWMMNC18} and the third equality follows from \cite[Eq.\ (F7) and (F9)]{BWMMNC18}. We also consider a constant system density $\eta=\cO{\omega}$ following the setting of \cite{BWMMNC18}. Applying \thm{fermionic_seminorm}, we find that a $p$th-order formula $\mathscr{S}_p(t)$ can approximate the evolution of electronic-structure Hamiltonian with Trotter error
\begin{equation}
\begin{aligned}
	\norm{\mathscr{S}_p(t)-e^{-itH}}_\eta
	&=\cO{\left(\norm{\tau}+\norm{\nu}_{\max}\eta\right)^{p-1}
		\norm{\tau}\norm{\nu}_{\max}\eta^2 t^{p+1}}\\
	&=\cO{\left(\frac{n^{2/3}}{\eta^{2/3}}+n^{1/3}\eta^{2/3}\right)^{p}n^{1/3}\eta^{2/3} t^{p+1}}.
\end{aligned}
\end{equation}
This approximation is accurate for sufficiently small $t$. To evolve for a longer time, we divide the evolution into $r$ steps and use $\mathscr{S}_p(t/r)$ within each step, which gives an approximation with error
\begin{equation}
\begin{aligned}
	\norm{\mathscr{S}_p^r(t/r)-e^{-itH}}_\eta
	\leq r\norm{\mathscr{S}_p(t/r)-e^{-i\frac{t}{r}H}}_\eta
	=\cO{\left(\frac{n^{2/3}}{\eta^{2/3}}+n^{1/3}\eta^{2/3}\right)^{p}n^{1/3}\eta^{2/3} \frac{t^{p+1}}{r^p}}.
\end{aligned}
\end{equation}
To simulate with accuracy $\epsilon$, it suffices to choose
\begin{equation}
	r=\cO{\left(\frac{n^{2/3}}{\eta^{2/3}}+n^{1/3}\eta^{2/3}\right)n^{1/3p}\eta^{2/3p} \frac{t^{1+1/p}}{\epsilon^{1/p}}}.
\end{equation}
Note that this can also be achieved using the weaker bound \eq{general_path_bound} from path counting, since both $\norm{\tau}$ and $n\norm{\tau}_{\max}$ have the same asymptotic scaling.

To simplify our discussion, we consider digital quantum simulation with constant time and accuracy, obtaining
\begin{equation}
	r=\cO{\left(\frac{n^{2/3}}{\eta^{2/3}}+n^{1/3}\eta^{2/3}\right)n^{1/p}}.
\end{equation}
We further implement each Trotter step using the approach of \cite[Sect.\ 5]{LW18}, and obtain a quantum circuit with gate complexity
\begin{equation}
g=\cO{\left(\frac{n^{5/3}}{\eta^{2/3}}+n^{4/3}\eta^{2/3}\right)n^{1/p}\mathrm{polylog}(n)}.
\end{equation}
which implies
\begin{equation}
g=\left(\frac{n^{5/3}}{\eta^{2/3}}+n^{4/3}\eta^{2/3}\right)n^{o(1)}
\end{equation}
by choosing the order $p$ sufficiently large.

Up to a negligible factor $n^{o(1)}$, this gate complexity improves the best previous result of the electronic-structure simulation in the second-quantized plane-wave basis. This is because our approach improves the performance of digital quantum simulation by simultaneously exploiting commutativity of the Hamiltonian and prior knowledge of the initial state, whereas previous results were only able to employ at most one of these information. Indeed, previous work \cite[Appendix G]{BWMMNC18} gave a Trotterization with error bound
\begin{equation}
\norm{\mathscr{S}_p(t)-e^{-itH}}_{\eta}
=\cO{\left(\norm{\tau}+\norm{\nu}_{\max}\eta\right)^{p-1}
	\norm{\tau}\norm{\nu}_{\max}\eta^{p+2} t^{p+1}}.
\end{equation}
Their approach used the initial-state information by computing the Trotter error within the $\eta$-electron manifold, but the commutativity of the Hamiltonian was ignored, giving a simulation with gate count $\left(n^{5/3}\eta^{1/3}+n^{4/3}\eta^{5/3}\right)n^{o(1)}$ worse than our result. On the other hand, the work \cite[Proposition F.4]{CSTWZ19} used commutativity of the Hamiltonian to show
\begin{equation}
\norm{\mathscr{S}_p(t)-e^{-itH}}_{\eta}
=\cO{\left(\norm{\tau}_{\max}+\norm{\nu}_{\max}\right)^{p-1}
	\norm{\tau}_{\max}\norm{\nu}_{\max}n^{p+2} t^{p+1}}.
\end{equation}
and gave a simulation with complexity $\frac{n^{7/3}}{\eta^{1/3}}n^{o(1)}$, whereas Ref.\ \cite{LW18} gave an interaction-picture approach with cost $\cO{\frac{n^{8/3}}{\eta^{2/3}}\mathrm{polylog}(n)}$. Our new result matches these when $\eta$ and $n$ are comparable to each other, but can be much more efficient in the regime where $\eta$ is much smaller than $n$.

Interestingly, our asymptotic result remains conditionally advantageous even when compared with the first-quantized simulations. There, the best previous approach is the interaction-picture approach \cite{Babbush2019} with gate complexity $\cO{n^{1/3}\eta^{8/3}\mathrm{polylog}(n)}$, larger than our new complexity when $n=\cO{\eta^{2-2/(p+1)}\mathrm{polylog}(n)}$. A related approach was described in \cite{Babbush2019} based on qubitization, which has a similar performance comparison with our result.\footnote{It is a subject for future work to compare the concrete resources required by our approach and those required by the first-quantized approach of \cite{SBWRB21}.} See \tab{result_summary} for details.

We mention however that there is one caveat when ignoring the factor $n^{o(1)}$ in our above discussion. This is achieved by choosing the order $p$ of Trotterization sufficiently large, which can result in a gate complexity with an unrealistically large prefactor depending on the definition of higher-order formulas.\footnote{For the $p$th-order Suzuki formula, we have a factor of $5^p$ in the gate complexity, although this may be different for different formulas.} Nevertheless, recent work suggests that Trotterization remains advantageous for simulating the plane-wave-basis electronic structure even with a low-order formula \cite{kivlichan2020improved}, to which our paper provides new theoretical insights.

\subsection{Fermi-Hubbard model}
\label{sec:app_hubbard}
We also consider applications of our result to simulations of the Fermi-Hubbard Hamiltonian, which models many important properties of interacting electrons. This Hamiltonian is defined as
\begin{equation}
H=-s \sum_{\langle j,k\rangle,\sigma}\left(A_{j,\sigma}^\dagger A_{k,\sigma}+A_{k,\sigma}^\dagger A_{j,\sigma}\right)
+v\sum_{j}N_{j,0}N_{j,1},
\end{equation}
where $\langle j,k\rangle$ denotes a summation over nearest-neighbor lattice sites and $\sigma\in\{0,1\}$.

We note that this Hamiltonian can be represented in terms of a sparse interacted Hamiltonian. Indeed, in the one-dimensional case, we have
\begin{equation}
H=-s \sum_{j,\sigma}\left(A_{j,\sigma}^\dagger A_{j+1,\sigma}+A_{j+1,\sigma}^\dagger A_{j,\sigma}\right)
+v\sum_{j}N_{j,0}N_{j,1},
\end{equation}
where $j=0,1,\ldots,n-1$ and $\sigma=0,1$. Comparing with the definition of interacting-electronic model \eq{second_quantized_ham}, we see that
\begin{equation}
	\tau=-s\sum_{j}\left(\ket{j}\bra{j+1}+\ket{j+1}\bra{j}\right)\otimes(\ket{0}\bra{0}+\ket{1}\bra{1}),\qquad
	\nu=\frac{v}{2}\sum_{j}\ket{j}\bra{j}\otimes\left(\ket{0}\bra{1}+\ket{1}\bra{0}\right),
\end{equation}
so the coefficient matrices $\tau$ and $\nu$ are $2$-sparse. Similar analysis holds for the higher-dimensional Fermi-Hubbard model, with the sparsity $d=2^m$ where $m$ is the dimensionality of the lattice.

We can therefore apply \thm{fermionic_seminorm} to conclude that a $p$th-order formula $\mathscr{S}_p(t)$ approximates the evolution of Fermi-Hubbard Hamiltonian with Trotter error
\begin{equation}
	\norm{\mathscr{S}_p(t)-e^{-itH}}_\eta
	=\cO{(s+v)^{p-1}sv2^{m(p+1)}\eta t^{p+1}}
	=\cO{\eta t^{p+1}},
\end{equation}
assuming $s$, $v$, and $m$ are constant. For $r$ steps of Trotterization, we apply the triangle inequality to get
\begin{equation}
	\norm{\mathscr{S}_p^r(t/r)-e^{-itH}}_\eta
	\leq r\norm{\mathscr{S}_p(t/r)-e^{-i\frac{t}{r}H}}_\eta
	=\cO{\eta \frac{t^{p+1}}{r^p}}.
\end{equation}
To simulate with constant time and accuracy, it thus suffices to choose
\begin{equation}
	r=\cO{\eta^{1/p}},
\end{equation}
giving gate complexity\footnotemark
\begin{equation}
	g=\cO{n\eta^{1/p}}.
\end{equation}
\footnotetext{By exploiting locality, we can perform $e^{-itV}$ with $\OO{n}$ gates. An additional logarithmic factor will be introduced if we implement $e^{-itT}$ using the fermionic Fourier transform. However, this can be avoided by further decomposing $T=T_{\text{odd}}+T_{\text{even}}$ as in \cite{CS19}, so that $H=T_{\text{odd}}+T_{\text{even}}+V$. The analysis of Trotter error proceeds in a similar way as in \sec{path}, which establishes the claimed gate complexity of $\cO{n\eta^{1/p}}$.}

The Fermi-Hubbard model only contains nearest-neighbor interactions and, according to \cite{CS19}, can be near optimally simulated with $\cO{n^{1+1/p}}$ gates. On the other hand, recent work \cite{CBC20} shows that Trotterization algorithm has gate complexity $\cO{n\eta^{1+1/p}}$ when restricted to the $\eta$-electron manifold. Our result improves over those previous work by combining the sparsity of interactions, commutativity of the Hamiltonian and information about the initial state.

\section{Discussion}
\label{sec:discuss}

We have given improved quantum simulations of a class of interacting electrons using Trotterization, by simultaneously exploiting commutativity of the Hamiltonian, sparsity of interactions, and prior knowledge of the initial state. We identified applications to simulating the plane-wave-basis electronic structure, improving the best previous result in second quantization up to a negligible factor while conditionally outperforming the first-quantized simulation. We obtained further speedups when the electronic Hamiltonian has $d$-sparse interactions, which gave faster Trotterization of the Fermi-Hubbard model. We constructed concrete electronic systems for which our bounds are almost saturated, providing a provable guarantee on the tightness of our analysis.

Our focus has been on the asymptotic performance of digital quantum simulation throughout this paper. However, we believe that the techniques we have developed can also be used to give quantum simulations with low constant-prefactor overhead, for instance, through a more careful application of our \prop{multilayer}, \prop{multilayer_seminorm}, \prop{path_bound_seminorm} and \prop{cntpathsparse}. Such improvements would especially benefit the simulation of plane-wave-basis electronic structure, where many pairs of Hamiltonian terms commute and the number of electrons can be significantly smaller than the spin-orbital number. Existing numerical studies almost exclusively used the second-order Suzuki formula \cite{kivlichan2020improved,wecker2015solving,babbush2015chemical,Pou15,Campbell20} and their results did not fully leverage the commutativity of the Hamiltonian and the initial-state knowledge, which may then be improved by the techniques presented here using general Trotterization schemes.

Our analysis is applicable to a class of electronic Hamiltonians of the form $H=\sum_{j,k}\tau_{j,k}A_j^\dagger A_k+\sum_{l,m}\nu_{l,m}N_l N_m$. By imposing further constraints on the coefficients, we may somewhat sacrifice this generality but instead get further improvement on the simulation performance. One possibility is to consider the subclass of systems that are translation-invariant, i.e., $\tau_{j,k}=\tau_{j+q,k+q}$ and $\nu_{l,m}=\nu_{l+q,m+q}$. This translational invariance is used in the circuit implementations for both our applications (electronic-structure Hamiltonians and Fermi-Hubbard model), but is nevertheless ignored in the proof of our upper bounds (\thm{fermionic_seminorm}) and tightness result (\thm{fermionic_seminorm_tightness}). By incorporating additional features of the Hamiltonian such as translational invariance, it is plausible that our current complexity estimate can be further improved.

A natural problem that has yet to be addressed here is the simulation of electronic-structure Hamiltonians in a more compact molecular basis. Such Hamiltonians typically take the form $H=\sum_{j,k}h_{j,k}A_j^\dagger A_k+\sum_{j,k,l,m}h_{j,k,l,m}A_j^\dagger A_k A_l^\dagger A_m$, more complex than the electronic model \eq{second_quantized_ham} considered here. In this case, the exponentials of the two-body terms $\sum_{j,k,l,m}h_{j,k,l,m}A_j^\dagger A_k A_l^\dagger A_m$ do not have a convenient circuit implementation and our current approach is not directly applicable. This motivates further developments of \emph{hybrid quantum simulation}, in which Trotterization is combined with more advanced quantum algorithms to speed up digital quantum simulation. We leave a detailed study of such problems as a subject for future work.

More generally, we could consider digital quantum simulations of other types of physical systems, such as bosonic systems \cite{Sawaya2020} or fermion-boson interacting systems \cite{Shaw2020quantumalgorithms}.
We hope our techniques could offer insights to such problems and find further applications in digital quantum simulation beyond what have been discussed here.

\section*{Acknowledgements}
We thank Fernando Brand\~{a}o for inspiring discussions during the initial stages of this work, and G\'{e}za Giedke and anonymous referees for their comments on an earlier draft. YS thanks Nathan Wiebe, Guang Hao Low, Ryan Babbush, Minh Cong Tran, Kunal Sharma, John Preskill, and Andrew Childs for helpful discussions. He is supported by the National Science Foundation RAISE-TAQS 1839204 and Amazon Web Services, AWS Quantum Program. HH is supported by the J. Yang \& Family Foundation. The Institute for Quantum Information and Matter is an NSF Physics Frontiers Center PHY-1733907.

\appendix
\section{Analysis of single-layer commutator}
\label{append:singlelayer}

In this appendix, we complete the proof of \prop{singlelayer_seminorm} that bounds the terms arising in the commutator analysis of first-order formula.

For the third statement of \prop{singlelayer_seminorm}, we let $X=\sum_{j,k,l}\tau_{j,k}\nu_{l,k}A_j^\dagger N_lA_k$ and compute
\begin{equation}
\begin{aligned}
	X^\dagger X
	&=\sum_{j_1,k_1,l_1,j_2,k_2,l_2}\bar{\tau}_{j_1,k_1}\bar{\nu}_{l_1,k_1}\tau_{j_2,k_2}\nu_{l_2,k_2}A_{k_1}^\dagger N_{l_1}A_{j_1}A_{j_2}^\dagger N_{l_2}A_{k_2}\\
	&=\sum_{j_1,k_1,l_1,k_2,l_2}\bar{\tau}_{j_1,k_1}\bar{\nu}_{l_1,k_1}\tau_{j_1,k_2}\nu_{l_2,k_2}A_{k_1}^\dagger N_{l_1}N_{l_2}A_{k_2}\\
	&\quad-\sum_{j_1,k_1,l_1,j_2,k_2,l_2}\bar{\tau}_{j_1,k_1}\bar{\nu}_{l_1,k_1}\tau_{j_2,k_2}\nu_{l_2,k_2}A_{k_1}^\dagger N_{l_1}A_{j_2}^\dagger A_{j_1}N_{l_2}A_{k_2}.
\end{aligned}
\end{equation}
Applying the operator Cauchy-Schwarz inequality (\lem{cauchy}) similarly as in \eq{cauchy_cal},
\begin{equation}
\begin{aligned}
	X^\dagger X
	&\leq\sum_{j_1,k_1,l_1,k_2,l_2}\bar{\tau}_{j_1,k_1}\bar{\nu}_{l_1,k_1}\tau_{j_1,k_2}\nu_{l_2,k_2}A_{k_1}^\dagger N_{l_1}N_{l_2}A_{k_2}\\
	&\quad+\sum_{j_1,k_1,l_1,j_2,k_2,l_2}\bar{\tau}_{j_1,k_1}\bar{\nu}_{l_1,k_1}\tau_{j_1,k_2}\nu_{l_2,k_2}A_{k_1}^\dagger N_{l_1}A_{j_2}^\dagger A_{j_2}N_{l_2}A_{k_2}\\
	&=\sum_{j_1,k_1,l_1,k_2,l_2}\bar{\tau}_{j_1,k_1}\bar{\nu}_{l_1,k_1}\tau_{j_1,k_2}\nu_{l_2,k_2}A_{k_1}^\dagger N_{l_1}N_{l_2}A_{k_2}N.
\end{aligned}
\end{equation}
We now perform diagonalization using \lem{diagonalization}, obtaining
\begin{equation}
	X^\dagger X
	\leq\norm{\tau}^2\sum_{k_1,l_1,l_2}\bar{\nu}_{l_1,k_1}\nu_{l_2,k_1}A_{k_1}^\dagger N_{l_1}N_{l_2}A_{k_1}N.
\end{equation}
Using the H\"{o}lder-type inequality for the expectation value (\lem{fermionic_holder}), we have
\begin{equation}
\begin{aligned}
	\norm{X^\dagger X}_\eta
	&\leq\norm{\norm{\tau}^2\sum_{k_1,l_1,l_2}\bar{\nu}_{l_1,k_1}\nu_{l_2,k_1}A_{k_1}^\dagger N_{l_1}N_{l_2}A_{k_1}N}_\eta
	=\norm{\tau}^2\eta\norm{\sum_{k_1,l_1,l_2}\bar{\nu}_{l_1,k_1}\nu_{l_2,k_1}A_{k_1}^\dagger N_{l_1}N_{l_2}A_{k_1}}_\eta\\
	&\leq\norm{\tau}^2\eta\norm{\sum_{k_1}A_{k_1}^\dagger A_{k_1}}_\eta\max_{k_1}\norm{\sum_{l_1,l_2}\bar{\nu}_{l_1,k_1}\nu_{l_2,k_1}N_{l_1}N_{l_2}}_{\eta-1},
\end{aligned}
\end{equation}
where $\norm{\sum_{k_1}A_{k_1}^\dagger A_{k_1}}_\eta=\eta$ and $\norm{\sum_{l_1,l_2}\bar{\nu}_{l_1,k_1}\nu_{l_2,k_1}N_{l_1}N_{l_2}}_{\eta-1}
\leq\norm{\nu}_{\max}^2\eta^2$.
This completes the proof of the third statement of \prop{singlelayer_seminorm}.

For the fourth statement, we let $X=\sum_{j,k,m}\tau_{j,k}\nu_{j,m}A_j^\dagger N_mA_k$ and compute
\begin{equation}
\begin{aligned}
	X^\dagger X
	&=\sum_{j_1,k_1,m_1,j_2,k_2,m_2}\bar{\tau}_{j_1,k_1}\bar{\nu}_{j_1,m_1}\tau_{j_2,k_2}\nu_{j_2,m_2}A_{k_1}^\dagger N_{m_1}A_{j_1} A_{j_2}^\dagger N_{m_2}A_{k_2}\\
	&=\sum_{j_1,k_1,m_1,k_2,m_2}\bar{\tau}_{j_1,k_1}\bar{\nu}_{j_1,m_1}\tau_{j_1,k_2}\nu_{j_1,m_2}A_{k_1}^\dagger N_{m_1} N_{m_2}A_{k_2}\\
	&\quad-\sum_{j_1,k_1,m_1,j_2,k_2,m_2}\bar{\tau}_{j_1,k_1}\bar{\nu}_{j_1,m_1}\tau_{j_2,k_2}\nu_{j_2,m_2}A_{k_1}^\dagger N_{m_1}A_{j_2}^\dagger A_{j_1}N_{m_2}A_{k_2}.
\end{aligned}
\end{equation}
Applying the operator Cauchy-Schwarz inequality (\lem{cauchy}),
\begin{equation}
\begin{aligned}
	X^\dagger X
	&\leq\sum_{j_1,k_1,m_1,k_2,m_2}\bar{\tau}_{j_1,k_1}\bar{\nu}_{j_1,m_1}\tau_{j_1,k_2}\nu_{j_1,m_2}A_{k_1}^\dagger N_{m_1} N_{m_2}A_{k_2}\\
	&\quad+\sum_{j_1,k_1,m_1,j_2,k_2,m_2}\bar{\tau}_{j_1,k_1}\bar{\nu}_{j_1,m_1}\tau_{j_1,k_2}\nu_{j_1,m_2}A_{k_1}^\dagger N_{m_1}A_{j_2}^\dagger A_{j_2}N_{m_2}A_{k_2}\\
	&=\sum_{j_1,k_1,m_1,k_2,m_2}\bar{\tau}_{j_1,k_1}\bar{\nu}_{j_1,m_1}\tau_{j_1,k_2}\nu_{j_1,m_2}A_{k_1}^\dagger N_{m_1}N_{m_2}A_{k_2}N.
\end{aligned}
\end{equation}
We now use the H\"{o}lder-type inequality for the expectation value (\lem{fermionic_holder}) to get
\begin{equation}
\begin{aligned}
	\norm{X^\dagger X}_{\eta}
	&\leq\norm{\sum_{j_1,k_1,m_1,k_2,m_2}\bar{\tau}_{j_1,k_1}\bar{\nu}_{j_1,m_1}\tau_{j_1,k_2}\nu_{j_1,m_2}A_{k_1}^\dagger N_{m_1}N_{m_2}A_{k_2}N}_\eta\\
	&\leq\eta\norm{\sum_{j_1,k_1,k_2}\bar{\tau}_{j_1,k_1}\tau_{j_1,k_2}A_{k_1}^\dagger A_{k_2}}_\eta\max_{j_1}\norm{\sum_{m_1,m_2}\bar{\nu}_{j_1,m_1}\nu_{j_1,m_2}N_{m_1}N_{m_2}}_{\eta-1}.
\end{aligned}
\end{equation}
The second fermionic seminorm can be directly bounded as $\norm{\sum_{m_1,m_2}\bar{\nu}_{j_1,m_1}\nu_{j_1,m_2}N_{m_1}N_{m_2}}_{\eta-1}
\leq\norm{\nu}_{\max}^2\eta^2$,
whereas the first seminorm can be bounded using diagonalization (\lem{diagonalization})
\begin{equation}
	\norm{\sum_{j_1,k_1,k_2}\bar{\tau}_{j_1,k_1}\tau_{j_1,k_2}A_{k_1}^\dagger A_{k_2}}_\eta
	\leq\norm{\sum_{k_1}\norm{\tau^\dagger \tau}A_{k_1}^\dagger A_{k_1}}_\eta
	\leq\norm{\tau^\dagger \tau}\eta.
\end{equation} 
This completes the proof of the fourth statement of \prop{singlelayer_seminorm}.

For the fifth statement, we let $X=\sum_{j,k}\tau_{j,k}\nu_{j,j}A_j^\dagger A_k$ and compute
\begin{equation}
\begin{aligned}
	X^\dagger X
	&=\sum_{j_1,k_1,j_2,k_2}\bar{\tau}_{j_1,k_1}\bar{\nu}_{j_1,j_1}\tau_{j_2,k_2}\nu_{j_2,j_2}A_{k_1}^\dagger A_{j_1}A_{j_2}^\dagger A_{k_2}\\
	&=\sum_{j_1,k_1,k_2}\bar{\tau}_{j_1,k_1}\bar{\nu}_{j_1,j_1}\tau_{j_1,k_2}\nu_{j_1,j_1}A_{k_1}^\dagger A_{k_2}
	-\sum_{j_1,k_1,j_2,k_2}\bar{\tau}_{j_1,k_1}\bar{\nu}_{j_1,j_1}\tau_{j_2,k_2}\nu_{j_2,j_2}A_{k_1}^\dagger A_{j_2}^\dagger A_{j_1}A_{k_2}.
\end{aligned}
\end{equation}
Applying the operator Cauchy-Schwarz inequality (\lem{cauchy}),
\begin{equation}
\begin{aligned}
	X^\dagger X
	&\leq\sum_{j_1,k_1,k_2}\bar{\tau}_{j_1,k_1}\bar{\nu}_{j_1,j_1}\tau_{j_1,k_2}\nu_{j_1,j_1}A_{k_1}^\dagger A_{k_2}
	+\sum_{j_1,k_1,j_2,k_2}\bar{\tau}_{j_1,k_1}\bar{\nu}_{j_1,j_1}\tau_{j_1,k_2}\nu_{j_1,j_2}A_{k_1}^\dagger A_{j_2}^\dagger A_{j_2}A_{k_2}\\
	&=\sum_{j_1,k_1,k_2}\bar{\tau}_{j_1,k_1}\bar{\nu}_{j_1,j_1}\tau_{j_1,k_2}\nu_{j_1,j_1}A_{k_1}^\dagger A_{k_2}N.
\end{aligned}
\end{equation}
We now use the H\"{o}lder-type inequality for the expectation value (\lem{fermionic_holder}) to get
\begin{equation}
\begin{aligned}
	\norm{X^\dagger X}_\eta
	&\leq\norm{\sum_{j_1,k_1,k_2}\bar{\tau}_{j_1,k_1}\bar{\nu}_{j_1,j_1}\tau_{j_1,k_2}\nu_{j_1,j_1}A_{k_1}^\dagger A_{k_2}N}_\eta\\
	&=\eta\norm{\sum_{j_1,k_1,k_2}\bar{\tau}_{j_1,k_1}\tau_{j_1,k_2}A_{k_1}^\dagger A_{k_2}}_\eta
	\max_{j_1}\norm{\bar{\nu}_{j_1,j_1}\nu_{j_1,j_1}I}_{\eta-1}.
\end{aligned}
\end{equation}
The second fermionic seminorm can be directly bounded by $\norm{\nu}_{\max}^2$, whereas we perform diagonalization to the first seminorm (\lem{diagonalization}):
\begin{equation}
	\norm{\sum_{j_1,k_1,k_2}\bar{\tau}_{j_1,k_1}\tau_{j_1,k_2}A_{k_1}^\dagger A_{k_2}}_\eta
	\leq\norm{\tau^\dagger \tau}\norm{\sum_{k_1}A_{k_1}^\dagger A_{k_1}}_\eta
	=\norm{\tau^\dagger \tau}\eta.
\end{equation}
This completes the proof of the fifth statement of \prop{singlelayer_seminorm}.

For the sixth statement, we let $X=\sum_{j,k,l}\tau_{j,k}\nu_{l,j}A_j^\dagger N_l A_k$ and compute
\begin{equation}
\begin{aligned}
	X^\dagger X
	&=\sum_{j_1,k_1,l_1,j_2,k_2,l_2}\tau_{j_1,k_1}\nu_{l_1,j_1}\tau_{j_2,k_2}\nu_{l_2,j_2}A_{k_1}^\dagger N_{l_1}A_{j_1}A_{j_2}^\dagger N_{l_2} A_{k_2}\\
	&=\sum_{j_1,k_1,l_1,k_2,l_2}\tau_{j_1,k_1}\nu_{l_1,j_1}\tau_{j_1,k_2}\nu_{l_2,j_1}A_{k_1}^\dagger N_{l_1}N_{l_2} A_{k_2}\\
	&\quad-\sum_{j_1,k_1,l_1,j_2,k_2,l_2}\tau_{j_1,k_1}\nu_{l_1,j_1}\tau_{j_2,k_2}\nu_{l_2,j_2}A_{k_1}^\dagger N_{l_1}A_{j_2}^\dagger A_{j_1}N_{l_2} A_{k_2}.
\end{aligned}
\end{equation}
Applying the operator Cauchy-Schwarz inequality (\lem{cauchy}),
\begin{equation}
\begin{aligned}
	X^\dagger X
	&\leq\sum_{j_1,k_1,l_1,k_2,l_2}\tau_{j_1,k_1}\nu_{l_1,j_1}\tau_{j_1,k_2}\nu_{l_2,j_1}A_{k_1}^\dagger N_{l_1}N_{l_2} A_{k_2}\\
	&\quad+\sum_{j_1,k_1,l_1,j_2,k_2,l_2}\tau_{j_1,k_1}\nu_{l_1,j_1}\tau_{j_1,k_2}\nu_{l_2,j_1}A_{k_1}^\dagger N_{l_1}A_{j_2}^\dagger A_{j_2}N_{l_2} A_{k_2}\\
	&=\sum_{j_1,k_1,l_1,k_2,l_2}\tau_{j_1,k_1}\nu_{l_1,j_1}\tau_{j_1,k_2}\nu_{l_2,j_1}A_{k_1}^\dagger N_{l_1}N_{l_2} A_{k_2}N.
\end{aligned}
\end{equation}
We now use the H\"{o}lder-type inequality for the expectation value (\lem{fermionic_holder}) to get
\begin{equation}
\begin{aligned}
	\norm{X^\dagger X}_\eta
	&\leq\norm{\sum_{j_1,k_1,l_1,k_2,l_2}\tau_{j_1,k_1}\nu_{l_1,j_1}\tau_{j_1,k_2}\nu_{l_2,j_1}A_{k_1}^\dagger N_{l_1}N_{l_2} A_{k_2}N}_\eta\\
	&=\eta\norm{\sum_{j_1,k_1,k_2}\tau_{j_1,k_1}\tau_{j_1,k_2}A_{k_1}^\dagger A_{k_2}}_\eta
	\max_{j_1}\norm{\sum_{l_1,l_2}\nu_{l_1,j_1}\nu_{l_2,j_1}N_{l_1}N_{l_2}}_{\eta-1}.
\end{aligned}
\end{equation}
The second fermionic seminorm can be directly bounded by $\norm{\nu}_{\max}^2\eta^2$, whereas we perform diagonalization to the first seminorm (\lem{diagonalization}):
\begin{equation}
	\norm{\sum_{j_1,k_1,k_2}\tau_{j_1,k_1}\tau_{j_1,k_2}A_{k_1}^\dagger A_{k_2}}_\eta
	\leq\norm{\tau^\dagger \tau}\norm{\sum_{k_1}A_{k_1}^\dagger A_{k_2}}_\eta
	=\norm{\tau^\dagger \tau}\eta.
\end{equation}
This completes the proof of the sixth statement of \prop{singlelayer_seminorm}.

\section{Counting fermionic paths for non-sparse interactions}
\label{append:pathcountdense}

In this appendix, we use the path-counting technique to prove \eq{general_path_bound} for non-sparse interacting electrons.
We will make use of the following commutation relations
\begin{align}
	\left[A_j^\dagger A_k,A_{j_x}^\dagger\right]&=\delta_{k,j_x}A_j^\dagger,\quad
	\left[A_j^\dagger A_k,A_{k_y}\right]=-\delta_{k_y,j}A_k,\quad
	\left[A_j^\dagger A_k,N_{l_z}\right]
	=\delta_{k,l_z}A_j^\dagger A_k
	-\delta_{j,l_z}A_j^\dagger A_k, \label{eq:commutator_t_3}\\
	\left[N_lN_m,A_{j_x}^\dagger\right]
	&=\delta_{m,j_x}N_lA_{j_x}^\dagger+\delta_{l,j_x}N_m A_{j_x}^\dagger- \delta_{l,k_x} \delta_{m,k_x} A^\dagger_{j_x},\label{eq:commutator_v_3a}\\
	\left[N_lN_m,A_{k_x}\right]
	&=-\delta_{m,k_x} A_{k_x} N_l-\delta_{l,k_x} A_{k_x} N_m + \delta_{l,k_x} \delta_{m,k_x} A_{k_x},\label{eq:commutator_v_3b}
\end{align}
which are slightly different from the ones used before. These relations can be derived in a similar way as in Equation~\eqref{eq:commutator_t}~and~\eqref{eq:commutator_v}.

Our analysis of the non-sparse interactions mirrors that of the sparse case in \sec{path-count-sparse}.
\begin{proposition}[Non-sparse path-counting bound] \label{prop:cntpathdense} 
Under the same assumption as in \prop{path_bound_seminorm}, we have
\begin{align}
\sum_{\langle j_{p+1}, k_{p+1} \rangle} \ldots \sum_{\langle j_{1}, k_{1} \rangle} \sum_{P \rhd \left(H^{\gamma_{p+1}}_{j_{p+1} k_{p+1}}, \ldots, H^{\gamma_{1}}_{j_{1} k_{1}} \right)} \norm{P \ket{\pmb{c}_\eta}} &= \cO{n^{\abs{\pmb \gamma}}  \eta^{p+2-\abs{\pmb \gamma}}},\\
\sum_{\langle j_{p+1}, k_{p+1} \rangle} \ldots \sum_{\langle j_{1}, k_{1} \rangle} \sum_{P^\dagger \rhd \left(H^{\gamma_{p+1}}_{j_{p+1} k_{p+1}}, \ldots, H^{\gamma_{1}}_{j_{1} k_{1}} \right)} \norm{P^\dagger \ket{\pmb{c}_\eta}} &= \cO{n^{\abs{\pmb \gamma}}  \eta^{p+2-\abs{\pmb \gamma}}}.
\end{align}
\end{proposition}
\begin{proof}
We will prove the following claims by induction on $q = 2, \ldots, p+1$.
\begin{itemize}
\item All fermionic paths $P$ are products of $A^\dagger_i A_j$ and $N_k$.
\item All fermionic paths $P$ have at most $q+1$ elementary fermionic operators.
\item The number of fermionic paths $P$ that start with a fermionic operator acting on a specific site $i$ is at most $3^{q-1}q!n^{\sum_{q'=1}^{q} \gamma_{q'}}  \eta^{\sum_{q'=1}^{q} (1-\gamma_{q'})}$.
\end{itemize}
The base case $q=2$ can be easily verified by noting that we only need to consider $[T, V]$ or $[V, T]$.
For every site $i$, there are at most $6 n \eta$ fermionic paths starting with this site, all of which are products of $A^\dagger_i A_j$ and $N_k$. 
This is because there are at most three summation indices. The rightmost index must be equal to $i$ and the indices for $N_k$, $A_j$ have at most $\eta$ choices, while the remaining index has $n$ possible choices, giving a total of $n\eta$ choices.
The additional factor of $6$ comes from the number of different expansion terms in Equation~\eqref{eq:commutator_t_3},~\eqref{eq:commutator_v_3a},~\eqref{eq:commutator_v_3b}.
Furthermore, every fermionic path consists of at most $3$ fermionic operators.
These established the claims for the base case $q=2$.

For every $q>2$, we now use the induction hypothesis for $q-1$ to prove the claims for $q$.
If $\gamma_q = 1$, then we take another commutator with $T = \sum_{j_q, k_q} \tau_{j_q, k_q} A^\dagger_{j_q} A_{k_q}$.
We can see that all fermionic paths $P \rhd \left(H^{\gamma_{q}}_{j_{q} k_{q}}, \ldots, H^{\gamma_{1}}_{j_{1} k_{1}} \right)$ come from the expansion of
\begin{equation}
[A^\dagger_{j_q} A_{k_q}, P'],\ \forall \langle j_q, k_q \rangle,\ \forall P' \rhd \left(H^{\gamma_{q-1}}_{j_{q-1} k_{q-1}}, \ldots, H^{\gamma_{1}}_{j_{1} k_{1}} \right).
\end{equation}
Using the commutation rule $[X, Y_1 \ldots Y_\kappa] = \sum_{k=1}^\kappa Y_1\ldots Y_{k-1} [X, Y_k] Y_{k+1} \ldots Y_{\kappa}$,
we show that the claims hold for $q$ as follows.
When we take the commutation of $A^\dagger_{j_q} A_{k_q}$ with $A_j^\dagger$ or $A_k$, we know from \eq{commutator_v_3a} that one free index will be introduced, resulting in an additional factor of $n$.
When we take the commutation of $A^\dagger_{j_q} A_{k_q}$ with $N_l$, we will remove the fermionic operator $N_l$ and replace it with $A^\dagger_{j_q} A_{k_q}$, which removes a factor of $\eta$ and adds an additional factor of $n \eta$.
Additionally, there are at most $(q-1)+1$ fermionic operators in $P'$.
Hence the number of fermionic paths $P$ that start with a fermionic operator acting on site $i$ is at most
\begin{equation}
(2q)3^{q-2}(q-1)! n n^{\sum_{q'=1}^{q-1} \gamma_{q'}}  \eta^{\sum_{q'=1}^{q-1} (1-\gamma_{q'})} \leq 3^{q-1}q! n^{\sum_{q'=1}^{q} \gamma_{q'}}  \eta^{\sum_{q'=1}^{q} (1-\gamma_{q'})}.
\end{equation}
Furthermore, in both cases, we add at most one additional fermionic operator.
Therefore, all fermionic paths will have at most $(q-1)+1+1 = q+1$ fermionic operators.
And $[A^\dagger_{j_q} A_{k_q}, P']$ remains a product of $A^\dagger_i A_j$ and $N_k$.

The inductive step for $\gamma_q = 0$ follows from a similar argument.
The first two claims can be directly verified.
For the last claim, we proceed in a slightly different way as follows.
When we take the commutator of $N^\dagger_{j_q} N_{k_q}$ with $A_j^\dagger$ or $A_k$, we will add $N_{j_q}$ or $N_{k_q}$ to the fermionic path, which results in an additional factor of $\eta$.
When we take the commutator of $N^\dagger_{j_q} N_{k_q}$ with $N_k$, the commutator is equal to zero.
Hence the number of fermionic paths that start with a fermionic operator acting on site $i$ is at most
\begin{equation}
(3q)3^{q-2}(q-1)! \eta n^{\sum_{q'=1}^{q-1} \gamma_{q'}}  \eta^{\sum_{q'=1}^{q-1} (1-\gamma_{q'})} \leq 3^{q-1}q! n^{\sum_{q'=1}^{q} \gamma_{q'}}  \eta^{\sum_{q'=1}^{q} (1-\gamma_{q'})}.
\end{equation}
We have thus shown that the claims hold for $q$.

Performing the induction on $q$ from $2$ to $p+1$ shows that
the number of fermionic paths starting with site $i$ is at most
\begin{equation}
3^{p}(p+1)! n^{\sum_{q=1}^{p+1} \gamma_q} \eta^{\sum_{q=1}^{p+1} (1-\gamma_q)} = \cO{n^{\sum_{q=1}^{p+1} \gamma_q}  \eta^{\sum_{q=1}^{p+1} (1-\gamma_q)}}.
\end{equation}
Because each fermionic path $P$ is a product of $A^\dagger_i A_j$ and $N_k$,
$\norm{P \ket{\pmb{c}_\eta}}$ would be nonzero only if the rightmost fermionic operator acts on one of the $\eta$ occupied sites in the configuration $\pmb{c}_\eta$.
Hence there are at most $\eta \mathcal{O}(n^{\sum_{q=1}^{p+1} \gamma_q}  \eta^{\sum_{q=1}^{p+1} (1-\gamma_q)})$ fermionic paths with non-zero $\norm{P \ket{\pmb{c}_\eta}}$.
Finally, recall from \lem{fermionic_path} that $\norm{P \ket{\pmb{c}_\eta}}$ is either $0$ or $1$. Therefore, we have
\begin{equation}
\sum_{\langle j_{p+1}, k_{p+1} \rangle} \ldots \sum_{\langle j_{1}, k_{1} \rangle} \sum_{P \rhd \left(H^{\gamma_{p+1}}_{j_{p+1} k_{p+1}}, \ldots, H^{\gamma_{1}}_{j_{1} k_{1}} \right)} \norm{P \ket{\pmb{c}_\eta}} = \cO{n^{\sum_{q=1}^{p+1} \gamma_q}  \eta^{1+\sum_{q=1}^{p+1} (1-\gamma_q)}}.
\end{equation}
The other bound
\begin{equation}
\sum_{\langle j_{p+1}, k_{p+1} \rangle} \ldots \sum_{\langle j_{1}, k_{1} \rangle} \sum_{P^\dagger \rhd \left(H^{\gamma_{p+1}}_{j_{p+1} k_{p+1}}, \ldots, H^{\gamma_{1}}_{j_{1} k_{1}} \right)} \norm{P^\dagger \ket{\pmb{c}_\eta}} = \cO{n^{\abs{\pmb \gamma}}  \eta^{p+2-\abs{\pmb \gamma}}}
\end{equation}
can be similarly proved using the fact that the leftmost fermionic operator of $P^\dagger$ must act on one of the $\eta$ occupied sites in any fixed configuration.
\end{proof}

We can combine the above proposition with the path-counting bound (\prop{path_bound_seminorm}) to obtain
\begin{align}
\norm{\left[H_{\gamma_{p+1}},\cdots\left[H_{\gamma_2},H_{\gamma_1}\right]\right]}_{\eta}
&= \cO{(n \norm{\tau}_{\max})^{\abs{\pmb \gamma}} (\eta \norm{\nu}_{\max})^{p+1-\abs{\pmb \gamma}}\eta}.
\end{align}
Finally, we use the Trotter error bound \eq{pf2k_eta} to obtain
\begin{equation}
\norm{\mathscr{S}_p(t)-e^{-itH}}_{\eta}
=\cO{ \left( n\norm{\tau}_{\max} + \eta \norm{\nu}_{\max} \right)^{p-1} \norm{\tau}_{\max} \norm{\nu}_{\max} n \eta^2 t^{p+1} },
\end{equation}
completing the proof of \eq{general_path_bound}.

Note that this bound is slightly worse than Eq.\ \eq{general_bound} of \thm{fermionic_seminorm}, as the norm inequality $\norm{\tau} \leq n \norm{\tau}_{\max}$ always holds but not necessarily saturates.
However, in the electronic-structure application, it indeed holds that $\norm{\tau}$ and $n \norm{\tau}_{\max}$ have the same asymptotic scaling, so \eq{general_bound} and \eq{general_path_bound} give digital quantum simulations with the same asymptotic gate complexity. See \sec{app_plane_wave} for further discussions.

\section{Lower-bounding \texorpdfstring{$\norm{[T,\ldots[T,V]]}_{\eta}$}{nested commutators with T}}
\label{append:commutator_t}

In this appendix, we prove \prop{tightness_t} that lower-bounds $\norm{[T,\ldots[T,V]]}_{\eta}$ for the electronic Hamiltonian \eq{def_tv}. After the fermionic Fourier transform \eq{ffft}, we have
\begin{equation}
\begin{aligned}
\widetilde{T}&=\mathrm{FFFT}^\dagger\cdot T\cdot\mathrm{FFFT}=nN_0,\\
\widetilde{V}&=\mathrm{FFFT}^\dagger\cdot V\cdot\mathrm{FFFT}
=\frac{1}{n^2}\sum_{j,k,l,m}
\left(\sum_{x=0}^{\frac{n}{2}-1}e^{\frac{2\pi i x(k-j)}{n}}\right)
\left(\sum_{y=0}^{\frac{n}{2}-1}e^{\frac{2\pi i y(m-l)}{n}}\right)
A_j^\dagger A_k A_l^\dagger A_m,
\end{aligned}
\end{equation}
which gives the commutator
\begin{equation}
\left[\widetilde{T},\widetilde{V}\right]
=\frac{1}{n}\sum_{k,l,m}H_{0klm}
-\frac{1}{n}\sum_{j,l,m}H_{j0lm}
+\frac{1}{n}\sum_{j,k,m}H_{jk0m}
-\frac{1}{n}\sum_{j,k,l}H_{jkl0}
\end{equation}
with
\begin{equation}
H_{jklm}=\tau_{jklm}A_j^\dagger A_k A_l^\dagger A_m,\qquad
\tau_{jklm}=\left(\sum_{x=0}^{\frac{n}{2}-1}e^{\frac{2\pi i x(k-j)}{n}}\right)
\left(\sum_{y=0}^{\frac{n}{2}-1}e^{\frac{2\pi i y(m-l)}{n}}\right).
\end{equation}

For $\eta\leq\frac{n}{2}$, we will choose the initial state from the two-dimensional subspace spanned by
\begin{equation}
\ket{\widetilde{\psi}_0}=\ket{010\cdots 0\overbrace{1\cdots 1}^{\eta-1}},\qquad
\ket{\widetilde{\psi}_1}=\ket{100\cdots 0\overbrace{1\cdots 1}^{\eta-1}}.
\end{equation}
Denoting the projection to this subspace as $\widetilde{\Pi}=\ket{\widetilde{\psi}_0}\bra{\widetilde{\psi}_0}+\ket{\widetilde{\psi}_1}\bra{\widetilde{\psi}_1}$, we have that $\widetilde{\Pi}$ commutes with $\widetilde{T}=nN_0$, which means $\widetilde{\Pi}[\widetilde{T},\ldots[\widetilde{T},\widetilde{V}]]\widetilde{\Pi}=[\widetilde{T},\ldots\widetilde{\Pi}[\widetilde{T},\widetilde{V}]\widetilde{\Pi}]$. We simplify the effective commutator $\widetilde{\Pi}[\widetilde{T},\widetilde{V}]\widetilde{\Pi}$ based on the following observations:
\begin{enumerate}
	\item $A_0^\dagger A_k A_l^\dagger A_m$: This will always nullify $\ket{\widetilde{\psi}_0}$ from left.
	For $\bra{\widetilde{\psi}_1}A_0^\dagger A_k A_l^\dagger A_m\ket{\widetilde{\psi}_1}$ to be nonzero, we must let one of $\{k,m\}$ be $0$, while the other is equal to $l$.
	For $\bra{\widetilde{\psi}_1}A_0^\dagger A_k A_l^\dagger A_m\ket{\widetilde{\psi}_0}$ to be nonzero, we must let one of $\{k,m\}$ be $1$, while the other is equal to $l$.
	\item $A_j^\dagger A_0 A_l^\dagger A_m$: For $\bra{\widetilde{\psi}_0}A_j^\dagger A_0 A_l^\dagger A_m\ket{\widetilde{\psi}_0}$ to be nonzero, we must let $l=0$ and $j=m$.
	For $\bra{\widetilde{\psi}_0}A_j^\dagger A_0 A_l^\dagger A_m\ket{\widetilde{\psi}_1}$ to be nonzero, we must let one of $\{j,l\}$ be $1$, while the other is equal to $m$.
	For $\bra{\widetilde{\psi}_1}A_j^\dagger A_0 A_l^\dagger A_m\ket{\widetilde{\psi}_1}$ to be nonzero, we must let $j=0$ and $l=m$.
	For $\bra{\widetilde{\psi}_1}A_j^\dagger A_0 A_l^\dagger A_m\ket{\widetilde{\psi}_0}$ to be nonzero, we must let $j=0$, $l=0$ and $m=1$.
	\item $A_j^\dagger A_k A_0^\dagger A_m$: For $\bra{\widetilde{\psi}_0}A_j^\dagger A_k A_0^\dagger A_m\ket{\widetilde{\psi}_0}$ to be nonzero, we must let $k=0$ and $j=m$.
	For $\bra{\widetilde{\psi}_1}A_j^\dagger A_k A_0^\dagger A_m\ket{\widetilde{\psi}_0}$ to be nonzero, we must let one of $\{k,m\}$ be $1$, while the other is equal to $j$.
	For $\bra{\widetilde{\psi}_1}A_j^\dagger A_k A_0^\dagger A_m\ket{\widetilde{\psi}_1}$ to be nonzero, we must let $m=0$ and $j=k$.
	For $\bra{\widetilde{\psi}_0}A_j^\dagger A_k A_0^\dagger A_m\ket{\widetilde{\psi}_1}$ to be nonzero, we must let $m=0$, $k=0$ and $j=1$.
	\item $A_j^\dagger A_k A_l^\dagger A_0$: This will always nullify $\ket{\widetilde{\psi}_0}$ from right.
	For $\bra{\widetilde{\psi}_1}A_j^\dagger A_k A_l^\dagger A_0\ket{\widetilde{\psi}_1}$ to be nonzero, we must let one of $\{j,l\}$ be $0$, while the other is equal to $k$.
	For $\bra{\widetilde{\psi}_0}A_j^\dagger A_k A_l^\dagger A_0\ket{\widetilde{\psi}_1}$ to be nonzero,  we must let one of $\{j,l\}$ be $1$, while the other is equal to $k$.
\end{enumerate}
After removing double-counting and canceling redundant terms, we obtain
\begin{equation}
\begin{aligned}
\widetilde{\Pi}\left[\widetilde{T},\widetilde{V}\right]\widetilde{\Pi}
&=\cancel{\frac{1}{n}\sum_{l}H_{00ll}}
+\frac{1}{n}\sum_{l}H_{01ll}
\cancel{+\frac{1}{n}\sum_{k}H_{0kk0}}
+\frac{1}{n}\sum_{k}H_{0kk1}\\
&\quad \cancel{-\frac{1}{n}H_{0000}}
-\frac{1}{n}H_{0111}\\
&\quad \cancel{-\frac{1}{n}\sum_{j}H_{j00j}}
-\frac{1}{n}\sum_{l}H_{10ll}
-\frac{1}{n}\sum_{j}H_{j01j}
\cancel{-\frac{1}{n}\sum_{l}H_{00ll}}
-\frac{1}{n}H_{0001}\\
&\quad \cancel{+\frac{1}{n}H_{0000}}
\underbrace{+\frac{1}{n}H_{1011}}_{0}\\
&\quad \cancel{+\frac{1}{n}\sum_{j}H_{j00j}}
+\frac{1}{n}\sum_{j}H_{jj01}
+\frac{1}{n}\sum_{j}H_{j10j}
\cancel{+\frac{1}{n}\sum_{j}H_{jj00}}
+\frac{1}{n}H_{1000}\\
&\quad \cancel{-\frac{1}{n}H_{0000}}
\underbrace{-\frac{1}{n}H_{1101}}_{0}\\
&\quad \cancel{-\frac{1}{n}\sum_{j}H_{jj00}}
-\frac{1}{n}\sum_{j}H_{jj10}
\cancel{-\frac{1}{n}\sum_{k}H_{0kk0}}
-\frac{1}{n}\sum_{k}H_{1kk0}\\
&\quad \cancel{+\frac{1}{n}H_{0000}}
+\frac{1}{n}H_{1110}.
\end{aligned}
\end{equation}
We merge the remaining twelve terms into four groups:
\begin{enumerate}
	\item The first group contains terms
	\begin{equation}
	\hspace*{-0.5cm}
	\begin{aligned}
	&\ \frac{1}{n}\sum_{l}H_{01ll}
	-\frac{1}{n}\sum_{l}H_{10ll}
	+\frac{1}{n}\sum_{j}H_{jj01}
	-\frac{1}{n}\sum_{j}H_{jj10}\\
	=&\ \frac{1}{n}\sum_{l}
	\tau_{01ll}
	A_0^\dagger A_1 A_l^\dagger A_l
	-\frac{1}{n}\sum_{l}
	\tau_{10ll}
	A_1^\dagger A_0 A_l^\dagger A_l
	+\frac{1}{n}\sum_{j}
	\tau_{jj01}
	A_j^\dagger A_j A_0^\dagger A_1
	-\frac{1}{n}\sum_{j}
	\tau_{jj10}
	A_j^\dagger A_j A_1^\dagger A_0\\
	=&\ N\left(\sum_{x=0}^{\frac{n}{2}-1}e^{\frac{2\pi i x}{n}}\right)A_0^\dagger A_1
	-N\left(\sum_{x=0}^{\frac{n}{2}-1}e^{-\frac{2\pi i x}{n}}\right)A_1^\dagger A_0.
	\end{aligned}
	\end{equation}
	We will see that this is the dominant contribution to the effective commutator that is at least $\Om{n\eta}$.
	\item The second group contains terms
	\begin{equation}
	\begin{aligned}
	&-\frac{1}{n}H_{0111}
	-\frac{1}{n}H_{0001}
	+\frac{1}{n}H_{1000}
	+\frac{1}{n}H_{1110}\\
	=&-\frac{1}{n}
	\tau_{0111}
	A_0^\dagger A_1 A_1^\dagger A_1
	-\frac{1}{n}
	\tau_{0001}
	A_0^\dagger A_0 A_0^\dagger A_1
	+\frac{1}{n}
	\tau_{1000}
	A_1^\dagger A_0 A_0^\dagger A_0
	+\frac{1}{n}
	\tau_{1110}
	A_1^\dagger A_1 A_1^\dagger A_0\\
	=&-\left(\sum_{x=0}^{\frac{n}{2}-1}e^{\frac{2\pi i x}{n}}\right)A_0^\dagger A_1
	+\left(\sum_{x=0}^{\frac{n}{2}-1}e^{-\frac{2\pi i x}{n}}\right)A_1^\dagger A_0=\cO{n},
	\end{aligned}
	\end{equation}
	which does not dominate the result scaling.
	\item The third group contains terms
	\begin{equation}
	\begin{aligned}
	\frac{1}{n}\sum_{k}H_{0kk1}
	-\frac{1}{n}\sum_{k}H_{1kk0}
	&=\frac{1}{n}\sum_{k}
	\tau_{0kk1}
	A_0^\dagger A_k A_k^\dagger A_1
	-\frac{1}{n}\sum_{k}
	\tau_{1kk0}
	A_1^\dagger A_k A_k^\dagger A_0\\
	&=\frac{1}{n}
	\tau_{0001}
	A_0^\dagger A_1
	+\frac{1}{n}\sum_{k}
	\tau_{0kk1}
	A_k A_k^\dagger A_0^\dagger A_1\\
	&\quad -\frac{1}{n}
	\tau_{1110}
	A_1^\dagger A_0
	-\frac{1}{n}\sum_{k}
	\tau_{1kk0}
	A_k A_k^\dagger A_1^\dagger A_0\\
	&=\frac{1}{n}
	\tau_{0001}
	A_0^\dagger A_1
	+\frac{1}{n}\sum_{k}
	\tau_{0kk1}
	A_0^\dagger A_1
	-\frac{1}{n}\sum_{k}
	\tau_{0kk1}
	A_k^\dagger A_k A_0^\dagger A_1\\
	&\quad -\frac{1}{n}
	\tau_{1110}
	A_1^\dagger A_0
	-\frac{1}{n}\sum_{k}
	\tau_{1kk0}
	A_1^\dagger A_0
	+\frac{1}{n}\sum_{k}
	\tau_{1kk0}
	A_k^\dagger A_k A_1^\dagger A_0,
	\end{aligned}
	\end{equation}
	where
	\begin{equation}
	\begin{aligned}
	\frac{1}{n}
	\tau_{0001}
	A_0^\dagger A_1
	-\frac{1}{n}
	\tau_{1110}
	A_1^\dagger A_0
	=\frac{1}{2}\left(\sum_{x=0}^{\frac{n}{2}-1}e^{\frac{2\pi i x}{n}}\right)A_0^\dagger A_1
	-\frac{1}{2}\left(\sum_{x=0}^{\frac{n}{2}-1}e^{-\frac{2\pi i x}{n}}\right)A_1^\dagger A_0
	=\cO{n},
	\end{aligned}
	\end{equation}
	and
	\begin{equation}
	\hspace*{-0.5cm}
	\begin{aligned}
	&\ \frac{1}{n}\sum_{k}
	\tau_{0kk1}
	A_0^\dagger A_1
	-\frac{1}{n}\sum_{k}
	\tau_{1kk0}
	A_1^\dagger A_0\\
	=&\ \frac{1}{n}\sum_{k}
	\left(\sum_{x=0}^{\frac{n}{2}-1}e^{\frac{2\pi i xk}{n}}\sum_{y=0}^{\frac{n}{2}-1}e^{\frac{2\pi i y(1-k)}{n}}\right)
	A_0^\dagger A_1
	-\frac{1}{n}\sum_{k}
	\left(\sum_{x=0}^{\frac{n}{2}-1}e^{\frac{2\pi i x(k-1)}{n}}\sum_{y=0}^{\frac{n}{2}-1}e^{-\frac{2\pi i yk}{n}}\right)
	A_1^\dagger A_0\\
	=&\ \left(\sum_{x=0}^{\frac{n}{2}-1}e^{\frac{2\pi i x}{n}}\right)A_0^\dagger A_1
	-\left(\sum_{x=0}^{\frac{n}{2}-1}e^{-\frac{2\pi i x}{n}}\right)A_1^\dagger A_0
	=\cO{n}.
	\end{aligned}
	\end{equation}
	We rewrite the remaining terms as
	\begin{equation}
	\begin{aligned}
	&-\frac{1}{n}\sum_{k}
	\tau_{0kk1}
	A_k^\dagger A_k A_0^\dagger A_1
	+\frac{1}{n}\sum_{k}
	\tau_{1kk0}
	A_k^\dagger A_k A_1^\dagger A_0\\
	=&-\frac{1}{n}\sum_{k}
	\left(\sum_{x=0}^{\frac{n}{2}-1}e^{\frac{2\pi i xk}{n}}\sum_{y=0}^{\frac{n}{2}-1}e^{\frac{2\pi i y(1-k)}{n}}\right)
	A_k^\dagger A_k A_0^\dagger A_1\\
	&\ +\frac{1}{n}\sum_{k}
	\left(\sum_{x=0}^{\frac{n}{2}-1}e^{\frac{2\pi i x(k-1)}{n}}\sum_{y=0}^{\frac{n}{2}-1}e^{-\frac{2\pi i yk}{n}}\right)
	A_k^\dagger A_k A_1^\dagger A_0.
	\end{aligned}
	\end{equation}
	\item The fourth group contains terms
	\begin{equation}
	\begin{aligned}
	&-\frac{1}{n}\sum_{j}H_{j01j}
	+\frac{1}{n}\sum_{j}H_{j10j}\\
	=&-\frac{1}{n}\sum_{j}
	\tau_{j01j}
	A_j^\dagger A_0 A_1^\dagger A_j
	+\frac{1}{n}\sum_{j}
	\tau_{j10j}
	A_j^\dagger A_1 A_0^\dagger A_j\\
	=&-\frac{1}{n}
	\tau_{1011}
	A_1^\dagger A_0
	+\frac{1}{n}\sum_{j}
	\tau_{j01j}
	A_j^\dagger A_j A_1^\dagger A_0
	+\frac{1}{n}
	\tau_{0100}
	A_0^\dagger A_1
	-\frac{1}{n}\sum_{j}
	\tau_{j10j}
	A_j^\dagger A_j A_0^\dagger A_1.
	\end{aligned}
	\end{equation}
	Similar to the previous case, we have
	\begin{equation}
	-\frac{1}{n}
	\tau_{1011}
	A_1^\dagger A_0
	+\frac{1}{n}
	\tau_{0100}
	A_0^\dagger A_1
	=\cO{n},
	\end{equation}
	whereas the remaining terms can be rewritten as
	\begin{equation}
	\begin{aligned}
	&\ \frac{1}{n}\sum_{j}
	\tau_{j01j}
	A_j^\dagger A_j A_1^\dagger A_0
	-\frac{1}{n}\sum_{j}
	\tau_{j10j}
	A_j^\dagger A_j A_0^\dagger A_1\\
	=&\ \frac{1}{n}\sum_{j}
	\left(\sum_{x=0}^{\frac{n}{2}-1}e^{-\frac{2\pi i xj}{n}}\sum_{y=0}^{\frac{n}{2}-1}e^{\frac{2\pi i y(j-1)}{n}}\right)
	A_j^\dagger A_j A_1^\dagger A_0\\
	&\ -\frac{1}{n}\sum_{j}
	\left(\sum_{x=0}^{\frac{n}{2}-1}e^{\frac{2\pi i x(1-j)}{n}}\sum_{y=0}^{\frac{n}{2}-1}e^{\frac{2\pi i yj}{n}}\right)
	A_j^\dagger A_j A_0^\dagger A_1.
	\end{aligned}
	\end{equation}
\end{enumerate}

To summarize, the effective commutator $\widetilde{\Pi}\left[\widetilde{T},\widetilde{V}\right]\widetilde{\Pi}$ has action
\begin{equation}
\begin{aligned}
	\widetilde{\Pi}\left[\widetilde{T},\widetilde{V}\right]\widetilde{\Pi}
	=&\ N\left(\sum_{x=0}^{\frac{n}{2}-1}e^{\frac{2\pi i x}{n}}\right)A_0^\dagger A_1
	-N\left(\sum_{x=0}^{\frac{n}{2}-1}e^{-\frac{2\pi i x}{n}}\right)A_1^\dagger A_0\\
	&\ -\frac{2}{n}\sum_{k}
	\left(\sum_{x=0}^{\frac{n}{2}-1}e^{\frac{2\pi i xk}{n}}\sum_{y=0}^{\frac{n}{2}-1}e^{\frac{2\pi i y(1-k)}{n}}\right)
	A_k^\dagger A_k A_0^\dagger A_1\\
	&\ +\frac{2}{n}\sum_{k}
	\left(\sum_{x=0}^{\frac{n}{2}-1}e^{\frac{2\pi i x(k-1)}{n}}\sum_{y=0}^{\frac{n}{2}-1}e^{-\frac{2\pi i yk}{n}}\right)
	A_k^\dagger A_k A_1^\dagger A_0
	+\cO{n}.
\end{aligned}
\end{equation}
We now take the expectation of this operator with respect to the state
\begin{equation}
	\ket{\widetilde{\psi}_\eta}=\frac{\ket{010\cdots 0\overbrace{1\cdots 1}^{\eta-1}}+\ket{100\cdots 0\overbrace{1\cdots 1}^{\eta-1}}}{\sqrt{2}}.
\end{equation}
Using the limit
\begin{equation}
	\lim\limits_{z\rightarrow 0}\left(\frac{2}{1-e^{2\pi iz}}+\frac{1}{\pi iz}\right)=1
	=\lim\limits_{z\rightarrow 0}\left(-\frac{2}{1-e^{-2\pi iz}}+\frac{1}{\pi iz}\right),
\end{equation}
we have
\begin{equation}
\begin{aligned}
	&\ \bra{\widetilde{\psi}_\eta}
	\left(N\left(\sum_{x=0}^{\frac{n}{2}-1}e^{\frac{2\pi i x}{n}}\right)A_0^\dagger A_1
	-N\left(\sum_{x=0}^{\frac{n}{2}-1}e^{-\frac{2\pi i x}{n}}\right)A_1^\dagger A_0\right)
	\ket{\widetilde{\psi}_\eta}\\
	=&\ \eta\bra{\widetilde{\psi}_\eta}
	\left(\frac{2}{1-e^{\frac{2\pi i}{n}}}A_0^\dagger A_1
	-\frac{2}{1-e^{-\frac{2\pi i}{n}}}A_1^\dagger A_0\right)
	\ket{\widetilde{\psi}_\eta}\\
	=&-\frac{n\eta}{\pi i}\bra{\widetilde{\psi}_\eta}
	\left(A_0^\dagger A_1+A_1^\dagger A_0\right)
	\ket{\widetilde{\psi}_\eta}+\cO{\eta}
	=-\frac{n\eta}{\pi i}+\cO{\eta}.
\end{aligned}
\end{equation}
On the other hand,
\begin{equation}
\hspace*{-1.75cm}
\begin{aligned}
&\frac{1}{n}\bra{\widetilde{\psi}_\eta}
\left(-\sum_{k}
\left(\sum_{x=0}^{\frac{n}{2}-1}e^{\frac{2\pi i xk}{n}}\sum_{y=0}^{\frac{n}{2}-1}e^{\frac{2\pi i y(1-k)}{n}}\right)
A_k^\dagger A_k A_0^\dagger A_1
+\sum_{k}
\left(\sum_{x=0}^{\frac{n}{2}-1}e^{\frac{2\pi i x(k-1)}{n}}\sum_{y=0}^{\frac{n}{2}-1}e^{-\frac{2\pi i yk}{n}}\right)
A_k^\dagger A_k A_1^\dagger A_0\right)
\ket{\widetilde{\psi}_\eta}\\
=&\frac{1}{n}\bra{\widetilde{\psi}_\eta}
\left(-\sum_{k=n-\eta+1}^{n-1}
\frac{1-e^{\pi i k}}{1-e^{\frac{2\pi i k}{n}}}\frac{1-e^{\pi i (1-k)}}{1-e^{\frac{2\pi i (1-k)}{n}}}
A_0^\dagger A_1
+\sum_{k=n-\eta+1}^{n-1}
\frac{1-e^{\pi i (k-1)}}{1-e^{\frac{2\pi i (k-1)}{n}}}\frac{1-e^{-\pi i k}}{1-e^{-\frac{2\pi i k}{n}}}
A_1^\dagger A_0\right)\ket{\widetilde{\psi}_\eta}+\cO{n}
=\cO{n},
\end{aligned}
\end{equation}
where the last equality holds since for integer $k$ exactly one of $k$ and $k-1$ is even. We have thus proved
\begin{equation}
	\bra{\widetilde{\psi}_\eta}\left[\widetilde{T},\widetilde{V}\right]\ket{\widetilde{\psi}_\eta}
	=-\frac{n\eta}{\pi i}+\cO{n}.
\end{equation}

The above argument can be extended to analyze multilayer nested commutators. Indeed, for initial state
\begin{equation}
	\ket{\widetilde{\phi}_\eta}=\frac{\ket{010\cdots 0\overbrace{1\cdots 1}^{\eta-1}}+i\ket{100\cdots 0\overbrace{1\cdots 1}^{\eta-1}}}{\sqrt{2}},
\end{equation}
we have
\begin{equation}
\begin{aligned}
	\bra{\widetilde{\phi}_\eta}\left[\widetilde{T},\left[\widetilde{T},\widetilde{V}\right]\right]\ket{\widetilde{\phi}_\eta}
	&=-\frac{n^2\eta}{\pi i}\bra{\widetilde{\phi}_\eta}
	\left(A_0^\dagger A_1-A_1^\dagger A_0\right)
	\ket{\widetilde{\phi}_\eta}+\cO{n^2+n\eta}\\
	&=\frac{n^2\eta}{\pi}+\cO{n^2},
\end{aligned}
\end{equation}
and similar results hold for general nested commutators $[\widetilde{T},\ldots[\widetilde{T},\widetilde{V}]]$. This completes the proof of \prop{tightness_t}.

For sparse interactions, we have $u,w>0$, positive integer $2\leq d\leq\eta\leq\frac{n}{2}$ and consider the electronic Hamiltonian \eq{def_tv_sparse}. Similar to the above analysis, we compute the commutators by performing the fermionic Fourier transform, but only to the first $d$ spin orbitals, obtaining
\begin{equation}
\begin{aligned}
\widetilde{T}&=\mathrm{FFFT}_d^\dagger\cdot T\cdot\mathrm{FFFT}_d=udN_0,\\
\widetilde{V}&=\mathrm{FFFT}_d^\dagger\cdot V\cdot\mathrm{FFFT}_d
=\frac{w}{d^2}\sum_{j,k,l,m=0}^{d-1}
\left(\sum_{x=0}^{\frac{d}{2}-1}e^{\frac{2\pi i x(k-j)}{d}}\right)
\left(\sum_{y=0}^{\frac{d}{2}-1}e^{\frac{2\pi i y(m-l)}{d}}\right)
A_j^\dagger A_k A_l^\dagger A_m.
\end{aligned}
\end{equation}
We choose the initial state from the two-dimensional subspace spanned by
\begin{equation}
\ket{\widetilde{\psi}_{0,d}}=\ket{\overbrace{011\cdots 1}^{d}0\cdots0\overbrace{1\cdots 1}^{\eta-d+1}},\qquad
\ket{\widetilde{\psi}_{1,d}}=\ket{\overbrace{101\cdots 1}^{d}0\cdots0\overbrace{1\cdots 1}^{\eta-d+1}}
\end{equation}
and denote the projection to this subspace as $\widetilde{\Pi}_d=\ket{\widetilde{\psi}_{0,d}}\bra{\widetilde{\psi}_{0,d}}+\ket{\widetilde{\psi}_{1,d}}\bra{\widetilde{\psi}_{1,d}}$.
Then, the effective commutator $\widetilde{\Pi}_d\left[\widetilde{T},\widetilde{V}\right]\widetilde{\Pi}_d$ has action
\begin{equation}
\begin{aligned}
\widetilde{\Pi}_d\left[\widetilde{T},\widetilde{V}\right]\widetilde{\Pi}_d
=&\ uw\sum_{j=0}^{d-1}N_j\left(\sum_{x=0}^{\frac{d}{2}-1}e^{\frac{2\pi i x}{d}}\right)A_0^\dagger A_1
-\sum_{j=0}^{d-1}N_j\left(\sum_{x=0}^{\frac{d}{2}-1}e^{-\frac{2\pi i x}{d}}\right)A_1^\dagger A_0\\
&\ -\frac{2uw}{d}\sum_{k=0}^{d-1}
\left(\sum_{x=0}^{\frac{d}{2}-1}e^{\frac{2\pi i xk}{d}}\sum_{y=0}^{\frac{d}{2}-1}e^{\frac{2\pi i y(1-k)}{d}}\right)
A_k^\dagger A_k A_0^\dagger A_1\\
&\ +\frac{2uw}{d}\sum_{k=0}^{d-1}
\left(\sum_{x=0}^{\frac{d}{2}-1}e^{\frac{2\pi i x(k-1)}{d}}\sum_{y=0}^{\frac{d}{2}-1}e^{-\frac{2\pi i yk}{d}}\right)
A_k^\dagger A_k A_1^\dagger A_0
+\cO{uwd}.
\end{aligned}
\end{equation}
We then define the state
\begin{equation}
	\ket{\widetilde{\psi}_{\eta,d}}=\frac{\ket{\overbrace{011\cdots 1}^{d}0\cdots0\overbrace{1\cdots 1}^{\eta-d+1}}
		+\ket{\overbrace{101\cdots 1}^{d}0\cdots0\overbrace{1\cdots 1}^{\eta-d+1}}}{\sqrt{2}},
\end{equation}
so that
\begin{equation}
	\bra{\widetilde{\psi}_{\eta,d}}\left[\widetilde{T},\widetilde{V}\right]\ket{\widetilde{\psi}_{\eta,d}}
	=-\frac{uwd^2}{\pi i}+\cO{uwd}.
\end{equation}
The calculation can be extended to multilayer nested commutators, using either $\ket{\widetilde{\psi}_{\eta,d}}$ or
\begin{equation}
\ket{\widetilde{\phi}_{\eta,d}}=\frac{\ket{\overbrace{011\cdots 1}^{d}0\cdots0\overbrace{1\cdots 1}^{\eta-d+1}}
	+i\ket{\overbrace{101\cdots 1}^{d}0\cdots0\overbrace{1\cdots 1}^{\eta-d+1}}}{\sqrt{2}}
\end{equation}
as the initial state.

\section{Lower-bounding \texorpdfstring{$\norm{[V,\ldots[V,T]]}_{\eta}$}{nested commutators with V}}
\label{append:commutator_v}

In this appendix, we prove \prop{tightness_v} that lower-bounds $\norm{[V,\ldots[V,T]]}_{\eta}$ for the electronic Hamiltonian \eq{def_tv}. Recall that we have $H=T+V$ with
\begin{equation}
T=\sum_{j,k=0}^{n-1}A_j^\dagger A_k,\qquad
V=\sum_{x,y=0}^{\frac{n}{2}-1}N_{x}N_{y},
\end{equation}
which implies the commutator
\begin{equation}
[V,T]
=\sum_{x=0}^{\frac{n}{2}-1}N_x\bigg(\sum_{\substack{0\leq j\leq\frac{n}{2}-1\\ \frac{n}{2}\leq k\leq n-1}}-\sum_{\substack{\frac{n}{2}\leq j\leq n-1\\ 0\leq k\leq\frac{n}{2}-1}}\bigg)A_j^\dagger A_k
+\bigg(\sum_{\substack{0\leq j\leq\frac{n}{2}-1\\ \frac{n}{2}\leq k\leq n-1}}-\sum_{\substack{\frac{n}{2}\leq j\leq n-1\\ 0\leq k\leq\frac{n}{2}-1}}\bigg)A_j^\dagger A_k \sum_{y=0}^{\frac{n}{2}-1}N_y.
\end{equation}

For $\eta\leq\frac{n}{2}$, we will choose the initial state from the two-dimensional subspace spanned by
\begin{equation}
	\ket{{\psi}_0}=\ket{\overbrace{0\underbrace{1\cdots 1}_{\eta-1}0\cdots 0}^{\frac{n}{2}}10\cdots 0},\qquad
	\ket{{\psi}_1}=\ket{\overbrace{1\underbrace{1\cdots 1}_{\eta-1}0\cdots 0}^{\frac{n}{2}}00\cdots 0}.
\end{equation}
Denoting the projection to this subspace as ${\Pi}=\ket{{\psi}_0}\bra{{\psi}_0}+\ket{{\psi}_1}\bra{{\psi}_1}$, we have that ${\Pi}$ commutes with $\sum_{0\leq x\leq\frac{n}{2}-1}N_x$. Meanwhile,
\begin{equation}
	\Pi
	\bigg(\sum_{\substack{0\leq j\leq\frac{n}{2}-1\\ \frac{n}{2}\leq k\leq n-1}}-\sum_{\substack{\frac{n}{2}\leq j\leq n-1\\ 0\leq k\leq\frac{n}{2}-1}}\bigg)A_j^\dagger A_k
	\Pi
	=A_{0}^\dagger A_{\frac{n}{2}}-A_{\frac{n}{2}}^\dagger A_{0}.
\end{equation}
This shows that the effective commutator $\Pi[V,T]\Pi$ has the action
\begin{equation}
	\Pi[V,T]\Pi
	=\sum_{x=0}^{\frac{n}{2}-1}N_x\left(A_{0}^\dagger A_{\frac{n}{2}}-A_{\frac{n}{2}}^\dagger A_{0}\right)
	+\left(A_{0}^\dagger A_{\frac{n}{2}}-A_{\frac{n}{2}}^\dagger A_{0}\right)\sum_{y=0}^{\frac{n}{2}-1}N_y.
\end{equation}

We now take the expectation of this operator with respect to the state
\begin{equation}
	\ket{\psi_\eta}=\frac{\ket{\overbrace{0\underbrace{1\cdots 1}_{\eta-1}0\cdots 0}^{\frac{n}{2}}10\cdots 0}
		+i\ket{\overbrace{1\underbrace{1\cdots 1}_{\eta-1}0\cdots 0}^{\frac{n}{2}}00\cdots 0}}{\sqrt{2}},
\end{equation}
which gives
\begin{equation}
\begin{aligned}
	\bra{\psi_\eta}[V,T]\ket{\psi_\eta}
	&=\bra{\psi_\eta}
	\left(\sum_{x=0}^{\frac{n}{2}-1}N_x\left(A_{0}^\dagger A_{\frac{n}{2}}-A_{\frac{n}{2}}^\dagger A_{0}\right)
	+\left(A_{0}^\dagger A_{\frac{n}{2}}-A_{\frac{n}{2}}^\dagger A_{0}\right)\sum_{y=0}^{\frac{n}{2}-1}N_y\right)
	\ket{\psi_\eta}\\
	&=2\eta\bra{\psi_\eta}
	\left(A_{0}^\dagger A_{\frac{n}{2}}-A_{\frac{n}{2}}^\dagger A_{0}\right)
	\ket{\psi_\eta}+\cO{1}
	=(-1)^\eta 2i\eta+\cO{1}.
\end{aligned}
\end{equation}
This proves the desired scaling for the single-layer commutator. This argument can be extended to analyze multilayer nested commutators. Indeed, for initial state
\begin{equation}
	\ket{\phi_\eta}=\frac{\ket{\overbrace{0\underbrace{1\cdots 1}_{\eta-1}0\cdots 0}^{\frac{n}{2}}10\cdots 0}
		+\ket{\overbrace{1\underbrace{1\cdots 1}_{\eta-1}0\cdots 0}^{\frac{n}{2}}00\cdots 0}}{\sqrt{2}},
\end{equation}
we have
\begin{equation}
\begin{aligned}
	\bra{\phi_\eta}[V,[V,T]]\ket{\phi_\eta}
	&=\bra{\phi_\eta}\left(\sum_{x=0}^{\frac{n}{2}-1}N_x\right)^2\left(A_0^\dagger A_{\frac{n}{2}}+A_{\frac{n}{2}}^\dagger A_0\right)\ket{\phi_\eta}\\
	&\quad+2\bra{\phi_\eta}\left(\sum_{x=0}^{\frac{n}{2}-1}N_x\right)\left(A_0^\dagger A_{\frac{n}{2}}+A_{\frac{n}{2}}^\dagger A_0\right)\left(\sum_{y=0}^{\frac{n}{2}-1}N_y\right)\ket{\phi_\eta}\\
	&\quad+\bra{\phi_\eta}\left(A_0^\dagger A_{\frac{n}{2}}+A_{\frac{n}{2}}^\dagger A_0\right)\left(\sum_{y=0}^{\frac{n}{2}-1}N_y\right)^2\ket{\phi_\eta}
	=(-1)^{\eta-1}4\eta^2+\cO{\eta},
\end{aligned}
\end{equation}
and similar results hold for general nested commutators $[{V},\ldots[{V},{T}]]$. This completes the proof of \prop{tightness_v}.

For sparse interactions, we have $u,w>0$, positive integer $2\leq d\leq\eta\leq\frac{n}{2}$ and consider the electronic Hamiltonian \eq{def_tv_sparse}. Similar to above, we have the commutator
\begin{equation}
[V,T]
=uw\sum_{x=0}^{\frac{d}{2}-1}N_x\bigg(\sum_{\substack{0\leq j\leq\frac{d}{2}-1\\ \frac{d}{2}\leq k\leq d-1}}-\sum_{\substack{\frac{d}{2}\leq j\leq d-1\\ 0\leq k\leq\frac{d}{2}-1}}\bigg)A_j^\dagger A_k
+uw\bigg(\sum_{\substack{0\leq j\leq\frac{d}{2}-1\\ \frac{d}{2}\leq k\leq d-1}}-\sum_{\substack{\frac{d}{2}\leq j\leq d-1\\ 0\leq k\leq\frac{d}{2}-1}}\bigg)A_j^\dagger A_k \sum_{y=0}^{\frac{d}{2}-1}N_y.
\end{equation}
We choose the initial state from the two-dimensional subspace spanned by
\begin{equation}
\ket{{\psi}_{0,d}}=\ket{\overbrace{\underbrace{01\cdots1}_{\frac{d}{2}}10\cdots0}^{d}0\cdots0\overbrace{1\cdots1}^{\eta-\frac{d}{2}}},\qquad
\ket{{\psi}_{1,d}}=\ket{\overbrace{\underbrace{11\cdots1}_{\frac{d}{2}}00\cdots0}^{d}0\cdots0\overbrace{1\cdots1}^{\eta-\frac{d}{2}}}
\end{equation}
and denote the projection to this subspace as ${\Pi}_{d}=\ket{{\psi}_{0,d}}\bra{{\psi}_{0,d}}+\ket{{\psi}_{1,d}}\bra{{\psi}_{1,d}}$.
Then, the effective commutator $\Pi_d[V,T]\Pi_d$ has the action
\begin{equation}
\Pi_d[V,T]\Pi_d
=uw\sum_{x=0}^{\frac{d}{2}-1}N_x\left(A_{0}^\dagger A_{\frac{d}{2}}-A_{\frac{d}{2}}^\dagger A_{0}\right)
+uw\left(A_{0}^\dagger A_{\frac{d}{2}}-A_{\frac{d}{2}}^\dagger A_{0}\right)\sum_{y=0}^{\frac{d}{2}-1}N_y.
\end{equation}
Choosing the initial state
\begin{equation}
	\ket{\psi_{\eta,d}}=\frac{\ket{\overbrace{\underbrace{01\cdots1}_{\frac{d}{2}}10\cdots0}^{d}0\cdots0\overbrace{1\cdots1}^{\eta-\frac{d}{2}}}
		+i\ket{\overbrace{\underbrace{11\cdots1}_{\frac{d}{2}}00\cdots0}^{d}0\cdots0\overbrace{1\cdots1}^{\eta-\frac{d}{2}}}}{\sqrt{2}},
\end{equation}
we have
\begin{equation}
	\bra{\psi_{\eta,d}}[V,T]\ket{\psi_{\eta,d}}
	=(-1)^{\frac{d}{2}}iuwd+\cO{uw}.
\end{equation}
The calculation can be extended to multilayer nested commutators, using either $\ket{\psi_{\eta,d}}$ or
\begin{equation}
\ket{\phi_{\eta,d}}=\frac{\ket{\overbrace{\underbrace{01\cdots1}_{\frac{d}{2}}10\cdots0}^{d}0\cdots0\overbrace{1\cdots1}^{\eta-\frac{d}{2}}}
	+\ket{\overbrace{\underbrace{11\cdots1}_{\frac{d}{2}}00\cdots0}^{d}0\cdots0\overbrace{1\cdots1}^{\eta-\frac{d}{2}}}}{\sqrt{2}},
\end{equation}
as the initial state.

\bibliographystyle{plainnat}
\bibliography{TrotterElectron}

\begin{thebibliography}{88}
\providecommand{\natexlab}[1]{#1}
\providecommand{\url}[1]{\texttt{#1}}
\expandafter\ifx\csname urlstyle\endcsname\relax
  \providecommand{\doi}[1]{doi: #1}\else
  \providecommand{\doi}{doi: \begingroup \urlstyle{rm}\Url}\fi

\bibitem[Aharonov and Ta-Shma(2003)]{AT03}
Dorit Aharonov and Amnon Ta-Shma.
\newblock Adiabatic quantum state generation and statistical zero knowledge.
\newblock In \emph{Proceedings of the 35th ACM Symposium on Theory of
  Computing}, pages 20--29, 2003.
\newblock \doi{10.1145/780542.780546}.
\newblock arXiv:quant-ph/0301023.

\bibitem[An and Lin(2019)]{AL19}
Dong An and Lin Lin.
\newblock Quantum linear system solver based on time-optimal adiabatic quantum
  computing and quantum approximate optimization algorithm, 2019.
\newblock arXiv:1909.05500.

\bibitem[An et~al.(2021)An, Fang, and Lin]{AFL21}
Dong An, Di~Fang, and Lin Lin.
\newblock Time-dependent unbounded {H}amiltonian simulation with vector norm
  scaling.
\newblock \emph{{Quantum}}, 5:\penalty0 459, May 2021.
\newblock ISSN 2521-327X.
\newblock \doi{10.22331/q-2021-05-26-459}.
\newblock arXiv:2012.13105.

\bibitem[Aspuru-Guzik et~al.(2005)Aspuru-Guzik, Dutoi, Love, and
  Head-Gordon]{aspuru2005simulated}
Al{\'a}n Aspuru-Guzik, Anthony~D. Dutoi, Peter~J. Love, and Martin Head-Gordon.
\newblock Simulated quantum computation of molecular energies.
\newblock \emph{Science}, 309\penalty0 (5741):\penalty0 1704--1707, 2005.
\newblock \doi{10.1126/science.1113479}.
\newblock arXiv:quant-ph/0604193.

\bibitem[Babbush et~al.(2015)Babbush, McClean, Wecker, Aspuru-Guzik, and
  Wiebe]{babbush2015chemical}
Ryan Babbush, Jarrod McClean, Dave Wecker, Al{\'a}n Aspuru-Guzik, and Nathan
  Wiebe.
\newblock Chemical basis of {T}rotter-{S}uzuki errors in quantum chemistry
  simulation.
\newblock \emph{Physical Review A}, 91\penalty0 (2):\penalty0 022311, 2015.
\newblock \doi{10.1103/PhysRevA.91.022311}.
\newblock arXiv:1410.8159.

\bibitem[Babbush et~al.(2018)Babbush, Wiebe, McClean, McClain, Neven, and
  Chan]{BWMMNC18}
Ryan Babbush, Nathan Wiebe, Jarrod McClean, James McClain, Hartmut Neven, and
  Garnet Kin-Lic Chan.
\newblock {Low-depth quantum simulation of materials}.
\newblock \emph{Physical Review X}, 8:\penalty0 011044, Mar 2018.
\newblock \doi{10.1103/PhysRevX.8.011044}.
\newblock arXiv:1706.00023.

\bibitem[Babbush et~al.(2019)Babbush, Berry, McClean, and Neven]{Babbush2019}
Ryan Babbush, Dominic~W. Berry, Jarrod~R. McClean, and Hartmut Neven.
\newblock Quantum simulation of chemistry with sublinear scaling in basis size.
\newblock \emph{npj Quantum Information}, 5\penalty0 (1):\penalty0 92, Nov
  2019.
\newblock ISSN 2056-6387.
\newblock \doi{10.1038/s41534-019-0199-y}.
\newblock arXiv:1807.09802.

\bibitem[Bauer et~al.(2020)Bauer, Bravyi, Motta, and Chan]{bauer2020quantum}
Bela Bauer, Sergey Bravyi, Mario Motta, and Garnet~Kin Chan.
\newblock Quantum algorithms for quantum chemistry and quantum materials
  science.
\newblock \emph{Chemical Reviews}, 120\penalty0 (22):\penalty0 12685--12717,
  2020.
\newblock \doi{10.1021/acs.chemrev.9b00829}.
\newblock arXiv:2001.03685.

\bibitem[Berry(2014)]{Berry14}
Dominic~W. Berry.
\newblock High-order quantum algorithm for solving linear differential
  equations.
\newblock \emph{Journal of Physics A: Mathematical and Theoretical},
  47\penalty0 (10):\penalty0 105301, feb 2014.
\newblock \doi{10.1088/1751-8113/47/10/105301}.
\newblock arXiv:1010.2745.

\bibitem[Berry et~al.(2007)Berry, Ahokas, Cleve, and
  Sanders]{berry2007efficient}
Dominic~W. Berry, Graeme Ahokas, Richard Cleve, and Barry~C. Sanders.
\newblock Efficient quantum algorithms for simulating sparse {H}amiltonians.
\newblock \emph{Communications in Mathematical Physics}, 270\penalty0
  (2):\penalty0 359--371, 2007.
\newblock \doi{10.1007/s00220-006-0150-x}.
\newblock arXiv:quant-ph/0508139.

\bibitem[Berry et~al.(2014)Berry, Childs, Cleve, Kothari, and
  Somma]{FractionalQuery14}
Dominic~W. Berry, Andrew~M. Childs, Richard Cleve, Robin Kothari, and
  Rolando~D. Somma.
\newblock Exponential improvement in precision for simulating sparse
  {H}amiltonians.
\newblock In \emph{Proceedings of the 46th Annual ACM Symposium on Theory of
  Computing}, pages 283--292, 2014.
\newblock \doi{10.1145/2591796.2591854}.
\newblock arXiv:1312.1414.

\bibitem[Berry et~al.(2015{\natexlab{a}})Berry, Childs, Cleve, Kothari, and
  Somma]{Berry15}
Dominic~W. Berry, Andrew~M. Childs, Richard Cleve, Robin Kothari, and
  Rolando~D. Somma.
\newblock Simulating {H}amiltonian dynamics with a truncated {T}aylor series.
\newblock \emph{Physical Review Letters}, 114\penalty0 (9):\penalty0 090502,
  2015{\natexlab{a}}.
\newblock \doi{10.1103/PhysRevLett.114.090502}.
\newblock arXiv:1412.4687.

\bibitem[Berry et~al.(2015{\natexlab{b}})Berry, Childs, and Kothari]{BCK15}
Dominic~W. Berry, Andrew~M. Childs, and Robin Kothari.
\newblock {H}amiltonian simulation with nearly optimal dependence on all
  parameters.
\newblock In \emph{Proceedings of the 56th IEEE Symposium on Foundations of
  Computer Science}, pages 792--809, 2015{\natexlab{b}}.
\newblock \doi{10.1109/FOCS.2015.54}.
\newblock arXiv:1501.01715.

\bibitem[Berry et~al.(2019)Berry, Gidney, Motta, McClean, and
  Babbush]{Berry2019qubitizationof}
Dominic~W. Berry, Craig Gidney, Mario Motta, Jarrod~R. McClean, and Ryan
  Babbush.
\newblock Qubitization of arbitrary basis quantum chemistry leveraging sparsity
  and low rank factorization.
\newblock \emph{{Quantum}}, 3:\penalty0 208, December 2019.
\newblock ISSN 2521-327X.
\newblock \doi{10.22331/q-2019-12-02-208}.
\newblock arXiv:1902.02134.

\bibitem[Brandao and Svore(2017)]{BS17}
Fernando G. S.~L. Brandao and Krysta~M. Svore.
\newblock Quantum speed-ups for solving semidefinite programs.
\newblock In \emph{Proceedings of the 58th IEEE Symposium on Foundations of
  Computer Science}, pages 415--426, 2017.
\newblock \doi{10.1109/FOCS.2017.45}.
\newblock arXiv:1609.05537.

\bibitem[Cade et~al.(2020)Cade, Mineh, Montanaro, and
  Stanisic]{cade2019strategies}
Chris Cade, Lana Mineh, Ashley Montanaro, and Stasja Stanisic.
\newblock Strategies for solving the {F}ermi-{H}ubbard model on near-term
  quantum computers.
\newblock \emph{Physical Review B}, 102:\penalty0 235122, Dec 2020.
\newblock \doi{10.1103/PhysRevB.102.235122}.
\newblock arXiv:1912.06007.

\bibitem[Cai(2020)]{Cai20}
Zhenyu Cai.
\newblock Resource estimation for quantum variational simulations of the
  {H}ubbard model.
\newblock \emph{Physical Review Applied}, 14:\penalty0 014059, Jul 2020.
\newblock \doi{10.1103/PhysRevApplied.14.014059}.
\newblock arXiv:1910.02719.

\bibitem[Campbell(2019)]{campbell2019random}
Earl Campbell.
\newblock Random compiler for fast {H}amiltonian simulation.
\newblock \emph{Physical Review Letters}, 123:\penalty0 070503, Aug 2019.
\newblock \doi{10.1103/PhysRevLett.123.070503}.
\newblock arXiv:1811.08017.

\bibitem[Campbell(2020)]{Campbell20}
Earl~T. Campbell.
\newblock Early fault-tolerant simulations of the {H}ubbard model, 2020.
\newblock arXiv:2012.09238.

\bibitem[Cao et~al.(2019)Cao, Romero, Olson, Degroote, Johnson, Kieferov\'{a},
  Kivlichan, Menke, Peropadre, Sawaya, Sim, Veis, and
  Aspuru-Guzik]{cao2019quantum}
Yudong Cao, Jonathan Romero, Jonathan~P. Olson, Matthias Degroote, Peter~D.
  Johnson, M\'{a}ria Kieferov\'{a}, Ian~D. Kivlichan, Tim Menke, Borja
  Peropadre, Nicolas P.~D. Sawaya, Sukin Sim, Libor Veis, and Al{\'a}n
  Aspuru-Guzik.
\newblock Quantum chemistry in the age of quantum computing.
\newblock \emph{Chemical Reviews}, 119\penalty0 (19):\penalty0 10856--10915,
  2019.
\newblock \doi{10.1021/acs.chemrev.8b00803}.
\newblock arXiv:1812.09976.

\bibitem[Chen et~al.(2020)Chen, Huang, Kueng, and Tropp]{CHKT20}
Chi-Fang Chen, Hsin-Yuan Huang, Richard Kueng, and Joel~A. Tropp.
\newblock Quantum simulation via randomized product formulas: Low gate
  complexity with accuracy guarantees, 2020.
\newblock arXiv:2008.11751.

\bibitem[Childs and Su(2019)]{CS19}
Andrew~M. Childs and Yuan Su.
\newblock Nearly optimal lattice simulation by product formulas.
\newblock \emph{Physical Review Letters}, 123:\penalty0 050503, Aug 2019.
\newblock \doi{10.1103/PhysRevLett.123.050503}.
\newblock arXiv:1901.00564.

\bibitem[Childs et~al.(2003)Childs, Cleve, Deotto, Farhi, Gutmann, and
  Spielman]{CCDFGS03}
Andrew~M. Childs, Richard Cleve, Enrico Deotto, Edward Farhi, Sam Gutmann, and
  Daniel~A. Spielman.
\newblock Exponential algorithmic speedup by quantum walk.
\newblock In \emph{Proceedings of the 35th ACM Symposium on Theory of
  Computing}, pages 59--68, 2003.
\newblock \doi{10.1145/780542.780552}.
\newblock arXiv:quant-ph/0209131.

\bibitem[Childs et~al.(2018)Childs, Maslov, Nam, Ross, and
  Su]{childs2018toward}
Andrew~M. Childs, Dmitri Maslov, Yunseong Nam, Neil~J. Ross, and Yuan Su.
\newblock Toward the first quantum simulation with quantum speedup.
\newblock \emph{Proceedings of the National Academy of Sciences}, 115\penalty0
  (38):\penalty0 9456--9461, 2018.
\newblock \doi{10.1073/pnas.1801723115}.
\newblock arXiv:0905.0887.

\bibitem[Childs et~al.(2019{\natexlab{a}})Childs, Ostrander, and
  Su]{childs2018faster}
Andrew~M. Childs, Aaron Ostrander, and Yuan Su.
\newblock Faster quantum simulation by randomization.
\newblock \emph{{Quantum}}, 3:\penalty0 182, September 2019{\natexlab{a}}.
\newblock ISSN 2521-327X.
\newblock \doi{10.22331/q-2019-09-02-182}.
\newblock arXiv:1805.08385.

\bibitem[Childs et~al.(2019{\natexlab{b}})Childs, Su, Tran, Wiebe, and
  Zhu]{CSTWZ19arXiv}
Andrew~M. Childs, Yuan Su, Minh~C.\ Tran, Nathan Wiebe, and Shuchen Zhu.
\newblock A theory of {T}rotter error, 2019{\natexlab{b}}.
\newblock arXiv:1912.08854.

\bibitem[Childs et~al.(2021)Childs, Su, Tran, Wiebe, and Zhu]{CSTWZ19}
Andrew~M. Childs, Yuan Su, Minh~C. Tran, Nathan Wiebe, and Shuchen Zhu.
\newblock Theory of {T}rotter error with commutator scaling.
\newblock \emph{Physical Review X}, 11:\penalty0 011020, Feb 2021.
\newblock \doi{10.1103/PhysRevX.11.011020}.

\bibitem[Clinton et~al.(2020)Clinton, Bausch, and Cubitt]{CBC20}
Laura Clinton, Johannes Bausch, and Toby Cubitt.
\newblock Hamiltonian simulation algorithms for near-term quantum hardware,
  2020.
\newblock arXiv:2003.06886.

\bibitem[Ferris(2014)]{Ferris14}
Andrew~J. Ferris.
\newblock Fourier transform for fermionic systems and the spectral tensor
  network.
\newblock \emph{Physical Review Letters}, 113:\penalty0 010401, Jul 2014.
\newblock \doi{10.1103/PhysRevLett.113.010401}.
\newblock arXiv:1310.7605.

\bibitem[Feynman(1982)]{Fey82}
Richard~P. Feynman.
\newblock Simulating physics with computers.
\newblock \emph{International Journal of Theoretical Physics}, 21\penalty0
  (6-7):\penalty0 467--488, 1982.
\newblock \doi{10.1007/BF02650179}.

\bibitem[Gharibyan et~al.(2020)Gharibyan, Hanada, Honda, and Liu]{GHHL20}
Hrant Gharibyan, Masanori Hanada, Masazumi Honda, and Junyu Liu.
\newblock Toward simulating superstring/{M}-theory on a quantum computer, 2020.
\newblock arXiv:2011.06573.

\bibitem[Gily{\'e}n et~al.(2019)Gily{\'e}n, Su, Low, and Wiebe]{GSLW19}
Andr\'{a}s Gily{\'e}n, Yuan Su, Guang~Hao Low, and Nathan Wiebe.
\newblock Quantum singular value transformation and beyond: Exponential
  improvements for quantum matrix arithmetics.
\newblock In \emph{Proceedings of the 51st Annual ACM SIGACT Symposium on
  Theory of Computing}, pages 193--204, 2019.
\newblock ISBN 978-1-4503-6705-9.
\newblock \doi{10.1145/3313276.3316366}.
\newblock arXiv:1806.01838.

\bibitem[Haah et~al.(2018)Haah, Hastings, Kothari, and Low]{haah2018quantum}
Jeongwan Haah, Matthew~B. Hastings, Robin Kothari, and Guang~Hao Low.
\newblock Quantum algorithm for simulating real time evolution of lattice
  {H}amiltonians.
\newblock In \emph{Proceedings of the 59th IEEE Symposium on Foundations of
  Computer Science}, pages 350--360, 2018.
\newblock \doi{10.1109/FOCS.2018.00041}.
\newblock arXiv:1801.03922.

\bibitem[Hadfield and Papageorgiou(2018)]{HP18}
Stuart Hadfield and Anargyros Papageorgiou.
\newblock Divide and conquer approach to quantum {H}amiltonian simulation.
\newblock \emph{New Journal of Physics}, 20\penalty0 (4):\penalty0 043003,
  2018.
\newblock \doi{10.1088/1367-2630/aab1ef}.

\bibitem[Halimeh et~al.(2020)Halimeh, Lang, Mildenberger, Jiang, and
  Hauke]{HLMJH20}
Jad~C. Halimeh, Haifeng Lang, Julius Mildenberger, Zhang Jiang, and Philipp
  Hauke.
\newblock Gauge-symmetry protection using single-body terms, 2020.
\newblock arXiv:2007.00668.

\bibitem[Harrow et~al.(2009)Harrow, Hassidim, and Lloyd]{HHL09}
Aram~W. Harrow, Avinatan Hassidim, and Seth Lloyd.
\newblock Quantum algorithm for linear systems of equations.
\newblock \emph{Physical Review Letters}, 103\penalty0 (15):\penalty0 150502,
  2009.
\newblock \doi{10.1103/PhysRevLett.103.150502}.
\newblock arXiv:0811.3171.

\bibitem[Helgaker et~al.(2013)Helgaker, J\o{}rgensen, and
  Olsen]{helgaker2014molecular}
Trygve Helgaker, Poul J\o{}rgensen, and Jeppe Olsen.
\newblock \emph{Molecular electronic-structure theory}.
\newblock John Wiley \& Sons, 2013.
\newblock \doi{10.1002/9781119019572}.

\bibitem[Horn and Johnson(2012)]{horn2012matrix}
Roger~A. Horn and Charles~R. Johnson.
\newblock \emph{Matrix analysis}.
\newblock Cambridge university press, 2012.
\newblock \doi{10.1017/CBO9781139020411}.

\bibitem[Jordan et~al.(2012)Jordan, Lee, and Preskill]{JLP12}
Stephen~P. Jordan, Keith S.~M. Lee, and John Preskill.
\newblock Quantum algorithms for quantum field theories.
\newblock \emph{Science}, 336\penalty0 (6085):\penalty0 1130--1133, 2012.
\newblock \doi{10.1126/science.1217069}.
\newblock arXiv:1111.3633.

\bibitem[Jordan et~al.(2014)Jordan, Lee, and Preskill]{JLP14b}
Stephen~P. Jordan, Keith S.~M. Lee, and John Preskill.
\newblock Quantum algorithms for fermionic quantum field theories, 2014.
\newblock arXiv:1404.7115.

\bibitem[Kivlichan et~al.(2018)Kivlichan, McClean, Wiebe, Gidney, Aspuru-Guzik,
  Chan, and Babbush]{kivlichan2018quantum}
Ian~D. Kivlichan, Jarrod McClean, Nathan Wiebe, Craig Gidney, Al{\'a}n
  Aspuru-Guzik, Garnet Kin-Lic Chan, and Ryan Babbush.
\newblock Quantum simulation of electronic structure with linear depth and
  connectivity.
\newblock \emph{Physical Review Letters}, 120\penalty0 (11):\penalty0 110501,
  2018.
\newblock \doi{10.1103/PhysRevLett.120.110501}.
\newblock arXiv:1711.04789.

\bibitem[Kivlichan et~al.(2020)Kivlichan, Gidney, Berry, Wiebe, McClean, Sun,
  Jiang, Rubin, Fowler, Aspuru-Guzik, Neven, and
  Babbush]{kivlichan2020improved}
Ian~D. Kivlichan, Craig Gidney, Dominic~W. Berry, Nathan Wiebe, Jarrod McClean,
  Wei Sun, Zhang Jiang, Nicholas Rubin, Austin Fowler, Al{\'a}n Aspuru-Guzik,
  Hartmut Neven, and Ryan Babbush.
\newblock Improved fault-tolerant quantum simulation of condensed-phase
  correlated electrons via {T}rotterization.
\newblock \emph{Quantum}, 4:\penalty0 296, 2020.
\newblock \doi{10.22331/q-2020-07-16-296}.
\newblock arXiv:1902.10673.

\bibitem[Kuwahara et~al.(2021)Kuwahara, Alhambra, and Anshu]{KAA20}
Tomotaka Kuwahara, \'Alvaro~M. Alhambra, and Anurag Anshu.
\newblock Improved thermal area law and quasilinear time algorithm for quantum
  {G}ibbs states.
\newblock \emph{Physical Review X}, 11:\penalty0 011047, Mar 2021.
\newblock \doi{10.1103/PhysRevX.11.011047}.
\newblock arXiv:2007.11174.

\bibitem[LeBlanc et~al.(2015)LeBlanc, Antipov, Becca, Bulik, Chan, Chung, Deng,
  Ferrero, Henderson, Jim\'enez-Hoyos, Kozik, Liu, Millis, Prokof'ev, Qin,
  Scuseria, Shi, Svistunov, Tocchio, Tupitsyn, White, Zhang, Zheng, Zhu, and
  Gull]{HalfPercent}
J.~P.~F. LeBlanc, Andrey~E. Antipov, Federico Becca, Ireneusz~W. Bulik, Garnet
  Kin-Lic Chan, Chia-Min Chung, Youjin Deng, Michel Ferrero, Thomas~M.
  Henderson, Carlos~A. Jim\'enez-Hoyos, E.~Kozik, Xuan-Wen Liu, Andrew~J.
  Millis, N.~V. Prokof'ev, Mingpu Qin, Gustavo~E. Scuseria, Hao Shi, B.~V.
  Svistunov, Luca~F. Tocchio, I.~S. Tupitsyn, Steven~R. White, Shiwei Zhang,
  Bo-Xiao Zheng, Zhenyue Zhu, and Emanuel Gull.
\newblock Solutions of the two-dimensional hubbard model: Benchmarks and
  results from a wide range of numerical algorithms.
\newblock \emph{Physical Review X}, 5:\penalty0 041041, Dec 2015.
\newblock \doi{10.1103/PhysRevX.5.041041}.
\newblock arXiv:1505.02290.

\bibitem[Lee et~al.(2020)Lee, Berry, Gidney, Huggins, McClean, Wiebe, and
  Babbush]{LBGHMWB20}
Joonho Lee, Dominic~W. Berry, Craig Gidney, William~J. Huggins, Jarrod~R.
  McClean, Nathan Wiebe, and Ryan Babbush.
\newblock Even more efficient quantum computations of chemistry through tensor
  hypercontraction, 2020.
\newblock arXiv:2011.03494.

\bibitem[Lin and Tong(2020{\natexlab{a}})]{LT20}
Lin Lin and Yu~Tong.
\newblock Near-optimal ground state preparation.
\newblock \emph{{Quantum}}, 4:\penalty0 372, December 2020{\natexlab{a}}.
\newblock ISSN 2521-327X.
\newblock \doi{10.22331/q-2020-12-14-372}.
\newblock arXiv:2002.12508.

\bibitem[Lin and Tong(2020{\natexlab{b}})]{Lin2020optimalpolynomial}
Lin Lin and Yu~Tong.
\newblock Optimal polynomial based quantum eigenstate filtering with
  application to solving quantum linear systems.
\newblock \emph{{Quantum}}, 4:\penalty0 361, November 2020{\natexlab{b}}.
\newblock ISSN 2521-327X.
\newblock \doi{10.22331/q-2020-11-11-361}.
\newblock arXiv:1910.14596.

\bibitem[Linke et~al.(2018)Linke, Johri, Figgatt, Landsman, Matsuura, and
  Monroe]{Linke18}
Norbert~M. Linke, Sonika Johri, Caroline Figgatt, Kevin~A. Landsman, Anne~Y.
  Matsuura, and Christopher Monroe.
\newblock Measuring the {R}\'enyi entropy of a two-site {F}ermi-{H}ubbard model
  on a trapped ion quantum computer.
\newblock \emph{Physical Review A}, 98:\penalty0 052334, Nov 2018.
\newblock \doi{10.1103/PhysRevA.98.052334}.
\newblock arXiv:1712.08581.

\bibitem[Liu et~al.(2020)Liu, Hines, Li, Ajoy, and Cappellaro]{LHLAC20}
Yi-Xiang Liu, Jordan Hines, Zhi Li, Ashok Ajoy, and Paola Cappellaro.
\newblock High-fidelity trotter formulas for digital quantum simulation.
\newblock \emph{Physical Review A}, 102:\penalty0 010601, Jul 2020.
\newblock \doi{10.1103/PhysRevA.102.010601}.
\newblock arXiv:1903.01654.

\bibitem[Lloyd(1996)]{lloyd1996universal}
Seth Lloyd.
\newblock Universal quantum simulators.
\newblock \emph{Science}, pages 1073--1078, 1996.
\newblock \doi{10.1126/science.273.5278.1073}.

\bibitem[Low and Chuang(2017{\natexlab{a}})]{LC17}
Guang~Hao Low and Isaac~L. Chuang.
\newblock Hamiltonian simulation by uniform spectral amplification,
  2017{\natexlab{a}}.
\newblock arXiv:1707.05391.

\bibitem[Low and Chuang(2017{\natexlab{b}})]{QSP17}
Guang~Hao Low and Isaac~L. Chuang.
\newblock Optimal {H}amiltonian simulation by quantum signal processing.
\newblock \emph{Physical Review Letters}, 118:\penalty0 010501,
  2017{\natexlab{b}}.
\newblock \doi{10.1103/PhysRevLett.118.010501}.
\newblock arXiv:1606.02685.

\bibitem[Low and Chuang(2019)]{Low2019hamiltonian}
Guang~Hao Low and Isaac~L. Chuang.
\newblock Hamiltonian simulation by qubitization.
\newblock \emph{{Quantum}}, 3:\penalty0 163, July 2019.
\newblock \doi{10.22331/q-2019-07-12-163}.
\newblock arXiv:1610.06546.

\bibitem[Low and Wiebe(2018)]{LW18}
Guang~Hao Low and Nathan Wiebe.
\newblock {H}amiltonian simulation in the interaction picture, 2018.
\newblock arXiv:1805.00675.

\bibitem[Low et~al.(2019)Low, Kliuchnikov, and Wiebe]{LKW19}
Guang~Hao Low, Vadym Kliuchnikov, and Nathan Wiebe.
\newblock Well-conditioned multiproduct {H}amiltonian simulation, 2019.
\newblock arXiv:1907.11679.

\bibitem[McArdle et~al.(2020)McArdle, Endo, Aspuru-Guzik, Benjamin, and
  Yuan]{mcardle2020quantum}
Sam McArdle, Suguru Endo, Al{\'a}n Aspuru-Guzik, Simon~C. Benjamin, and Xiao
  Yuan.
\newblock Quantum computational chemistry.
\newblock \emph{Reviews of Modern Physics}, 92\penalty0 (1):\penalty0 015003,
  2020.
\newblock \doi{10.1103/RevModPhys.92.015003}.
\newblock arXiv:1808.10402.

\bibitem[McClean et~al.(2014)McClean, Babbush, Love, and Aspuru-Guzik]{MBLA14}
Jarrod~R. McClean, Ryan Babbush, Peter~J. Love, and Al{\'a}n Aspuru-Guzik.
\newblock Exploiting locality in quantum computation for quantum chemistry.
\newblock \emph{The Journal of Physical Chemistry Letters}, 5\penalty0
  (24):\penalty0 4368--4380, 2014.
\newblock \doi{10.1021/jz501649m}.
\newblock arXiv:1407.7863.

\bibitem[Meister et~al.(2020)Meister, Benjamin, and
  Campbell]{meister2020tailoring}
Richard Meister, Simon~C. Benjamin, and Earl~T. Campbell.
\newblock Tailoring term truncations for electronic structure calculations
  using a linear combination of unitaries, 2020.
\newblock arXiv:2007.11624.

\bibitem[Motta et~al.(2021)Motta, Ye, McClean, Li, Minnich, Babbush, and
  Chan]{MYMLMBC13}
Mario Motta, Erika Ye, Jarrod~R. McClean, Zhendong Li, Austin~J. Minnich, Ryan
  Babbush, and Garnet Kin-Lic Chan.
\newblock Low rank representations for quantum simulation of electronic
  structure.
\newblock \emph{npj Quantum Information}, 7\penalty0 (1):\penalty0 83, May
  2021.
\newblock ISSN 2056-6387.
\newblock \doi{10.1038/s41534-021-00416-z}.
\newblock arXiv:1312.2579.

\bibitem[Ortiz et~al.(2001)Ortiz, Gubernatis, Knill, and Laflamme]{OGKL01}
G.~Ortiz, J.~E. Gubernatis, E.~Knill, and R.~Laflamme.
\newblock Quantum algorithms for fermionic simulations.
\newblock \emph{Physical Review A}, 64:\penalty0 022319, Jul 2001.
\newblock \doi{10.1103/PhysRevA.64.022319}.
\newblock arXiv:cond-mat/0012334.

\bibitem[Otte(2010)]{Otte10}
Peter Otte.
\newblock Boundedness properties of fermionic operators.
\newblock \emph{Journal of Mathematical Physics}, 51\penalty0 (8):\penalty0
  083503, 2010.
\newblock \doi{10.1063/1.3464264}.
\newblock arXiv:0911.4438.

\bibitem[Ouyang et~al.(2020)Ouyang, White, and Campbell]{ouyang2020compilation}
Yingkai Ouyang, David~R. White, and Earl~T. Campbell.
\newblock Compilation by stochastic {H}amiltonian sparsification.
\newblock \emph{Quantum}, 4:\penalty0 235, 2020.
\newblock \doi{10.22331/q-2020-02-27-235}.
\newblock arXiv:1910.06255.

\bibitem[Peng et~al.(2020)Peng, Harrow, Ozols, and Wu]{PHOW20}
Tianyi Peng, Aram~W. Harrow, Maris Ozols, and Xiaodi Wu.
\newblock Simulating large quantum circuits on a small quantum computer.
\newblock \emph{Physical Review Letters}, 125:\penalty0 150504, Oct 2020.
\newblock \doi{10.1103/PhysRevLett.125.150504}.
\newblock arXiv:1904.00102.

\bibitem[Peskin and Schroeder(2018)]{peskin2019introduction}
Michael~E. Peskin and Daniel~V. Schroeder.
\newblock \emph{An introduction to quantum field theory}.
\newblock CRC press, 2018.
\newblock \doi{10.1201/9780429503559}.

\bibitem[Poulin et~al.(2015)Poulin, Hastings, Wecker, Wiebe, Doherty, and
  Troyer]{Pou15}
David Poulin, Matthew~B. Hastings, Dave Wecker, Nathan Wiebe, Andrew~C.
  Doherty, and Matthias Troyer.
\newblock The {T}rotter step size required for accurate quantum simulation of
  quantum chemistry.
\newblock \emph{Quantum Information and Computation}, 15\penalty0
  (5-6):\penalty0 361--384, 2015.
\newblock arXiv:1406.4920.

\bibitem[Quantum and collaborators(2020)]{GoogleHubbard20}
Google~AI Quantum and collaborators.
\newblock Observation of separated dynamics of charge and spin in the
  {F}ermi-{H}ubbard model, 2020.
\newblock arXiv:2010.07965.

\bibitem[Rall(2021)]{Rall21}
Patrick Rall.
\newblock Faster coherent quantum algorithms for phase, energy, and amplitude
  estimation, 2021.
\newblock arXiv:2103.09717.

\bibitem[Reiher et~al.(2017)Reiher, Wiebe, Svore, Wecker, and Troyer]{RWSWT17}
Markus Reiher, Nathan Wiebe, Krysta~M. Svore, Dave Wecker, and Matthias Troyer.
\newblock Elucidating reaction mechanisms on quantum computers.
\newblock \emph{Proceedings of the National Academy of Sciences}, 114\penalty0
  (29):\penalty0 7555--7560, 2017.
\newblock \doi{10.1073/pnas.1619152114}.
\newblock arXiv:1605.03590.

\bibitem[{\c{S}}ahino\u{g}lu and Somma(2020)]{SS20}
Burak {\c{S}}ahino\u{g}lu and Rolando~D. Somma.
\newblock Hamiltonian simulation in the low energy subspace, 2020.
\newblock arXiv:2006.02660.

\bibitem[Sawaya et~al.(2020)Sawaya, Menke, Kyaw, Johri, Aspuru-Guzik, and
  Guerreschi]{Sawaya2020}
Nicolas P.~D. Sawaya, Tim Menke, Thi~Ha Kyaw, Sonika Johri, Al{\'a}n
  Aspuru-Guzik, and Gian~Giacomo Guerreschi.
\newblock Resource-efficient digital quantum simulation of $d$-level systems
  for photonic, vibrational, and spin-$s$ {H}amiltonians.
\newblock \emph{npj Quantum Information}, 6\penalty0 (1):\penalty0 49, Jun
  2020.
\newblock ISSN 2056-6387.
\newblock \doi{10.1038/s41534-020-0278-0}.
\newblock arXiv:1909.12847.

\bibitem[Seeley et~al.(2012)Seeley, Richard, and Love]{SRL12}
Jacob~T. Seeley, Martin~J. Richard, and Peter~J. Love.
\newblock The {B}ravyi-{K}itaev transformation for quantum computation of
  electronic structure.
\newblock \emph{The Journal of Chemical Physics}, 137\penalty0 (22):\penalty0
  224109, 2012.
\newblock \doi{10.1063/1.4768229}.
\newblock arXiv:1208.5986.

\bibitem[Shaw et~al.(2020)Shaw, Lougovski, Stryker, and
  Wiebe]{Shaw2020quantumalgorithms}
Alexander~F. Shaw, Pavel Lougovski, Jesse~R. Stryker, and Nathan Wiebe.
\newblock Quantum algorithms for simulating the lattice {S}chwinger model.
\newblock \emph{{Quantum}}, 4:\penalty0 306, August 2020.
\newblock ISSN 2521-327X.
\newblock \doi{10.22331/q-2020-08-10-306}.
\newblock arXiv:2002.11146.

\bibitem[Somma(2015)]{Somma15}
Rolando~D. Somma.
\newblock Quantum simulations of one dimensional quantum systems, 2015.
\newblock arXiv:1503.06319.

\bibitem[Somma(2016)]{Somma16}
Rolando~D. Somma.
\newblock A {T}rotter-{S}uzuki approximation for {L}ie groups with applications
  to {H}amiltonian simulation.
\newblock \emph{Journal of Mathematical Physics}, 57:\penalty0 062202, 2016.
\newblock \doi{10.1063/1.4952761}.
\newblock arXiv:1512.03416.

\bibitem[Su et~al.(2021)Su, Berry, Wiebe, Rubin, and Babbush]{SBWRB21}
Yuan Su, Dominic~W. Berry, Nathan Wiebe, Nicholas Rubin, and Ryan Babbush.
\newblock Fault-tolerant quantum simulations of chemistry in first
  quantization, 2021.
\newblock arXiv:2105.12767.

\bibitem[Suzuki(1985)]{Suzuki85}
Masuo Suzuki.
\newblock Decomposition formulas of exponential operators and {L}ie
  exponentials with some applications to quantum mechanics and statistical
  physics.
\newblock \emph{Journal of Mathematical Physics}, 26\penalty0 (4):\penalty0
  601--612, 1985.
\newblock \doi{10.1063/1.526596}.

\bibitem[Suzuki(1990)]{suzuki1990fractal}
Masuo Suzuki.
\newblock Fractal decomposition of exponential operators with applications to
  many-body theories and {M}onte {C}arlo simulations.
\newblock \emph{Physics Letters A}, 146\penalty0 (6):\penalty0 319--323, 1990.
\newblock \doi{10.1016/0375-9601(90)90962-N}.

\bibitem[Toloui and Love(2013)]{TL13}
Borzu Toloui and Peter~J. Love.
\newblock Quantum algorithms for quantum chemistry based on the sparsity of the
  {CI}-matrix, 2013.
\newblock arXiv:1312.2579.

\bibitem[Tran et~al.(2019)Tran, Guo, Su, Garrison, Eldredge, Foss-Feig, Childs,
  and Gorshkov]{Tran18}
Minh~C. Tran, Andrew~Y. Guo, Yuan Su, James~R. Garrison, Zachary Eldredge,
  Michael Foss-Feig, Andrew~M. Childs, and Alexey~V. Gorshkov.
\newblock Locality and digital quantum simulation of power-law interactions.
\newblock \emph{Physical Review X}, 9:\penalty0 031006, Jul 2019.
\newblock \doi{10.1103/PhysRevX.9.031006}.
\newblock arXiv:1808.05225.

\bibitem[Tran et~al.(2021)Tran, Su, Carney, and Taylor]{TSCT20}
Minh~C. Tran, Yuan Su, Daniel Carney, and Jacob~M. Taylor.
\newblock Faster digital quantum simulation by symmetry protection.
\newblock \emph{PRX Quantum}, 2:\penalty0 010323, Feb 2021.
\newblock \doi{10.1103/PRXQuantum.2.010323}.
\newblock arXiv:2006.16248.

\bibitem[von Burg et~al.(2020)von Burg, Low, H\"{a}ner, Steiger, Reiher,
  Roetteler, and Troyer]{BLHSRRT20}
Vera von Burg, Guang~Hao Low, Thomas H\"{a}ner, Damian~S. Steiger, Markus
  Reiher, Martin Roetteler, and Matthias Troyer.
\newblock Quantum computing enhanced computational catalysis, 2020.
\newblock arXiv:2007.14460.

\bibitem[Wan and Kim(2020)]{WK20}
Kianna Wan and Isaac Kim.
\newblock Fast digital methods for adiabatic state preparation, 2020.
\newblock arXiv:2004.04164.

\bibitem[Wecker et~al.(2014)Wecker, Bauer, Clark, Hastings, and
  Troyer]{WBCHT14}
Dave Wecker, Bela Bauer, Bryan~K. Clark, Matthew~B. Hastings, and Matthias
  Troyer.
\newblock Gate count estimates for performing quantum chemistry on small
  quantum computers.
\newblock \emph{Physical Review A}, 90:\penalty0 022305, Aug 2014.
\newblock \doi{10.1103/PhysRevA.90.022305}.
\newblock arXiv:1312.1695.

\bibitem[Wecker et~al.(2015)Wecker, Hastings, Wiebe, Clark, Nayak, and
  Troyer]{wecker2015solving}
Dave Wecker, Matthew~B Hastings, Nathan Wiebe, Bryan~K Clark, Chetan Nayak, and
  Matthias Troyer.
\newblock Solving strongly correlated electron models on a quantum computer.
\newblock \emph{Physical Review A}, 92\penalty0 (6):\penalty0 062318, 2015.
\newblock \doi{10.1103/PhysRevA.92.062318}.
\newblock arXiv:1506.05135.

\bibitem[Whitfield et~al.(2011)Whitfield, Biamonte, and Aspuru-Guzik]{WBA11}
James~D. Whitfield, Jacob Biamonte, and Al{\'a}n Aspuru-Guzik.
\newblock Simulation of electronic structure hamiltonians using quantum
  computers.
\newblock \emph{Molecular Physics}, 109\penalty0 (5):\penalty0 735--750, 2011.
\newblock \doi{10.1080/00268976.2011.552441}.
\newblock arXiv:1001.3855.

\bibitem[Wick(1950)]{Wick50}
Gian~Carlo Wick.
\newblock The evaluation of the collision matrix.
\newblock \emph{Physical Review}, 80:\penalty0 268--272, Oct 1950.
\newblock \doi{10.1103/PhysRev.80.268}.

\bibitem[Xu et~al.(2020)Xu, Susskind, Su, and Swingle]{XSSS20}
Shenglong Xu, Leonard Susskind, Yuan Su, and Brian Swingle.
\newblock A sparse model of quantum holography, 2020.
\newblock arXiv:2008.02303.

\bibitem[Zheng et~al.(2017)Zheng, Chung, Corboz, Ehlers, Qin, Noack, Shi,
  White, Zhang, and Chan]{zheng2017stripe}
Bo-Xiao Zheng, Chia-Min Chung, Philippe Corboz, Georg Ehlers, Ming-Pu Qin,
  Reinhard~M Noack, Hao Shi, Steven~R White, Shiwei Zhang, and Garnet Kin-Lic
  Chan.
\newblock Stripe order in the underdoped region of the two-dimensional
  {H}ubbard model.
\newblock \emph{Science}, 358\penalty0 (6367):\penalty0 1155--1160, 2017.
\newblock \doi{10.1126/science.aam7127}.
\newblock arXiv:1701.00054.

\end{thebibliography}

\end{document}